\documentclass{amsart}

\usepackage[top=1in, bottom=1in, left=1in, right=1in]{geometry}

\usepackage[english]{babel}
\usepackage{color}
\usepackage{graphicx}
\usepackage{framed}
\usepackage[normalem]{ulem}
\usepackage{mathtools}
\usepackage{amsmath}
\usepackage{amsthm}
\usepackage{amssymb}
\usepackage{amsfonts}
\usepackage{a4wide}
\usepackage{enumerate}
\usepackage{appendix}
\usepackage{stmaryrd}
\usepackage{hyperref}
\usepackage{url}
\usepackage{latexsym}
\usepackage{indentfirst}
\usepackage{placeins}
\usepackage{booktabs}
\usepackage{algorithm}
\usepackage[noend]{algpseudocode}
\usepackage{array, makecell}
\usepackage{color,xcolor,colortbl}
\usepackage{subfig}
\usepackage{graphicx,color} 
\usepackage{diagbox}
\usepackage[capitalize]{cleveref}
\usepackage{braket}
\usepackage{nicematrix}
\usepackage{qcircuit}

\theoremstyle{plain}
\newtheorem{theorem}{Theorem}[section]
\newtheorem{lemma}[theorem]{Lemma}

\newtheorem{cor}[theorem]{Corollary}

\newtheorem{problem}{Problem}
\newtheorem{task}{Task}

\theoremstyle{definition}

\theoremstyle{remark}

\DeclareMathOperator{\polylog}{polylog}
\DeclareMathOperator{\poly}{poly}

\DeclareMathOperator{\Arg}{Arg}

\newcommand{\ol}[1]{\ensuremath{\overline{#1}}}

\newcommand{\eps}{\epsilon}

\newcommand{\abs}[1]{\left|#1\right|}

\newcommand{\norm}[1]{\left\lVert#1\right\rVert}

\newcommand{\bbm}{\begin{bmatrix}}
\newcommand{\ebm}{\end{bmatrix}}
\newcommand{\SU}[1]{\ensuremath{\text{SU(#1)}}}
\newcommand{\RR}{\mathbb{R}}
\newcommand{\CC}{\mathbb{C}}
\newcommand{\DD}{\mathbb{D}}
\newcommand{\NN}{\mathbb{N}}
\newcommand{\TT}{\mathbb{T}}
\newcommand{\ZZ}{\mathbb{Z}}

\newcommand{\cDD}{\overline{\mathbb{D}}}

\newcommand{\lzero}{\ell_0 (\ZZ, \CC)}

\newcommand{\ceil}[1]{\left\lceil #1\right\rceil}

\renewcommand{\Re}{\mathrm{Re}}
\renewcommand{\Im}{\mathrm{Im}}

\newcommand{\diag}{\mathrm{diag}}

\newcommand{\half}{\frac{1}{2}}
\newcommand{\mo}[1]{\mathcal{O}\left(#1\right)}

\newcommand{\ba}{\mathbf{a}}
\newcommand{\bb}{\mathbf{b}}

\newcommand{\be}{\mathbf{e}}

\newcommand{\bq}{\mathbf{q}}

\newcommand{\bs}{\mathbf{s}}

\newcommand{\bu}{\mathbf{u}}
\newcommand{\bv}{\mathbf{v}}

\newcommand{\rev}{\mathrm{rev}}

\newcommand{\ema}{\epsilon_{\mathrm{ma}}}

\newcommand{\ha}{\hat{a}}
\newcommand{\hb}{\hat{b}}
\newcommand{\hga}{\hat{\gamma}}
\newcommand{\bga}{\boldsymbol{\gamma}}

\newcommand{\hba}{\hat{\mathbf{a}}}
\newcommand{\hbb}{\hat{\mathbf{b}}}
\newcommand{\hbga}{\hat{\boldsymbol{\gamma}}}

\newcommand{\mk}{\mathcal{K}}

\newcommand{\lup}{L_{\mathrm{up}}}
\newcommand{\ldown}{L_{\mathrm{down}}}
\newcommand{\kdown}{K_{\mathrm{down}}}

\newcommand{\ddown}{D_{\mathrm{down}}}
\newcommand{\yup}{\mathbf{b}_{\mathrm{up},0}}
\newcommand{\ydown}{\mathbf{b}_{\mathrm{down},0}}
\newcommand{\hbam}{\hat{\ba}_m}
\newcommand{\hbbm}{\hat{\bb}_m}
\newcommand{\hgm}{\hat{G}_m}
\newcommand{\gup}{\boldsymbol{\gamma}_{\mathrm{up}}}
\newcommand{\gdo}{\boldsymbol{\gamma}_{\mathrm{down}}}
\newcommand{\hgup}{\hat{\boldsymbol{\gamma}}_{\mathrm{up}}}
\newcommand{\hgdo}{\hat{\boldsymbol{\gamma}}_{\mathrm{down}}}
\newcommand{\vxi}{\boldsymbol{\xi}}
\newcommand{\veta}{\boldsymbol{\eta}}
\newcommand{\txi}{\tilde{\boldsymbol{\xi}}}
\newcommand{\teta}{\tilde{\boldsymbol{\eta}}}
\newcommand{\cxi}{\check{\boldsymbol{\xi}}}
\newcommand{\ceta}{\check{\boldsymbol{\eta}}}
\newcommand{\fft}{\mathrm{FFT}}
\newcommand{\ifft}{\mathrm{IFFT}}
\newcommand{\hai}{\hat{\boldsymbol{\xi}}}
\newcommand{\heta}{\hat{\boldsymbol{\eta}}}

\newcommand{\Or}{\mathcal{O}}

\newcommand{\ao}{a_{\mathrm{o}}}
\newcommand{\ano}{a_{\mathrm{no}}}

\begin{document}

\title[Inverse Nonlinear Fast Fourier Transform]{Inverse nonlinear fast Fourier transform on SU(2) with applications to quantum signal processing}

\author{Hongkang Ni}
\address{Institute for Computational and Mathematical Engineering, Stanford University, Stanford, CA}
\curraddr{}
\email{hongkang@stanford.edu}
\thanks{}

\author{Rahul Sarkar}
\address{Department of Mathematics, University of California, Berkeley, CA}
\curraddr{}
\email{rsarkar@berkeley.edu}
\thanks{}

\author{Lexing Ying}
\address{Department of Mathematics, Stanford University, CA}
\curraddr{}
\email{lexing@stanford.edu}
\thanks{}

\author{Lin Lin}
\address{Department of Mathematics, University of California, Berkeley, CA}
\address{Applied Mathematics and Computational Research Division, Lawrence Berkeley National Laboratory, Berkeley, CA}
\curraddr{}
\email{linlin@math.berkeley.edu}
\thanks{}


\date{May 18, 2025}

\dedicatory{}

\commby{}

\begin{abstract}
The nonlinear Fourier transform (NLFT) extends the classical Fourier transform by replacing addition with matrix multiplication. While the NLFT on $\mathrm{SU}(1,1)$ has been widely studied, its $\mathrm{SU}(2)$ variant has only recently attracted attention due to emerging applications in quantum signal processing (QSP) and quantum singular value transformation (QSVT). In this paper, we investigate the inverse NLFT on $\mathrm{SU}(2)$ and establish the numerical stability of the layer stripping algorithm for the first time under suitable conditions. Furthermore, we develop a fast and numerically stable algorithm, called inverse nonlinear fast Fourier transform, for performing inverse NLFT with near-linear complexity. This algorithm is applicable to computing phase factors for both QSP and the generalized QSP (GQSP).
\end{abstract}

\maketitle

\section{Introduction}
\label{sec:intro}

The Fourier transform is a fundamental tool in mathematics and is used ubiquitously in scientific and engineering computations. In comparison, the nonlinear Fourier transform (NLFT) is far less well-known despite having a long and rich history. Its origin can be traced back to Schur's 1917 study of the properties of bounded holomorphic functions on the unit disk, now known as Schur functions~\cite{schur1918potenzreihen}. Over the following century, NLFT has been rediscovered in seemingly unrelated contexts under different names, including scattering theory~\cite{beals1985inverse,winebrenner1998linear,damanik2004half,hitrik2001properties}, integrable systems~\cite{tanaka1972some,fokas1994integrability}, orthogonal polynomials~\cite{szego1939ortho,case1975orthogonal,denisov2002probability,deift2000orthogonal}, Jacobi matrices~\cite{simon2004canonical,killip2003sum,damanik2004half,volberg2002inverse}, logarithmic integrals~\cite{koosis1998logarithmic}, and stationary Gaussian processes~\cite{dym2008gaussian}, to name a few. Briefly speaking, NLFT replaces the addition operation in the linear Fourier transform with matrix multiplication. It maps a complex sequence $\bga=(\gamma_k)_{k\in\ZZ}$ to a one-parameter family of  matrices $\overbrace{\bga}(z)$, where $z$ is on the unit circle $\mathbb{T}$. Moreover, $\overbrace{\bga}(z)$ can be expressed as a product of $z$-dependent matrices, where each matrix is in a Lie group and is parameterized by an entry $\gamma_k$. In the case of Schur functions and the aforementioned applications, this Lie group is $$\mathrm{SU}(1,1):=\Set{\begin{pmatrix}
a & b\\
\overline{b} & \overline{a}
\end{pmatrix} : \abs{a}^2-\abs{b}^2=1, \quad a,b\in\CC},$$
where $\overline{a}$ denotes the complex conjugation of $a$. The transformation from $\bga$ to $\overbrace{\bga}$ is called the forward NLFT, and the mapping from $\overbrace{\bga}$ back to $\bga$ is called the inverse NLFT. We refer readers to \cite{tao2012nonlinear} for further background.

Compared to the $\mathrm{SU}(1,1)$ case, the NLFT on the Lie group $$\mathrm{SU}(2):=\Set{\begin{pmatrix}
a & b\\
-\overline{b} & \overline{a}
\end{pmatrix} : \abs{a}^2+\abs{b}^2=1, \quad a,b\in\CC}
$$
has been studied much later. This was first systematically explored in the thesis of Tsai \cite{tsai2005nlft}, which derives analytic results in the $\mathrm{SU}(2)$ setting that parallel those of $\mathrm{SU}(1,1)$. However, there are important differences between these two cases, particularly in terms of the domain,  range, and injectivity of the NLFT map (for example, compare~\cite[Theorem~2.3]{tsai2005nlft} and~\cite[Theorem~1]{tao2012nonlinear} for the relevant results for compactly supported sequences). The NLFT on $\mathrm{SU}(2)$ has applications in the study of solitons from certain nonlinear Schr\"{o}dinger equations (see e.g., \cite{faddeev1987hamiltonian},~\cite[Chapter~5]{tsai2005nlft}), but its most significant application has emerged only in the past few years from the quantum computing community, independent of the existing NLFT literature. This development is known as quantum signal processing (QSP), first introduced by Low and Chuang \cite{LowChuang2017}. QSP is defined in terms of a product of parameterized matrices in $\mathrm{SU}(2)$, and the parameters are known as phase factors. QSP was subsequently extended by Gily{\'e}n \textit{et al.}  \cite{GilyenSuLowEtAl2019} into quantum singular value transformation (QSVT), which has since been recognized as a seminal development in quantum algorithms, providing a unifying framework for many existing and new quantum algorithms \cite{MartynRossiTanEtAl2021}. 
The problem of determining phase factors was initially considered computationally challenging~\cite{ChildsMaslovNamEtAl2018}, and there have been significant algorithmic advancements in recent years on this topic~\cite{ChaoDingGilyenEtAl2020,DongMengWhaleyEtAl2021,Ying2022,WangDongLin2022,DongLinNiEtAl2024_iqsp,DongLinNiEtAl2024_newton,alexis2024quantum,alexis2024infinite,motlagh2024generalized,ni2024fast}. 

The connection between QSP and the $\mathrm{SU}(2)$ NLFT was recently established in \cite{alexis2024quantum}, which showed that determining the phase factors in QSP is equivalent to solving a variant of the inverse NLFT, a task that we refer to as the incomplete inverse NLFT (\cref{task:nlft-incomplete}). This connection, along with its generalization, is discussed in \cref{sec:connection}. Given an incomplete inverse NLFT, one first completes it to obtain a standard inverse NLFT problem and solves the latter. In this paper, we focus exclusively on the $\mathrm{SU}(2)$ version of the NLFT; any reference to the NLFT or its inverse hereon will refer to this case, unless stated otherwise.

Specifically, in this paper,  with all formal definitions appearing later in \cref{sec:prelim}, we consider a sequence $\bga: \ZZ \rightarrow \CC$ compactly supported on $[m,n] \subseteq \ZZ$, whose NLFT is $\overbrace{\bga}(z) = \left( \begin{smallmatrix}
            a(z) & b(z)\\ -b^*(z) & a^*(z)
    \end{smallmatrix} \right) = \prod_{k=m}^{n} \frac{1}{\sqrt{1 + |\gamma_k|^2}}
    \left( \begin{smallmatrix}
        1 & \gamma_k z^k \\
        - \ol{\gamma_k} z^{-k} & 1
    \end{smallmatrix} \right)$,
with $a^*(z) := \overline{a(\overline{z}^{-1})}$. Then $a(z)$ and $b(z)$ are Laurent polynomials satisfying $aa^* + bb^*=1$ and $0 < a^*(0) < \infty$, and all such pairs $(a, b)$ are elements of the set $\mathcal{S}$, which is the image of the NLFT map. Our focus will be on the development of stable and efficient algorithms for solving the inverse NLFT problem:
\begin{problem}[Inverse NLFT]
\label{prob:nlft-inverse}
    Given Laurent polynomials $a$ and $b$ such that $(a,b) \in \mathcal{S}$, determine a compactly supported sequence $\bga$ such that $\overbrace{\bga} = \left( \begin{smallmatrix}
    a & b \\ -b^\ast & a^\ast
\end{smallmatrix} \right)$.
\end{problem}

A key challenge in solving the inverse NLFT problem is its numerical stability. 
A numerically stable algorithm should be able to compute the nonzero entries of $\bga$ supported on an interval of size $n$ to within an error $\epsilon$, using floating-point arithmetic with only $\operatorname{polylog}(n/\epsilon)$ bits of precision.
One widely used approach for inverse NLFT is the layer stripping algorithm. However, its numerical stability is not guaranteed due to the potential accumulation of round-off errors. The algorithm is inherently sequential, and small inaccuracies introduced in earlier steps can propagate and amplify exponentially as the computation proceeds.
The recently developed Riemann-Hilbert factorization method~\cite{alexis2024infinite} addresses this issue by replacing the sequential process with a parallel evaluation of the entries of $\bga$. This method yields the first provably numerically stable algorithm for inverse NLFT under certain conditions and for solving the associated QSP phase factor problem.
Nevertheless, it remains an open question whether numerical stability is intrinsic to the inverse NLFT problem itself or is due to specific algorithms, such as layer stripping. Another important issue is computational efficiency. All known algorithms currently scale polynomially with $n$, and it is still unknown whether there exists a numerically stable algorithm that scales linearly in $n$, for arbitrary input Laurent polynomials $a(z)$ and $b(z)$.

\subsection{Contribution}

The contributions of this work are as follows:

\begin{enumerate}[(i)]

\item (Numerical Stability of layer stripping). Under a certain condition, namely that $a^*(z)$ is a polynomial with no zeros in the closed unit disk, we establish that the layer stripping algorithm for computing the inverse NLFT over $\mathrm{SU}(2)$ is numerically stable (\cref{thm: forward stable layer stripping}). This result significantly refines the previous analysis in \cite{Haah2019}, which only demonstrated the correctness of the algorithm while requiring $\mathcal{\Or}(n)$ bits of precision. Our findings provide a sharper characterization of the numerical behavior of layer stripping in this setting.

\item (Inverse nonlinear fast Fourier transform). The standard layer stripping algorithm has a computational complexity of $\mathcal{\Or}(n^2)$. In this work, we develop an efficient $\mathcal{O}(n \log^2 n)$ algorithm, the inverse nonlinear FFT (\cref{alg: phase factor finding}), that significantly reduces the computational cost. Moreover, we prove that this fast algorithm remains numerically stable under the same assumptions that guarantee the stability of layer stripping (\cref{thm: stab fast alg}).

\item (Unified perspective of quantum signal processing and generalized quantum signal processing). Our results have direct implications for QSP and its recent generalization called the generalized quantum signal processing (GQSP)~\cite{motlagh2024generalized}. We show that NLFT provides a unified framework for computing phase factors in both QSP and GQSP (\cref{thm:QSP-NLFT}, see also \cite[Lemmas~1,2]{alexis2024quantum}, and \cref{thm:gqsp-nlft}). Such a connection has also been recently reported in \cite{laneve2025generalized} for the GQSP case. Furthermore, leveraging our fast algorithm, we establish the first numerically stable $\widetilde{\mathcal{\Or}}(n)$\footnote{The notation $\tilde{\Or}(n)$ hides polylogarithmic factors in $n$. } method for computing phase factors efficiently.

\item (Locally Lipschitz estimates). We prove a locally Lipschitz estimate for the inverse NLFT map for compactly supported sequences (\cref{lem:inv-NLFT-modulus-continuity}). We then use it to derive a simple condition under which the inverse NLFT map is Lipschitz continuous (\cref{cor:inv-NLFT-lipschitz}). This condition relaxes an outerness condition for similar Lipschitz estimates in recent literature \cite{alexis2024infinite} when specialized to compactly supported sequences.
\end{enumerate}

We provide a comparative summary of existing algorithms for inverse NLFT in \cref{tab:inverse-nlft-methods}, where our contributions are indicated in blue color.

\begin{table}[h!]
\centering

\begin{tabular}{|c|c|c|c|}
    \hline
    Algorithm & Time complexity & \makecell{Space \\ complexity} & \makecell{Numerical \\ stability} \\
    \hline
    Riemann-Hilbert \cite{alexis2024infinite} \rule[0pt]{0pt}{3ex} & $\Or (n^4)$ & $\Or(n^2)$& \checkmark \\
    Half-Cholesky \cite{ni2024fast} & $\Or (n^2 )$ & $\Or(n)$ & \checkmark \\
    Layer stripping \cite{tsai2005nlft} & $\Or(n^2)$ & {$\Or(n)$} & \textcolor{blue}{\checkmark} \\
    \textcolor{blue}{Inverse nonlinear FFT} & \textcolor{blue}{$\Or(n \log^2 n)$} & \textcolor{blue}{$\Or(n)$} & \textcolor{blue}{\checkmark} \\
    \hline
\end{tabular}
\vspace{0.1cm}
\caption{Comparison of methods for computing the inverse NLFT of $(a,b) \in \mathcal{S}_\eta$ (see \cref{eq:S_eta-def}) for compactly supported sequences of length $n$. Of these methods, the Riemann-Hilbert method is the only one capable of finding all the phase factors in parallel. All the algorithms assume that $a^\ast$ has no zeros in the closed unit disk. The items marked in \textcolor{blue}{blue} color are the contributions of this work. }
\label{tab:inverse-nlft-methods}
\end{table}

\subsection{Related works}
\label{ssec:related-works}

The previously known layer stripping algorithm \cite{tsai2005nlft} for solving the inverse NLFT problem lacked a guarantee of numerical stability, prompting the development of two new approaches. These are the Riemann-Hilbert factorization method \cite{alexis2024infinite} and its optimized variant, the Half-Cholesky algorithm \cite{ni2024fast}, which have runtime complexities of $\Or(n^4)$ and $\Or(n^2)$, respectively, for a compactly supported sequence of length $n$. 
The Riemann-Hilbert method computes each coefficient of the inverse NLFT sequence independently by solving a separate linear system; as a result, the numerical errors in computing each coefficient are independent, a property that contributes to its numerical stability. However, solving a separate linear system for each coefficient introduces considerable redundancy. The Half-Cholesky method eliminates this inefficiency by leveraging the displacement structure of the problem, while crucially maintaining numerical stability. 

The central topic of this paper is to investigate stable and efficient algorithms for inverse NLFT without Riemann-Hilbert factorization. We show that the original layer stripping algorithm is also numerically stable, achieving the same $\Or(n^2)$ runtime as the Half-Cholesky method, while offering a smaller pre-constant. Beyond that, the inverse nonlinear fast Fourier transform that we introduce improves the runtime to $\Or(n \log^2 n)$, while keeping the space complexity $\Or(n)$ unchanged. This fast algorithm is directly adapted from a similar algorithm underlying the superfast Toeplitz solver \cite{ammar1989numerical} in the  ${\rm SU}(1,1)$ case. Interestingly, all four algorithms share the same conditions under which numerical stability is guaranteed, and their key features are tabulated in \cref{tab:inverse-nlft-methods}.

The displacement low-rank structure plays a crucial role in structured matrix theory and has a surprising connection to inverse NLFT algorithms. 
Extensive research has been conducted on fast triangular decomposition methods for such matrices \cite{Kailath1995displacement, Poloni2010a, Stewart1997cholesky, Gohberg1995fast}, as well as on other efficient techniques for solving related linear systems  \cite{di1993cg, Ammar1988superfast, Xia2012a, Pan2000superfast, Chandrasekaran2008a, Chan1996conjugate, Xi2014superfast}. Stability concerns have been raised for these fast algorithms, and several methods have been proposed to mitigate these issues \cite{ammar1989numerical, VanBarel2001a, Stewart2003a, Gohberg1995fast}. In particular, stability results have been established in the special case where the matrix is a positive definite Toeplitz-like matrix \cite{Chandrasekaran1996stabilizing, Stewart1997stability}. The triangular decomposition of displacement low-rank matrices is closely related to the layer stripping process in NLFT. We leverage this connection to establish stability results for both the layer stripping method and the inverse nonlinear FFT algorithm (see \cref{sec:stability}).

Finding QSP phase factors for a target polynomial is an important application of inverse NLFT algorithms. We briefly review the relevant literature here and refer the reader to \cite[Table 1]{ni2024fast} for a more detailed comparison of the complexity and stability of existing methods. Early algorithms can be grouped into two broad categories: \textit{direct methods} and \textit{iterative methods}.
The direct methods first find a complementary polynomial, which is analogous to the incomplete inverse NLFT task mentioned above, and then perform layer stripping to retrieve phase factors. \cite{GilyenSuLowEtAl2019, Haah2019} determine the complementary polynomial via root-finding, which is numerically unstable and requires extended precision arithmetic to compute the roots accurately. Various improvements to the direct method have been carried out in ~\cite{Ying2022,ChaoDingGilyenEtAl2020}, but the numerical stability of these approaches has not yet been proved. For the GQSP problem, a similar complementary polynomial step is handled using an optimization-based approach~\cite{laneve2025generalized}, but without guaranteed convergence. 
In contrast, the iterative methods \cite{DongMengWhaleyEtAl2021,DongLinNiEtAl2024_iqsp,DongLinNiEtAl2024_newton} find the QSP phase factors by solving an optimization problem or a nonlinear system of equations, and are numerically stable. Despite the complex energy landscape~\cite{WangDongLin2022}, the efficiency of these methods can be proved for a subset of target functions~\cite{WangDongLin2022,DongLinNiEtAl2024_iqsp}. To our knowledge, current theoretical tools cannot guarantee the convergence of such iterative methods for all functions admissible under the QSP framework. Recently developed methods based on NLFT techniques~\cite{alexis2024quantum, alexis2024infinite, ni2024fast} can also be classified as direct methods. In these approaches, the complementary polynomial is obtained using the Weiss algorithm introduced in~\cite{alexis2024infinite}, and the phase factors are subsequently recovered via the inverse NLFT algorithms summarized in \cref{tab:inverse-nlft-methods}.

\subsection{Outline}
\label{ssec:outline}

The remaining sections are structured as follows. 
\begin{enumerate}[(i)]
    \item In \cref{ssec:nlft-prelim} we introduce the NLFT and its various properties in a self-contained way for compactly supported input sequences. The concept of complementary polynomials is discussed in \cref{ssec:complementarity-prelim}, where we also introduce the notion of an outer polynomial. We introduce the inverse NLFT problem in \cref{ssec:inv-nlft-prelim}.
    \item In \cref{sec:connection}, we discuss the QSP and GQSP problems in \cref{ssec:qsp,ssec:gqsp} respectively, and establish their correspondence with the NLFT problem in \cref{ssec:qsp-nlft-correspondence}.
    \item Next, in \cref{ssec:layer-stripping,ssec:fast-inverse-NLFT}, we introduce the layer stripping and inverse nonlinear FFT algorithms for computing the inverse NLFT. 
    \item \cref{ssec:floating-point-background} contains a brief summary of some key notions involved in the stability analysis of numerical algorithms. In \cref{ssec:numerical-instability} we provide some numerical evidence of instability encountered in computing the inverse NLFT using the layer stripping algorithm. The rest of \cref{sec:stability} is devoted to proving the numerical stability for the layer stripping and the inverse nonlinear FFT algorithms, under the central assumption of this paper that the polynomial $a^\ast$ is outer, essentially meaning  
    no zeros in the closed unit disk.
    \item Finally, \cref{sec:lipschitz-bounds-nlft} is devoted to establishing some Lipschitz and locally Lipschitz estimates for the NLFT and inverse NLFT maps. In particular, we prove that, while the inverse NLFT map is not uniformly continuous, it is still locally Lipschitz continuous, and the form of the estimate suggests a simple condition when the map is also Lipschitz continuous. 
\end{enumerate}

\subsection{Notation}
\label{ssec:notation}

Throughout this paper, we will denote the unit circle as $\TT$. The Riemann sphere will be denoted as $\CC \cup \{\infty\}$. We define $\CC^\ast := \CC \setminus \{0\}$. The open unit disk will be denoted by $\DD$, and the closed unit disk by $\cDD$. For a Laurent polynomial $a(z)$, define $a^*(z) := \overline{a(\overline{z}^{-1})}$ for $z\in\CC \cup \{\infty\}$, where the $\overline{\cdot}$ means the complex conjugate. For a matrix $A \in \CC^{n \times m}$, we will denote its Hermitian conjugate (i.e., its adjoint) as $A^\ast$. This should not be confused with the conjugate of a Laurent polynomial since they are always denoted by lower-case letters. The natural logarithm will be denoted as $\log$.

Wherever encountered, Laurent polynomials and polynomials will always be defined for a single variable over the field $\CC$, and they will always have finite degree. The only exception where we specifically need real polynomials happens in \cref{sec:connection}, in the discussion of \cref{task:qsp}. Vectors, which will always be elements of finite-dimensional vector spaces such as $\CC^n$, will be denoted using boldface letters --- for example, $\ba$. We will use the notation $\ba^\ast$ to denote a row vector that is the Hermitian conjugate of $\ba$, while $\overline{\ba}$ denotes the complex conjugate. The $\ell_p$ norm ($p \geq 1$) of a vector $\ba \in \CC^n$ is defined as $\norm{\ba}_p := \left( \sum_{k=0}^{n-1} |a_k|^p \right)^{1/p}$, and we define its infinity norm by $\norm{\ba}_{\infty} := \max_{k} |a_k|$. For a matrix $A \in \CC^{n \times m}$, we define its $p$-norm (for $p \ge 1$) $\norm{A}_p$ and Frobenius norm $\norm{A}_F$ as
\begin{equation}
    \norm{A}_p := \sup_{\norm{x}_p = 1} \norm{Ax}_p, \;\; \norm{A}_F := \left( \sum\limits_{i=0}^{n} \sum\limits_{j=0}^{m} |A_{ij}|^2 \right)^{1/2}.
\end{equation}
We also use $\norm{\cdot}$ without a subscript to denote the $\norm{\cdot}_2$, unless specified otherwise. For a function $f$ defined on $\TT$, we will use the notation $\int_{\TT} f := (2\pi)^{-1} \int_{0}^{2\pi} f(e^{i\theta}) \; d\theta$. For $p \ge 1$, we will use $L^p(\TT)$ to denote the Banach space of measurable functions $f$ on $\TT$  such that $\norm{f}_{L^p(\TT)} := \left( \int_{\TT} |f|^p \right)^{1/p} < \infty$.

A \textit{multiset}\footnote{We will only need finite multisets in this paper.} can contain repeated elements (unlike a set), and the number of times an element appears in a multiset is called its \textit{multiplicity}. Two multisets $A$ and $B$ are equal if and only if the elements in them are equal when counted with their multiplicities, while we say $B \subseteq A$ if and only if the multiplicity of every element in $B$ is at most its multiplicity in $A$. If $B \subseteq A$, we define $A \setminus B$ to be the unique multiset so that $(A \setminus B) \cup B = A$.

Other definitions and notations will be introduced as needed in the subsequent sections.

\subsection*{Acknowledgments}\label{ssec:acknowledgment}
 This material is based upon work supported by the U.S. Department of Energy, Office of Science, Accelerated Research in Quantum Computing Centers, Quantum Utility through Advanced Computational Quantum Algorithms, grant no. DE-SC0025572 (L.L., R.S., L.Y.). Support is also acknowledged from the U.S. Department of Energy, Office of Science, National Quantum Information Science Research Centers, Quantum Systems Accelerator (L.L., R.S.). L.L. is a Simons Investigator in Mathematics. L.L., H.N., and R.S. thank the Mathematisches Forschungsinstitut Oberwolfach for its hospitality during the program \textit{Arbeitsgemeinschaft: Quantum Signal Processing and Nonlinear Fourier Analysis} in October 2024, where part of this work was initiated, and Christoph Thiele for highlighting the connection between Schur's algorithm and the layer stripping process.

\section{Nonlinear Fourier transform and its structure}
\label{sec:prelim}

\subsection{NLFT}
\label{ssec:nlft-prelim}
Let $\bga: \ZZ \rightarrow \CC$ be a compactly supported sequence, and we will denote the space of all such sequences by $\lzero$. The $k^{\text{th}}$ component of $\bga$ will be denoted as $\gamma_k$. For a $\bga \in \lzero$, whose support lies in $[m,n]$, where $m,n \in \ZZ$, the \textit{linear Fourier series} of $\bga$ is defined as
\begin{equation}
\label{eq:linear-Fourier-series}
G(z) := \sum_{k = m}^{n} \gamma_k z^k. 
\end{equation}

From the same sequence $\bga$, the \textit{nonlinear Fourier series} can be defined by considering a finite product of matrix-valued meromorphic functions on the Riemann sphere as follows:
\begin{equation}
\label{eq:nonlinear-Fourier-series}
    \overbrace{\bga}(z) := \prod_{k=m}^{n} \left[\frac{1}{\sqrt{1 + |\gamma_k|^2}} 
    \begin{pmatrix}
        1 & \gamma_k z^k \\
        - \ol{\gamma_k} z^{-k} & 1
    \end{pmatrix}\right].
\end{equation}

The expression in \cref{eq:nonlinear-Fourier-series} is also called the $\text{SU}(2)$ \textit{nonlinear Fourier transform}\footnote{With a slight abuse of terminology, we adopt the convention commonly used in the literature and use the terms \emph{nonlinear Fourier series} and \emph{nonlinear Fourier transform}  interchangeably. In this paper, NLFT always refers to the $\text{SU}(2)$ version.} of $\bga$. Taking the determinant of the matrix factors $(1 + |\gamma_k|^2)^{-1/2} \left(
\begin{smallmatrix}
        1 & \gamma_k z^k \\
        - \ol{\gamma_k} z^{-k} & 1
\end{smallmatrix} \right)$ appearing in \cref{eq:nonlinear-Fourier-series}, we see that the determinant of each factor in the product and also of $\overbrace{\bga}(z)$ is $1$ everywhere on the Riemann sphere, by analytic continuation. Moreover, the matrix factors are elements of $\text{SU}(2)$ when $z \in \TT$, and thus so is $\overbrace{\bga}(z)$.

Note that when $\delta=\sup_{k\in\ZZ}\abs{\gamma_k}$ is small, to the leading order in $\delta$, we have 
\begin{equation}
\overbrace{\bga}(z) \approx\begin{pmatrix}
    1 & \sum_{k=m}^{n} \gamma_k z^k \\
    -\sum_{k=m}^{n} \overline{\gamma_k} z^{-k} & 1
\end{pmatrix}.
\end{equation}
Therefore, the standard Fourier series can be viewed as the leading order contribution to the upper-right entry of $\overbrace{\bga}(z)$. When $\delta$ is large, the difference between the two quantities can become significant. 

Note that the defintion $a^*(z) := \overline{a(\overline{z}^{-1})}$  implies $(a^\ast)^\ast = a$, and $(ab)^\ast = a^\ast b^\ast$. Also note for instance that if $a(z)$ is a finite series of the form $a(z) = \sum_{k=0}^{n} \alpha_k z^{\beta_k}$, where $\alpha_k \in \CC$ and $\beta_k \in \ZZ$, for each $k$, then $a^*(z) = \sum_{k=0}^{n} \overline{\alpha_k} \left( \frac{1}{z} \right)^{\beta_k}$. Then we can show that a consequence of \cref{eq:nonlinear-Fourier-series} is that $\overbrace{\bga}(z)$ is always of the form
\begin{equation}
\label{eq:nlft-ab-def}
    \overbrace{\bga}(z) = 
    \begin{pmatrix}
            a(z) & b(z)\\ -b^*(z) & a^*(z)
    \end{pmatrix},
\end{equation}
where $a(z)$, and $b(z)$ are Laurent polynomials. This is because if we have two meromorphic matrix-valued functions $\left( \begin{smallmatrix}
    a(z) & b(z) \\ -b^\ast(z) & a^\ast(z)
\end{smallmatrix} \right)$ and $\left( \begin{smallmatrix}
    \tilde{a}(z) & \tilde{b}(z) \\ -\tilde{b}^\ast(z) & \tilde{a}^\ast(z)
\end{smallmatrix} \right)$ on the Riemann sphere (assume they have an identical set of poles), then their matrix product is also a matrix-valued function with the same set of poles and has the form $\left( \begin{smallmatrix}
    c(z) & d(z) \\ -d^\ast(z) & c^\ast(z)
\end{smallmatrix} \right)$, where $c = a \tilde{a} - b \tilde{b}^\ast$ and $d = a \tilde{b} + b \tilde{a}^\ast$, and note that every matrix factor in the product in \cref{eq:nonlinear-Fourier-series} is of this form, therefore implying \cref{eq:nlft-ab-def}. Moreover, the condition $\det \left( \overbrace{\bga}(z) \right) = 1$ implies that we have the following equality for all $z$ on the Riemann sphere:
\begin{equation}
\label{eq:det-eq-1}
a(z) a^\ast(z) + b(z) b^\ast(z) = 1.
\end{equation}
Borrowing notation from \cite[Section~3]{alexis2024quantum}, we will sometimes omit the second row of the matrix and denote (with a slight abuse of notation) $\overbrace{\bga}:= (a, b)$. Whether $\overbrace{\bga}$ denotes $(a, b)$ or the matrix in \cref{eq:nonlinear-Fourier-series} will be clear from the context.

We list two important facts of the NLFT for compactly supported sequences in the next two lemmas (cf. \cite[Section~3]{alexis2024quantum}, \cite[Chapter~2]{tsai2005nlft}). For proofs of these lemmas, the reader is referred to \cite[Lemma~2.1, Theorem~2.3]{tsai2005nlft}.
\begin{lemma}[NLFT structure]
\label{lem:nlft-ab-degree}
Let $\ell(m,n)$ be the space of all compactly supported sequences $\bga: \ZZ \rightarrow \CC$ supported on the interval $m \leq k \leq n$, in the strict sense that $\gamma_m, \gamma_n \neq 0$. Let $\bga \in \ell(m,n)$ and $\overbrace{\bga} = (a,b)$. Then the Laurent polynomials $a$ and $b$ satisfy  \cref{eq:det-eq-1}. The lowest and highest degrees of $b$ are $m$ and $n$, respectively, while the lowest and highest degrees of $a$ are $m-n$ and $0$, respectively. Moreover we have $0 < a^\ast(0) = \prod_{k = m}^{n} (1 + |\gamma_k|^2)^{-1/2}$.
\end{lemma}
\begin{lemma}[NLFT bijection]
\label{lem:nlft-bijection}
The NLFT is a bijection from $\lzero$ onto the space
\begin{equation}
\label{eq:nlft-image}
\mathcal{S} = \{(a,b): a,b \text{ are Laurent polynomials}, \; aa^\ast + bb^\ast = 1, \; 0 < a^\ast(0) < \infty \}.
\end{equation}
\end{lemma}

One of the important aspects of \cref{lem:nlft-ab-degree} is that the lowest and highest degrees of $b$ immediately indicate the support of the sequence $\bga$. It also puts a constraint on the highest and lowest degrees of $a$. Note that the conditions specifying the image of NLFT in \cref{lem:nlft-bijection}, namely (i) $aa^\ast + bb^\ast = 1$, and (ii) $0 < a^\ast(0) < \infty$, also automatically impose these constraints on the highest and lowest degrees of $a$ and $b$. Moreover, the degree conditions on $a$ in \cref{lem:nlft-ab-degree} imply that $a^\ast$ is a polynomial of lowest degree $0$ and highest degree $n-m$. Condition (i) then implies that $|a(z)|, |b(z)| \leq 1$ for all $z \in \TT$. 

\subsection{Complementary polynomials}
\label{ssec:complementarity-prelim}

Let us define the projection maps onto the first and second factors of $\mathcal{S}$, and denote the resulting projections as $\mathcal{A}$ and $\mathcal{B}$ respectively:
\begin{equation}
    \pi_1 : \mathcal{S} \ni (a,b) \mapsto a \in \mathcal{A}, \;\; \pi_2 : \mathcal{S} \ni (a,b) \mapsto b \in \mathcal{B}. 
\end{equation}
The questions of membership in the sets $\mathcal{A}$ and $\mathcal{B}$ have an answer based on the well-known Fej\'{e}r-Riesz theorem for Laurent polynomials \cite{dritschel2010operator}: a Laurent polynomial $f$ that is real on $\TT$ satisfies $f = f^\ast$, and if $f \geq 0$ on $\TT$ then it can be factorized as $f = gg^\ast$ for some Laurent polynomial $g$. Thus $b \in \mathcal{B}$ if and only of $1 - bb^\ast \ge 0$ on $\TT$. Similarly, $a \in \mathcal{A}$ if and only if $1-aa^\ast \ge 0$  on $\TT$ and $a^*(0) > 0$.

For $a,b$ such that $(a,b) \in \mathcal{S}$ in \cref{lem:nlft-bijection}, we say that $a$ and $b$ are \textit{complementary polynomials}. Note that if $(a,b) \in \mathcal{S}$, then it implies that $(a, z^k b) \in \mathcal{S}$ and $(z^k a, b) \not \in \mathcal{S}$, for every $k \in \NN$. Thus, a natural question to ask is how much freedom we have in choosing $a$ or $b$, given only one of them. This question has a clear answer, which we state in \cref{app:complementary-poly}.
A precise characterization of the preimage over each point in $\mathcal{A}$ and $\mathcal{B}$ is given in \cref{lem:preimage-a,lem:preimage-b}. An implication of these two lemmas is worth noticing, which we use later in \cref{ssec:numerical-instability}: if $(a,b) \in \mathcal{S}$ and $a^\ast$ has degree $n$, then there exists a unique $\lambda \in \CC^\ast$ such that $(\lambda z^{-n} a^\ast,b) \in \mathcal{S}$. \cref{lem:preimage-b} should be compared to \cite[Lemma~2.4]{tsai2005nlft}, which essentially follows from \cref{lem:preimage-b}. In fact, the roles of $a$ and $b$ can be interchanged to some extent, as stated in \cref{lem:preimage-a}.

A polynomial $f$ is called an \textit{outer polynomial} if and only if it has no zeros in $\DD$. A more general definition of outer functions beyond polynomials is given in \cref{sec:outer-functions}. The discussion of outer polynomials arises in NLFT in the context of complementary polynomials. If $b \in \mathcal{B}$, then a corollary of \cref{lem:preimage-b} is that there exists a unique $a \in \mathcal{A}$ such that $(a,b) \in \mathcal{S}$ and $a^\ast$ is outer, a result that also appears previously in \cite[Theorem~4]{tsai2005nlft} and \cite{alexis2024infinite}. The Weiss algorithm introduced in \cite{alexis2024infinite} provides an efficient and stable numerical procedure for computing the outer complementary polynomial $a^*$ given $b$. The outerness condition plays a crucial role in ensuring the stability of inverse NLFT algorithms, as further discussed in \cref{sec:stability}.

\subsection{Inverse NLFT}
\label{ssec:inv-nlft-prelim}
The bijective property of NLFT, as given by \cref{lem:nlft-bijection}, makes it possible to define the problem of computing the \textit{inverse nonlinear Fourier transform}  of a compactly supported sequence $\bga$. This leads to the statement of \cref{prob:nlft-inverse} that we have already introduced previously in \cref{sec:intro}. The solution $\bga$ to \cref{prob:nlft-inverse} is called the inverse NLFT of $(a,b)$. By \cref{lem:nlft-ab-degree}, determining the support of  $\bga$ is trivial---it follows directly from the lowest and highest degrees of $b$; let these degrees be $m$ and $n$ respectively. It is convenient to change \cref{prob:nlft-inverse} to one where the left endpoint of the support interval of $\bga$ is $0$. For this, we use the fact from \cite[Theorem~2]{alexis2024quantum}, that if  $\bga, \bga' \in \lzero$ such that $\gamma_{k-1} = \gamma'_k$ for all $k \in \ZZ$, then their NLFTs $\overbrace{\bga} = (a, b)$ and $\overbrace{\bga'} = (a',b')$ satisfy 
\begin{equation}
\label{eq: shift by 1}
    a = a'\text{, and } zb = b'. 
\end{equation}
Using this fact repeatedly, we can change \cref{prob:nlft-inverse} as follows: instead of determining the inverse NLFT of $(a,b)$, let us determine the inverse NLFT of $(a',b')$, where $a'=a$ and $b' = z^{-m} b$. Note that properties (i) and (ii) in \cref{prob:nlft-inverse} are also satisfied by the pair $(a',b')$. Thus, without loss of generality, we will assume in the rest of the paper that $b$ is a polynomial when talking about the inverse NLFT of $(a,b)$ in \cref{prob:nlft-inverse}. Our focus in this paper will be on two algorithms for computing the inverse NLFT in a numerically stable manner: (a) the previously known layer stripping algorithm \cite{tsai2005nlft}, and (b) a new fast algorithm that we adapt from the $\mathrm{SU}(1,1)$ case \cite{ammar1989numerical}.  We discuss both of them in \cref{ssec:layer-stripping,ssec:fast-inverse-NLFT}. 

Our stability analysis will rely on the condition that $a^*$ has no roots in $\overline{\DD}$, which is a slightly stronger condition than $a^\ast$ being outer, introduced in \cref{ssec:complementarity-prelim}. This condition is related to the function class $\mathbf{S}_{\eta}$ in \cite{alexis2024infinite,ni2024fast}. For a parameter $\eta \in (0,1)$, we define the decreasing family of subsets $\mathcal{S}_\eta \subseteq \mathcal{S}$ as
\begin{equation}
\label{eq:S_eta-def}
\mathcal{S}_{\eta} := \left\{(a,b) \in \mathcal{S}: b \text{ is a polynomial, } 0< \sup_{z \in \TT} |b(z)| \le 1-\eta \right\}.
\end{equation}
When $\sup_{z \in \TT} |b(z)|$ approaches $1$, the inverse NLFT problem becomes more and more ill-conditioned, as we will discover later in \cref{sec:stability} where $\eta^{-1}$ appears in the condition number bounds. If $(a,b) \in \mathcal{S}_\eta$ and $a^\ast$ is outer, then it must be that $a^\ast$ has no zeros in $\cDD$ (as $|a(z)|^2 + |b(z)|^2 = 1$ for all $z \in \TT$). We record this fact as a lemma below, where the forward implication is true by compactness of $\TT$:
\begin{lemma}[No zeros in $\cDD$ property]
\label{lem:no-zeros-closed-disk}
Suppose $(a,b) \in \mathcal{S}$ and $a^\ast$ is outer. Then $a^\ast$ has no zeros in $\cDD$ if and only if $(a,b) \in \mathcal{S}_{\eta}$, for some $\eta \in (0,1)$.
\end{lemma}
Our stability analysis in \cref{sec:stability} will rely precisely on this above assumption that the inverse NLFT of $(a,b) \in \mathcal{S}$ being computed satisfies the condition that $a^\ast$ has no zeros in $\overline{\mathbb{D}}$.

Both NLFT and inverse NLFT are continuous maps, as proven in \cite{tsai2005nlft}, assuming various topologies on the domain and codomain of the two maps. However, the function spaces involved there are infinite-dimensional. In the compactly supported setting, it is desirable to have simpler continuity estimates as we always work with finite-dimensional spaces, and such estimates are established in \cref{sec:lipschitz-bounds-nlft}. In particular, we show that the NLFT map is Lipschitz continuous, while the inverse NLFT map is merely locally Lipschitz continuous and not uniformly continuous, where the topologies on both the domain and codomain are inherited from the Euclidean metric.

\section{Connections between NLFT, QSP, and GQSP}
\label{sec:connection}

The connection between symmetric QSP and NLFT, both for finitely many and infinitely many phase factors, was established in previous works \cite{alexis2024quantum,alexis2024infinite}. In this section, we extend this correspondence to the case of QSP involving asymmetric phase factors \cite{GilyenSuLowEtAl2019}, as well as its generalization, GQSP \cite{motlagh2024generalized}. We use 
\[Z :=  \begin{pmatrix}
    1 & 0\\ 0 & -1
\end{pmatrix} ,\ X :=  \begin{pmatrix}
    0 & 1\\ 1 & 0
\end{pmatrix} ,\ \text{and}\ H := \frac{1}{\sqrt{2}}  \begin{pmatrix}
    1 & 1 \\1 & -1
\end{pmatrix}\]
to denote the $2 \times 2$ Pauli-$Z$, Pauli-$X$, and the Hadamard matrices respectively.

The connection between NLFT and QSP in all existing works depends on a close analog of the inverse NLFT problem (\cref{prob:nlft-inverse}), where we instead have incomplete NLFT data --- specifically, this means that we are only given one of the Laurent polynomials $a$ or $b$ such that $(a,b) \in \mathcal{S}$, while the other remains unspecified. In recent NLFT literature, such as \cite{alexis2024quantum,alexis2024infinite}, it is the second Laurent polynomial $b$ that is always specified, while the first one needs to be determined. We formalize this task below, and it will be our goal in this section to reduce both the QSP and GQSP phase factor finding tasks to it:
\begin{task}[Incomplete inverse NLFT]
\label{task:nlft-incomplete}
Given a  Laurent polynomial $b \in \mathcal{B}$, determine a complementary Laurent polynomial $a \in \mathcal{A}$ and a compactly supported sequence $\bga$ such that $\overbrace{\bga} = (a,b)$.
\end{task}
As mentioned in \cref{ssec:complementarity-prelim}, \cref{task:nlft-incomplete} can be converted to \cref{prob:nlft-inverse} with an outer $a^*$ by applying the Weiss algorithm. This places it squarely within the scope of the algorithmic and stability results developed later in this paper. As in the case of the inverse NLFT problem in \cref{ssec:inv-nlft-prelim}, in this section 
we will assume without loss of generality that $b \in \mathcal{B}$ is actually a polynomial with minimum degree zero, in the context of \cref{task:nlft-incomplete}. The general case easily reduces to this special case, just like in the case of inverse NLFT (\cref{prob:nlft-inverse}).

\subsection{Quantum signal processing}
\label{ssec:qsp}

We define the family of $\mathrm{SU}(2)$ valued matrices $W(x) := \left( \begin{smallmatrix}
    x & i \sqrt{1-x^2} \\
    i \sqrt{1-x^2} & x
\end{smallmatrix} \right)$, for all $x \in [-1,1]$. In (finite) quantum signal processing \cite{GilyenSuLowEtAl2019,low2016methodology}, we start with a sequence $\Psi:= \left \{\psi_k \in \left[-\frac{\pi}{2},\frac{\pi}{2} \right] : k=0,1,\dots,n \right\}$, and then consider the $\mathrm{SU}(2)$ valued map $U_n:  \left[-\frac{\pi}{2},\frac{\pi}{2} \right]^{n+1} \times [-1,1]$ defined as follows:
\begin{equation}
\label{eq:Ud-def}
U_n(\Psi,x) =\begin{pmatrix}
    u_n(\Psi,x) & i\, v_n(\Psi,x)\\
    i\, \overline{v_n}(\Psi,x) & \overline{u_n}(\Psi,x)
\end{pmatrix}:= e^{i \psi_0 Z} \prod_{k=1}^{n} \left(W(x) e^{i \psi_k Z}\right).
\end{equation}
The sequence $\Psi$ is called the \textit{phase factor sequence} of the QSP protocol\footnote{The reader should note that there are a few different conventions available in the literature for defining the QSP protocol. Our definition coincides with the \textit{Chebyshev QSP} protocol. A good summary of the relationships between our definition and other QSP protocols, like \textit{Laurent QSP} or \textit{Analytic QSP}, can be found in \cite{laneve2025generalized}.}, the matrix $W(x)$ is called the \textit{signal unitary}, and the $Z$-rotation matrices $\{e^{i \psi_k Z} \}_{k=0}^{n}$ are called \textit{control unitaries}. The task of finding the QSP phase factors is an inverse process of finding a $\Psi$ given partial knowledge of $U_n(\Psi,x)$ --- specifically, it refers to the following task:
\begin{task}[Determine QSP phase factors]
\label{task:qsp}
    Given an even or odd real target polynomial $f(x)\in\RR[x]$ of degree $n$, satisfying $ \|f\|_{\infty} := \max\limits_{x \in [-1, 1]} |f(x)| \le 1$, find a sequence $\Psi$ such that $\Im(u_n(\Psi,x)) = f(x)$.
\end{task}

\subsection{Generalized quantum signal processing}
\label{ssec:gqsp}
In the standard QSP protocol introduced above in \cref{eq:Ud-def}, each of the control unitaries $e^{i \psi_k Z}$ only has one degree of freedom. In a slightly different setting, the paper \cite{motlagh2024generalized} proposed a more flexible quantum algorithm called \textit{generalized quantum signal processing}, which uses a $2$-parameter family of control unitaries within the analytic QSP protocol. We will follow the same setup with a slightly different choice for the parameterized unitaries as follows:
\begin{equation}
\label{eq:R-def}
    R(\psi,\phi) := \begin{pmatrix}
        \cos\psi& e^{i\phi}\sin\psi\\ 
        -e^{-i\phi}\sin\psi & \cos\psi
    \end{pmatrix}, \;\;\; \psi,\phi \in [-\pi,\pi],
\end{equation}
and with this choice, the complete form of the GQSP protocol is given on the right-hand side of \cref{eq:gqsp-def}. The advantage of GQSP over QSP is that it lifts the parity restrictions on the target polynomial \cite[Theorem~3]{motlagh2024generalized}. Just like in the case of QSP, we may now also state the analogous phase factor finding task for GQSP:
\begin{task}[Determine GQSP phase factors]
\label{task:gqsp}
    Given a target polynomial $Q(z)\in\CC[z]$ of degree $n$, satisfying $\max\limits_{z \in \TT} |Q(z)| \le 1$, find sequences $\{\phi_k\}_{k=0}^n$ and $\{\psi_k\}_{k=0}^n$ such that
    \begin{equation}
    \label{eq:gqsp-def}
    \begin{pmatrix}
        \cdot & Q(z)\\ \cdot & \cdot
    \end{pmatrix} = R(\psi_0,\phi_0) \prod_{k=1}^n \left(\begin{pmatrix}
        z&\\&1
    \end{pmatrix}R(\psi_k,\phi_k)\right).
\end{equation}
\end{task}
The original GQSP problem in \cite{motlagh2024generalized} placed the target polynomial in the upper left corner rather than the upper right corner like us, but this does not matter since we can multiply a matrix $\left(\begin{smallmatrix}
    0 & 1\\-1& 0
\end{smallmatrix}\right)\in \SU{2}$ at the end from the right. Moreover our parameterization of $R(\psi,\phi)$ in \cref{eq:R-def} actually gives an element in $\mathrm{SU}(2)$, unlike the choice in \cite{motlagh2024generalized}. The sequences $\{\phi_k\}_{k=0}^n$ and $\{\psi_k\}_{k=0}^n$ in \cref{eq:gqsp-def} are called the \textit{GQSP phase factor sequences}.

\subsection{Correspondence between QSP, GQSP, and NLFT}
\label{ssec:qsp-nlft-correspondence}

To establish the correspondence between the QSP protocol discussed above and the incomplete inverse NLFT problem, we will need the following result whose part (b) is exactly analogous to \cite[Lemma~1]{alexis2024quantum}, and we include it here to elucidate the point that asymmetric QSP does not pose any special difficulty, as compared to the symmetric case. On the other hand, part (a) of the lemma has the nice feature that the sequence $\bga$ determined from $\Psi$ remains real when $\Psi$ is real, and it is what we will use later in \cref{thm:QSP-NLFT}.
\begin{lemma}
\label{lem:HUdH-lemma-nonsym}
For $x \in [-1,1]$, let $\theta \in [0,\pi]$ be such that $\cos \theta = x$. For $n \geq 0$, let $U_n$ be defined as in Eq.~\eqref{eq:Ud-def}, and additionally assume that $\psi_k \in (-\frac{\pi}{2}, \frac{\pi}{2})$ for all $k = 0,\dots,n$.
\begin{enumerate}[(a)]
    \item If we define the sequence $\bga: \ZZ \rightarrow \CC$ as $\gamma_k :=  \tan \psi_k$, for $k=0,\dots,n$, and zero otherwise, then for all $x \in [-1,1]$ we have
    \begin{equation}
    \label{eq:HUdH-lemma-nonsym}
    \begin{pmatrix}
        1&\\&i
    \end{pmatrix} H U_n(\Psi,\cos\theta) H \begin{pmatrix}
        1&\\&-i
    \end{pmatrix} = \overbrace{\bga}( e^{2i\theta}) 
    \begin{pmatrix}
        e^{i n \theta} & 0 \\
        0 & e^{-i n \theta}
    \end{pmatrix}.
    \end{equation}
    \item If we define the sequence $\bga: \ZZ \rightarrow \CC$ as $\gamma_k :=  i \tan \psi_k$, for $k=0,\dots,n$, and zero otherwise, then for all $x \in [-1,1]$ we have
    \begin{equation}
    \label{eq:HUdH-lemma-nonsym-1}
    H U_n(\Psi,\cos\theta) H = \overbrace{\bga}( e^{2i\theta}) 
    \begin{pmatrix}
        e^{i n \theta} & 0 \\
        0 & e^{-i n \theta}
    \end{pmatrix}.
    \end{equation}
\end{enumerate}

\end{lemma}

\begin{proof}
This proof works for both parts (a) and (b). For part (a) we denote $J:= \left( \begin{smallmatrix}
    1 & 0\\0 & i
\end{smallmatrix} \right)$, so that $J^\ast = \left( \begin{smallmatrix}
    1 & 0\\0 & -i
\end{smallmatrix} \right) = J^{-1}$, while in part (b), $J$ is simply the identity matrix.  

Let us define $\widetilde{W}(x) := H W(x) H$, and note that since $\cos \theta = x$, we may obtain 
\begin{equation}
\label{eq:W-Wtilde-relations}
    \widetilde{W}(x) =  \begin{pmatrix}
    e^{i \theta} & 0 \\ 0 & e^{-i\theta}
\end{pmatrix}, \;\; JH W(x)HJ^\ast = J \widetilde{W}(x)J^\ast = \widetilde{W}(x).
\end{equation}
Also recall that $HZH=X$, which together with the definition of $\gamma_k$, leads to
\begin{equation}
\label{eq:J-NLFT-relation}
    JHe^{i \psi_k Z}HJ^\ast = J e^{i \psi_k X} J^\ast = \frac{1}{\sqrt{1 + |\gamma_k|^2}} \begin{pmatrix}
        1 & \gamma_k \\ - \overline{\gamma_k} & 1
    \end{pmatrix}, \;\; k = 0,\dots,n.
\end{equation}
Next, the left-hand sides of \cref{eq:HUdH-lemma-nonsym,eq:HUdH-lemma-nonsym-1} equate to $JH U_n(\Psi,\cos\theta) HJ^\ast$, which we can also express in the following ordered product form using \cref{eq:W-Wtilde-relations,eq:J-NLFT-relation}:
\begin{equation}
\label{eq:J-NLFT-1}
    JH U_n(\Psi,\cos\theta) HJ^\ast = \frac{1}{\sqrt{1 + |\gamma_0|^2}} \begin{pmatrix}
        1 & \gamma_0 \\ - \overline{\gamma_0} & 1
    \end{pmatrix}
    \prod_{k=1}^{n} \left[ \widetilde{W}(x) \frac{1}{\sqrt{1 + |\gamma_k|^2}} \begin{pmatrix}
        1 & \gamma_k \\ - \overline{\gamma_k} & 1
    \end{pmatrix} \right].
\end{equation}
Finally, notice the following identity, for any $t \in \RR$:
\begin{equation}
    \begin{pmatrix}
    e^{i \theta t} & 0 \\ 0 & e^{-i\theta t}
    \end{pmatrix}
    \begin{pmatrix}
        1 & \gamma_k \\ - \overline{\gamma_k} & 1
    \end{pmatrix} 
    =
    \begin{pmatrix}
        1 & \gamma_k e^{2 i \theta t} \\ - \overline{\gamma_k} e^{-2 i \theta t} & 1
    \end{pmatrix} 
     \begin{pmatrix}
    e^{i \theta t} & 0 \\ 0 & e^{-i\theta t}
    \end{pmatrix},
\end{equation}
which allows us to simplify the right-hand side of \cref{eq:J-NLFT-1} to obtain
\begin{equation}
\label{eq:J-NLFT-2}
\begin{split}
    JH U_n(\Psi,\cos\theta) HJ^\ast &= 
    \left( \prod_{k=0}^{n} \left[  \frac{1}{\sqrt{1 + |\gamma_k|^2}} \begin{pmatrix}
        1 & \gamma_k e^{2 i k \theta} \\ - \overline{\gamma_k} e^{-2 i k \theta} & 1
    \end{pmatrix} \right] \right) 
    \begin{pmatrix}
    e^{i n \theta} & 0 \\ 0 & e^{-i n \theta}
    \end{pmatrix} \\
    &= \overbrace{\bga}(e^{2i\theta}) \begin{pmatrix}
    e^{i n \theta} & 0 \\ 0 & e^{-i n \theta}
    \end{pmatrix}.
\end{split}
\end{equation}
This finishes the proof.
\end{proof}

With the aid of \cref{eq:HUdH-lemma-nonsym}, the next theorem now shows that \cref{task:qsp} can be converted to \cref{task:nlft-incomplete}.
\begin{theorem}[QSP-NLFT correspondence]
\label{thm:QSP-NLFT}
    If $b(z)$ in \cref{task:nlft-incomplete} satisfies
    \begin{equation}
    \label{eq: b and f relation}
        \Re[b(e^{2i \theta}) e^{-i n \theta} ]  = f(\cos\theta),\ \forall \theta\in[0,\pi],
    \end{equation}
    and a sequence of real numbers $\{\gamma_k\}_{k=0}^n$ is a solution of \cref{task:nlft-incomplete}, then  a solution of \cref{task:qsp} is given by the phase factor sequence $\Psi = \{\psi_k\}_{k=0}^n$, where $\tan \psi_k = \gamma_k$ and $\psi_k \in (-\frac{\pi}{2}, \frac{\pi}{2})$ for every $k$.
\end{theorem}
\begin{proof}
    We may calculate the left-hand side of \cref{eq:HUdH-lemma-nonsym} and obtain
    \begin{equation*}
        \begin{pmatrix}
    1&\\&i
\end{pmatrix} H U_n(\Psi, \cos \theta) H \begin{pmatrix}
    1&\\&-i
\end{pmatrix} = \begin{pmatrix}
    \cdot& \Im[u_n(\Psi,\cos\theta)] -i \; \Im[v_n(\Psi,\cos\theta)]\\ \cdot &\cdot
\end{pmatrix},
    \end{equation*}
    while the right-hand side of \cref{eq:HUdH-lemma-nonsym} is
    \begin{equation*}
    \begin{pmatrix}
    \cdot& b(e^{2i \theta}) e^{-i n \theta} \\ \cdot &\cdot
    \end{pmatrix}.
    \end{equation*}
    Under the assumptions of \cref{lem:HUdH-lemma-nonsym}, when $\tan \psi_k = \gamma_k$, these two matrices are equal. By comparing the real part of the upper right element, we conclude that $\Im[u_n(\cos\theta,\Psi)] = \Re[b(e^{2i \theta}) e^{-i n \theta} ] = f(\cos\theta)$, which means that the prescribed $\Psi$ is a solution of \cref{task:qsp}.
\end{proof}

We remark that \cref{task:nlft-incomplete} itself does not guarantee that the solution $\bga$ is a real sequence. However, when $b(z)$ is a Laurent polynomial with real coefficients, the algorithms discussed in this paper always output a real sequence $\bga$. It is not hard to make $b(z)$ satisfy \cref{eq: b and f relation} and have real Laurent coefficients simultaneously. A common choice is to let $b(e^{2i \theta})  = e^{i n \theta} f(\cos\theta)$, which is adopted in most of the QSP literature, and the corresponding phase factor sequence $\Psi$ has the symmetric property \cite{WangDongLin2022}. Other possible choices could be 
$$b(e^{2i \theta})  = e^{i n \theta} (f(\cos\theta) - i \sin\theta \; g(\cos\theta))$$
for any proper polynomial $g$ as long as $b$ is a Laurent polynomial bounded by 1 on $\TT$, which is used in \cite{Ying2022}.

In the next theorem, we establish the correspondence between GQSP and NLFT:
\begin{theorem}[GQSP-NLFT correspondence]
\label{thm:gqsp-nlft}
\cref{task:gqsp} can be converted to \cref{task:nlft-incomplete} by setting $b(z) = Q(z)$, and the corresponding GQSP phase factor sequences are determined by $\psi_k = \arctan(|\gamma_k|)$ and $\phi_k = \Arg(\gamma_k)$, for $k=0,\dots,n$, where $\{\gamma_k\}_{k=0}^n$ is a solution of \cref{task:nlft-incomplete}.
\end{theorem}
\begin{proof}
    We prove this by direct calculation:
    \begin{equation}
        \begin{aligned}
            \begin{pmatrix}
        \cdot & Q(z)\\ \cdot & \cdot
    \end{pmatrix} &= \begin{pmatrix}
            a(z)&b(z)\\-b^*(z)&a^*(z)
        \end{pmatrix} \begin{pmatrix}
        z^n&\\&1
    \end{pmatrix}\\
    &= \prod_{k=0}^{n}\left[\frac{1}{\sqrt{1+|\gamma_k|^2}}  
    \begin{pmatrix}
        1 & \gamma_k z^k \\ 
        -\overline{\gamma_k}z^{-k} & 1
    \end{pmatrix}\right] \begin{pmatrix}
        z^n&\\&1
    \end{pmatrix}\\
    &= \left( \prod_{k=0}^{n}\frac{1}{\sqrt{1+|\gamma_k|^2}} \right)  \begin{pmatrix}
        1 & \gamma_0 \\ 
        -\overline{\gamma_0} & 1    \end{pmatrix} \prod_{k=1}^n\begin{pmatrix}
        z & \gamma_k z \\ 
        -\overline{\gamma_k} & 1
    \end{pmatrix}\\
    &= 
    \begin{pmatrix}
        \cos\psi_0& e^{i\phi_0}\sin\psi_0\\ 
        -e^{-i\phi_0}\sin\psi_0 & \cos\psi_0  \end{pmatrix} \prod_{k=1}^n\left[\begin{pmatrix}
            z&\\&1
        \end{pmatrix}\begin{pmatrix}
        \cos\psi_k& e^{i\phi_k}\sin\psi_k\\ 
        -e^{-i\phi_k}\sin\psi_k & \cos\psi_k    \end{pmatrix}\right],
        \end{aligned}
    \end{equation}
    where in the last step, we used $\gamma_k = e^{i \phi_k}\tan\psi_k$.
\end{proof}

\section{Layer stripping algorithm and inverse nonlinear fast Fourier transform}
\label{sec:inverse-NLFT}
In this section, we will introduce two algorithms for computing the inverse NLFT, i.e., the solution to \cref{prob:nlft-inverse}. We will first review the layer stripping algorithm in \cref{ssec:layer-stripping}. Next, in \cref{ssec:fast-inverse-NLFT} we will introduce the inverse nonlinear fast Fourier transform algorithm. 

\subsection{Layer stripping algorithm for inverse NLFT}
\label{ssec:layer-stripping}

The \textit{layer stripping} algorithm to compute the inverse NLFT in the $\mathrm{SU}(2)$ case, introduced in  \cite{tsai2005nlft}, follows the same idea as in the $\mathrm{SU}(1,1)$ case \cite{tao2012nonlinear}. Our discussion will follow the outline of the surjectivity part of the proof of \cite[Theorem~2.3]{tsai2005nlft}. For ease of presentation in this and the next subsection, it will be useful always to consider compactly supported sequences  $\bga \in \ell(0,r)$, for $r \geq 1$, but $r$ may vary. If $\bga \in \ell(0,r)$, we will use row vector notation (square brackets) to list the components of $\bga$, always starting from index zero and up to some index $r' \ge r$. For example, if we write $\bga = [\gamma_m,\dots,\gamma_n]$, then it will mean that $\gamma_m$ is the component at index zero, and $\gamma_n$ is the component at index $r'$. The NLFT of $\bga$ will be denoted $\overbrace{[\gamma_m,\dots,\gamma_n]}$. For instance, with this notation, we have
\begin{equation}
    \overbrace{[\gamma_m,\ldots, \gamma_{n}]} = \overbrace{[\gamma_m,\ldots, \gamma_{n},0,0,\ldots]} = \prod_{j=0}^{n-m} \left[\frac{1}{\sqrt{1 + |\gamma_{j+m}|^2}} 
    \begin{pmatrix}
        1 & \gamma_{j+m} z^j \\
        - \ol{\gamma_{j+m}} z^{-j} & 1
    \end{pmatrix}\right].
\end{equation}


Now we assume that we are given $(a,b) \in \mathcal{S}$, as in \cref{prob:nlft-inverse}, and that the lowest and highest degrees of $b$ are $0$ and $n-1$, respectively. Recall from \cref{lem:nlft-ab-degree} that the inverse NLFT of $(a,b)$ is an element of $\ell(0,n-1)$. Thus our task is to find the sequence $[\gamma_0,\ldots, \gamma_{n-1}]$ satisfying
\begin{equation}\label{eq: a_0 b_0}
    \begin{pmatrix}
            a(z) & b(z)\\ -b^*(z) & a^*(z)
        \end{pmatrix} = \overbrace{[\gamma_0,\gamma_1,\ldots, \gamma_{n-1}]}.
\end{equation}
The basic strategy for determining this sequence is to strip off the unitary matrices one at a time from the left, thereby reducing the problem size by one each time, and then apply the method recursively to the smaller sequence until the entire sequence is read off, which justifies the name ``layer stripping".  

The problem of determining the first component becomes finding $\gamma_0 \in \CC$ such that
\begin{equation}\label{eq: determine gamma_0}
    \frac{1}{\sqrt{1 + |\gamma_0|^2}} 
    \begin{pmatrix}
        1 & -\gamma_0 \\
        \ol{\gamma_0}  & 1
    \end{pmatrix}  \begin{pmatrix}
            a(z) & b(z)\\ -b^*(z) & a^*(z)
        \end{pmatrix} = \overbrace{[0,\gamma_1,\ldots, \gamma_{n-1}]}.
\end{equation}
If we let 
\begin{equation}\label{eq: a_1 b_1}
    \begin{pmatrix}
            a_1(z) &  b_1(z) \\ - b_1^*(z) & a_1^*(z)
        \end{pmatrix} = \overbrace{[\gamma_1,\ldots, \gamma_{n-1}]},
\end{equation} 
then from \cref{eq: shift by 1} we can rewrite \cref{eq: determine gamma_0} as
\begin{equation}
\label{eq: determine a_1 and b_1}
    \frac{1}{\sqrt{1 + |\gamma_0|^2}}  \begin{pmatrix}
            a(z) + \gamma_0 b^\ast(z) & b(z) - \gamma_0 a^*(z)\\ \ol{\gamma_0} a(z) - b^*(z) & \ol{\gamma_0} b(z) + a^*(z)
        \end{pmatrix} = \begin{pmatrix}
            a_1(z) & z b_1(z) \\ -z^{-1} b_1^*(z) & a_1^*(z)
        \end{pmatrix}.
\end{equation}
Comparing the upper right element, it is clear that the only way to make $\frac{b(z) - \gamma_0 a^*(z)}{z}$ still be a polynomial is to let $\gamma_0 = \frac{b(0)}{a^*(0)}$, which is well defined as $a^\ast(0) > 0$ by \cref{lem:nlft-ab-degree}. After determining $\gamma_0$, we can calculate $a_1(z)$ and $b_1(z)$ from \cref{eq: determine a_1 and b_1}. The remaining problem is to retrieve the rest of the sequence $\gamma_1, \gamma_2, \ldots, \gamma_{n-1}$ satisfying \cref{eq: a_1 b_1}. This is equivalent to solving \cref{eq: a_0 b_0} but with reduced size, and we may iteratively apply the same procedure to recover the remaining coefficients $\gamma_k$ one by one.

From \cref{eq: determine a_1 and b_1}, we can write the recursion formula as
\begin{equation}
\label{eq: rec a_k^* b_k}
    \gamma_k = \frac{b_k(0)}{a_k^*(0)}, \quad a_{k+1}^*(z) = \frac{a_k^*(z) + \ol{\gamma_k} b_k(z)}{\sqrt{1 + |\gamma_k|^2}},\quad b_{k+1}(z) = \frac{b_k(z) - \gamma_k a_k^*(z)}{z \sqrt{1 + |\gamma_k|^2}}.
\end{equation}
If we temporarily lift the restriction that $(a, b) \in \mathcal{S}$, then \cref{eq: rec a_k^* b_k} can also be applied iteratively to an arbitrary pair of functions $(a_0^*, b_0)$, as long as $a_0^\ast$ and $b_0$ satisfy two conditions: (i) $a_0^\ast$ and $b_0$ are both holomorphic at $0$, and (ii) $a_0^\ast(0) > 0$. This procedure then yields a sequence $\bga = (\gamma_0, \gamma_1,\ldots)$, and sequences of functions $\{a_k^\ast: k=0,1,\dots\}$ and $\{b_k: k = 0,1,\dots\}$, all of which also satisfy the same conditions as $a_0^\ast$ and $b_0$:
\begin{equation}
\label{eq:generalized-layer-stripping-conditions}
a_k^\ast, b_k \text{ are holomorphic at } 0, \;\;\; a_k^\ast(0) > 0, \;\; k \in \NN.
\end{equation}
Furthermore, from \cref{eq: rec a_k^* b_k} one checks that if $b_{k'} = 0$, then $b_k = 0$ and $a_k^\ast = a_{k'}$ for all $k \ge k'$; thus we may say that the procedure \textit{terminates} at iteration $k'$, if $k'$ is the smallest value for which $b_{k'}=0$. If $(a,b) \in \mathcal{S}$, this procedure terminates at iteration $k'$ (here $k'$ is precisely the size of the support of $\bga$) with $a_{k'}=1$, and the output sequence $\bga$ is the inverse NLFT of $(a,b)$. But when $(a,b) \not \in \mathcal{S}$, two main differences arise:
\begin{enumerate}[(i)]
\item The iteration may not terminate; that is, $\bga$ may contain infinitely many nonzero elements.
\item The NLFT $\overbrace{\bga}(z)$ of this sequence, when it exists, may not coincide with the matrix $\left( \begin{smallmatrix}
    a_0(z) & b_0(z) \\ -b_0^\ast(z) & a_0^\ast(z)
\end{smallmatrix} \right)$.
\end{enumerate}
Therefore, we may still call $\bga$ as the \textit{layer stripping sequence} of $(a_0^*, b_0)$, while $\overbrace{\bga} = (a_0, b_0)$ if and only if $(a_0,b_0)\in\mathcal{S}$.
Nonetheless, this observation proves useful in developing the inverse nonlinear FFT algorithm discussed in \cref{ssec:fast-inverse-NLFT}.

We also remark that if $c(z)$ is holomorphic at $z=0$ with $c(0) > 0$, and we apply the layer stripping procedure \cref{eq: rec a_k^* b_k} to the scaled pair $(\tilde{a}_0^*(z), \tilde{b}_0(z)) := \left(\frac{a_0^*(z)}{c(z)}, \frac{b_0(z)}{c(z)}\right)$, then by induction we have
\begin{equation}
    \tilde{\gamma}_k = \gamma_k, \quad \tilde{a}_{k+1}^*(z) = \frac{a_{k+1}^*(z)}{c(z)} ,\quad \tilde{b}_{k+1}(z) = \frac{b_{k+1}(z)}{c(z)}, \quad k=0,1,\dots.
\end{equation}
In other words, scaling $(a_0^*, b_0)$ by such a function $c(z)$ does not alter the layer stripping sequence. In particular, the pair $\left(1, \frac{b_0(z)}{a_0^*(z)}\right)$ yields the same layer stripping sequence as $(a_0^*(z), b_0(z))$. This offers another explanation for why the Half-Cholesky algorithm in \cite{ni2024fast} works.

Next, we discuss the matrix form of the layer stripping algorithm, which facilitates its numerical implementation. We define
\begin{equation}
    (a_k(z), b_k(z)) := \overbrace{[\gamma_k,\ldots, \gamma_{n-1}]},\ k=0,1,\ldots, n-1,
\end{equation}
and work on  the coefficients of $b_k(z) := \sum_{j=0}^{n-1-k} b_{j,k} z^j$ and $a_k^*(z) := \sum_{j=0}^{n-1-k} a_{j,k} z^j$.  Note that we are using the coefficients of $a_k^*(z)$ instead of $a_k(z)$ to avoid negative powers of $z$. We also point out that $(a_0,b_0) = (a,b)$. Define the column vectors $\ba_k := (a_{j,k})_{0\le j\le n-k-1}$, $\bb_k := (b_{j,k})_{0\le j\le n-k-1}$, and the $(n-k)\times 2$ matrix
\begin{equation}
\label{eq:Gk-layer-strip-def}
    G_k := (\ba_k, \bb_k).
\end{equation}
Throughout this subsection, we will also refer to $a_{n-k, k}=0$ for $k\le n-1$ and $a_{0,n} = 1$, whose choices will be evident from the discussion below. 
The $k^{\text{th}}$ step of layer stripping can be described as outputting $\gamma_k = \frac{b_{0,k}}{a_{0,k}}$, and  also obtaining the coefficients of the polynomials $a_{k+1}^*(z)$ and $z b_{k+1}(z)$ when $k < n-1$, by performing the following Givens rotation 
\begin{equation}
\label{eq:G_n recurrence}
 \begin{pmatrix}
a_{0,k+1} & 0 \\
a_{1,k+1} & b_{0,k+1} \\
\vdots & \vdots\\
a_{n-k-2,k+1} & b_{n-k-3,k+1} \\
a_{n-k-1,k+1} & b_{n-k-2,k+1} 
\end{pmatrix} =  \begin{pmatrix}
{a}_{0,k} & {b}_{0,k} \\
{a}_{1,k} & {b}_{1,k} \\
\vdots & \vdots\\
{a}_{n-k-2,k} & {b}_{n-k-2,k} \\
{a}_{n-k-1,k} & {b}_{n-k-1,k} 
\end{pmatrix} \frac{1}{\sqrt{1 + |\gamma_k|^2}} 
    \begin{pmatrix}
        1 & -\gamma_k  \\
         \ol{\gamma_k}  & 1
    \end{pmatrix},
\end{equation}
where the matrix $\frac{1}{\sqrt{1 + |\gamma_k|^2}} \left( \begin{smallmatrix}
    1 & -\gamma_k \\ \overline{\gamma_k} & 1
\end{smallmatrix} \right) \in \mathrm{SU}(2)$. When $k < n-1$, the element $a_{n-k-1,k+1}=0$, since the $a^\ast_{k+1}(z)$ given by $(a_{k+1}(z), b_{k+1}(z)) = \overbrace{[\gamma_{k+1},\ldots, \gamma_{n-1}]}$ is a polynomial of degree $n-k-2$. We then obtain 
\begin{equation}
\label{eq:G_n shift}
    G_{k+1} = (\ba_{k+1}, \bb_{k+1}) = \begin{pmatrix}
a_{0,k+1} & b_{0,k+1} \\
a_{1,k+1} & b_{1,k+1} \\
\vdots & \vdots\\
a_{n-k-2,k+1} & b_{n-k-2,k+1} 
\end{pmatrix}
\end{equation} 
by realigning the columns of the above matrix, corresponding to stripping the extra factor of $z$ in $z b_{k+1}(z)$.


Some easy properties of the layer stripping algorithm are recorded in the following lemma for later use in \cref{sec:lipschitz-bounds-nlft}:

\begin{lemma}[Layer stripping properties]
\label{lem:layer-stripping-prop}
Define $\Theta(t) :=  \frac{1}{\sqrt{1 + |t|^2}} \left( \begin{smallmatrix}
    1 & -t \\ \overline{t} & 1
\end{smallmatrix} \right) \in \mathrm{SU}(2)$, for $t \in \CC$. With this notation, the layer stripping algorithm satisfies
\begin{enumerate}[(a)]
    \item $a_{0,k} = a_{0,k-1} \sqrt{1 + |\gamma_{k-1}|^2}$, for every $k \geq 1$. This implies $ a_{0,0} \leq a_{0,1} \leq \dots \leq a_{0,n-1}$.
    \item $\norm{G_k \Theta(\gamma_k)}_2 = \norm{G_k}_2$, and $\norm{G_k \Theta (\gamma_k)}_F = \norm{G_k}_F = \norm{G_0}_F = 1$, for every $k \geq 0$.
    \item $\norm{G_{k+1}}_1 = \norm{G_k \Theta(\gamma_k)}_1$, for every $0 \leq k \leq n-2$.
    \item $|a_{j,k}|, |b_{j,k}| \leq 1$, for every $0 \leq k \leq n-1$, and $0 \leq j \leq n-k-1$.
    \item $|\gamma_k| \leq \left( a_{0,0} \prod_{j=0}^{k-1} \sqrt{1 + |\gamma_j|^2} \right)^{-1}$, for every $0 \leq k \leq n-1$.
\end{enumerate}
\end{lemma}

\begin{proof}
\begin{enumerate}[(a)]
    \item We may read off from the first row of \cref{eq:G_n recurrence} that $a_{0,k} = a_{0,k-1} \sqrt{1 + |\gamma_{k-1}|^2}$, for every $k \ge 1$.
    \item This follows as $\Theta(\gamma_k) \in \mathrm{SU}(2)$, because the matrix 2-norm and the Frobenius norm are invariant under right multiplication by unitary matrices. $\norm{G_0}^2_F = 1$ because $|a^\ast(z)|^2 + |b(z)|^2 = 1$ by \cref{lem:nlft-ab-degree} (see discussion near \cref{eq:ab-norm}).
    \item Note that for the matrix $G_k \Theta(\gamma_k)$, the first entry of the second column and the last entry of the first column are guaranteed to be zero, as we have argued before this lemma. After that, we obtain the matrix $G_{k+1}$ by shifting the first column of the matrix $G_k \Theta(\gamma_k)$ by one row and then deleting the first row (see \cref{eq:G_n shift}). Thus, the $1$-norms of both columns of $G_k \Theta(\gamma_k)$ and $G_{k+1}$ are the same, and the conclusion follows.
    \item Follows from part (b).
    \item Follows since $|\gamma_k| = \left| \frac{b_{0,k}} {a_{0,k}} \right|$, and then using parts (a) and (d).
\end{enumerate}
\end{proof}

As discussed earlier, the layer stripping procedure can be applied to general function pairs. This becomes even clearer from the matrix perspective, where the algorithm is iteratively performing the “rotation and shift” operations, which can be initialized from any $n \times 2$ matrix $G_0$ as long as $a_{0,0} > 0$, and produce a layer stripping sequence. The difference is that the $a_{n-k-1,k+1}=0$ may no longer hold in \cref{eq:G_n recurrence}, but we still drop this element after the shift.

\noindent
\textit{Complexity}. From the matrix form of the layer stripping algorithm, it is clear that the time complexity of the algorithm is $\mathcal{O}(n^2)$, because at every iteration $k=0,\dots,n-1$, the rotation operation acts on $\mathcal{O}(n)$ length vectors, and each shift operation is $\mathcal{O}(1)$. This also shows that the space complexity of the algorithm is $\mathcal{O}(n)$ as both the rotation and shift operations can be performed in place.

\subsection{Fast algorithm for inverse NLFT}
\label{ssec:fast-inverse-NLFT}

The layer stripping algorithm involves substantial redundant computation. This section aims to develop a more efficient alternative by carefully analyzing and restructuring the existing method. We refer to this new algorithm as the \emph{inverse nonlinear fast Fourier transform} (inverse nonlinear FFT) algorithm, as it shares the same divide-and-conquer structure as the standard FFT. However, unlike the standard FFT, where the forward and inverse algorithms are nearly identical, the inverse nonlinear FFT differs substantially from its forward counterpart (see \cite[Section 3.1]{ni2024fast}), despite the fact that they both rely on fast polynomial multiplication.

 Our method is inspired by the superfast Toeplitz system solver introduced in \cite{ammar1989numerical}. It is known that the $\mathrm{SU}(1,1)$ NLFT is closely related to Toeplitz solvers, and we observe that a similar idea can be adapted to the $\mathrm{SU}(2)$ inverse NLFT problem considered in this work.

Suppose $m=\ceil{\frac{n}{2}}$, and we divide the task of computing $\bga$ into two parts: $\gup = (\gamma_0,\ldots, \gamma_{m-1})^T$ and $\gdo = (\gamma_m,\ldots, \gamma_{n-1})^T$. It can be observed that computing $\gup$ depends only on the first $m$ rows of $\ba_0$ and $\bb_0$. This is because the output of the first $m$ steps of the layer stripping process, \cref{eq:G_n recurrence} and \cref{eq:G_n shift}, does not involve the last $n-m$ rows of the matrices $G_0, \ldots, G_{m-1}$. 

On the other hand, the computation of $\gdo$ only relies on $G_m = (\ba_m, \bb_m)$. Therefore, if we have an efficient way to compute $\ba_m$ and $\bb_m$, we can avoid computing $a_{j,k}$ and $b_{j,k}$ for $0 \le j \le m-1$ and $m-j \le k \le n-1-j$, which correspond to the last $n-m$ rows of the matrices $G_0, \ldots, G_{m-1}$.

Next, we explain how to efficiently calculate $(\ba_m, \bb_m)$ given $(\ba_0, \bb_0)$ and $\gup$. Define polynomials $\eta_{k\to m}(z)$ and $\xi_{k\to m}(z)$ by
\begin{equation}\label{eq: defi of xi and eta}
    \overbrace{[\gamma_k,\ldots, \gamma_{m-1}]} = \begin{pmatrix}
        \eta_{k\to m}^*(z)& \xi_{k\to m}(z)\\
        -\xi_{k\to m}^*(z)& \eta_{k\to m}(z)
    \end{pmatrix},
\end{equation}
and we will abbreviate $\eta_{0\to m}$ and $\xi_{0\to m}$ as $\eta_{m}$ and $\xi_{m}$. With these notations, the identity 
\begin{equation}\label{eq: split seq}
    \overbrace{[\gamma_0,\ldots, \gamma_{m-1}]} \; \overbrace{[0,\ldots,0,\gamma_m,\ldots, \gamma_{n-1}]} = \overbrace{[\gamma_0,\ldots, \gamma_{n-1}]}
\end{equation}
can be written as
\begin{equation}
    \begin{pmatrix}
        \eta_{m}^*(z)& \xi_{m}(z)\\
        -\xi_{m}^*(z)& \eta_{m}(z)
    \end{pmatrix}\begin{pmatrix}
            a_m(z) & z^m b_m(z)\\ -z^{-m} b_m^*(z) & a_m^*(z)
        \end{pmatrix} = \begin{pmatrix}
            a_0(z) & b_0(z)\\ -b_0^*(z) & a_0^*(z)
        \end{pmatrix}.
\end{equation}
We can invert the first matrix to obtain
\begin{equation}
\label{eq:key-layer-stripping-identity}
    \begin{pmatrix}
        z^m b_m(z)\\ a_m^*(z)
    \end{pmatrix} = \begin{pmatrix}
        \eta_{m}(z)& -\xi_{m}(z)\\
        \xi_{m}^*(z)& \eta_{m}^*(z)
    \end{pmatrix}\begin{pmatrix}
             b_0(z)\\ a_0^*(z)
        \end{pmatrix}.
\end{equation}
This involves only a few polynomial multiplications and can therefore be implemented efficiently using the fast Fourier transform (FFT). The only remaining task is to compute $\eta_m(z)$ and $\xi_m(z)$ from $[\gamma_0,\ldots, \gamma_{m-1}]$ efficiently. This can be accomplished using a divide-and-conquer approach as described in \cite[Section 3.1]{ni2024fast}. To be self-contained, we present the method here, but as a recursive process.
Assume $l = \ceil{\frac{m}{2}}$, and that the recursive steps have already yielded
$$
\overbrace{[\gamma_0,\ldots, \gamma_{l-1}]} \quad \text{and}\quad \overbrace{[\gamma_l,\ldots, \gamma_{m-1}]},
$$
and then we may compute $\overbrace{[\gamma_0,\ldots, \gamma_{m-1}]}$ using \cref{eq: split seq}, where the polynomial-valued matrix multiplications can also be efficiently carried out using FFT. Once this is done, $\eta_m(z)$ and $\xi_m(z)$ can be extracted directly from the resulting matrix $\overbrace{[\gamma_0,\ldots, \gamma_{m-1}]}$.

In the calculation procedure, we define polynomials $\xi_n^{\sharp} := z^n\xi_n^*$, and $\eta_n^{\sharp} := z^n\eta_n^*$, $\xi_{m\to n}^{\sharp} = z^{n-m}\xi_{m\to n}^*$, and likewise $\eta_{m\to n}^{\sharp} = z^{n-m}\eta_{m\to n}^*$. While this change does not essentially affect the underlying computation, working with $\xi_n^{\sharp}, \eta_n^{\sharp}$ rather than the Laurent polynomials $\xi_n^*,\eta_n^*$ eliminates the negative powers, and also makes polynomial multiplication directly correspond to the convolution of the coefficient vectors without any shifting. 

The complete algorithm is presented in \cref{alg: phase factor finding}. Note that the output also includes the coefficients of $\xi_n$ and $\eta_n$, which are needed in the recursive steps. We remark that in Step 2, the relevant coefficients $\{(a_{j,0},b_{j,0})\}_{j=0}^{m-1}$ does not correspond to a function pair in $\mathcal{S}$. However, a crucial fact is that \cref{eq:key-layer-stripping-identity} also holds for the recurrence relations given by \cref{eq: rec a_k^* b_k}, even when $(a_0,b_0) \not \in \mathcal{S}$. Thus, as discussed in the second last paragraph of \cref{ssec:layer-stripping}, there is no ambiguity in referring to its layer stripping sequence as $\{\gamma_k\}_{k=0}^{m-1}$ and proceeding with the recursive call.

\begin{algorithm}[htbp]
        \caption{Inverse Nonlinear Fast Fourier Transform}
        \label{alg: phase factor finding}
    \raggedright
\textbf{Input: } Sequences $\{a_{0,0},\ldots,a_{n-1,0}\}$, and $\{b_{0,0},\ldots,b_{n-1,0}\}$.

\textbf{Output: } Sequences $\{\gamma_0,\ldots, \gamma_{n-1}\}$, and the coefficients of $\xi_n(z)$ and $\eta_n(z)$.

\textbf{Step 1: }If $n=1$, exit algorithm with output $\gamma_0 = \frac{b_{0,0}}{a_{0,0}}$, $\xi_1 = \frac{\gamma_0}{\sqrt{1+|\gamma_0|^2}}$, and $\eta_1 = \frac{1}{\sqrt{1+|\gamma_0|^2}}$.

\textbf{Step 2: }If $n\ge 2$, let $m = \lceil \frac{n}{2}\rceil$, then use the first $m$ coefficients of $a_0^*(z)$ and $b_0(z)$ to calculate $\{\gamma_k\}_{k=0}^{m-1}$, and polynomials $\xi_m$, and $\eta_m$, which is a recursive call of size $\frac{n}{2}$.

\textbf{Step 3: }Use the first $n$ coefficients of $a_0^*(z)$ and $b_0(z)$, together with $\xi_m$ and $\eta_m$, to calculate the first $n-m$ coefficients of $a_m^*(z)$ and $b_m(z)$ through
\begin{equation}\label{eq: xm and ym update}
    \begin{pmatrix}
        a_m^*(z)\\b_m(z)
    \end{pmatrix} = \frac{1}{z^m}\begin{pmatrix}
        \eta_m^{\sharp}(z) & \xi_m^{\sharp}(z)\\-\xi_m(z) & \eta_m(z)
    \end{pmatrix}\begin{pmatrix}
        a_0^*(z)\\b_0(z)
    \end{pmatrix}.
\end{equation}
The $\frac{1}{z^m}$ term is used to cancel the common $z^m$ factors appearing in both $a_m^*$ and $b_m$.

\textbf{Step 4: }Use the first $n-m$ coefficients of $a_m^*(z)$ and $b_m(z)$ to calculate $\{\gamma_k\}_{k=m}^{n-1}$, and the polynomials $\xi_{m\to n}$ and $\eta_{m\to n}$. This is another recursive call of size $\frac{n}{2}$, with input being the coefficients of $a_m^*$ and $b_m$.

\textbf{Step 5: }Calculate $\xi_n$ and $\eta_n$ through 
\begin{equation}\label{eq: xi and eta update}
    \begin{pmatrix}
        \eta_{n}^{\sharp} & \xi_{n}\\ 
        -\xi_{n}^{\sharp} & \eta_{n}
    \end{pmatrix} = \begin{pmatrix}
        \eta_{m}^{\sharp} & \xi_{m}\\ 
        -\xi_{m}^{\sharp} & \eta_{m}
    \end{pmatrix}\begin{pmatrix}
        \eta_{m\to n}^{\sharp} & \xi_{m\to n}\\ 
        -\xi_{m\to n}^{\sharp} & \eta_{m\to n}
    \end{pmatrix}.
\end{equation}
Finally, output $\{\gamma_k\}_{k=0}^{n-1}$ and polynomials $\xi_n$, $\eta_n$.
\end{algorithm}

\noindent
\textit{Complexity}. The inverse nonlinear FFT algorithm contains two recursive calls of size $\frac{n}{2}$ (Steps 2 and 4), and several polynomial multiplications (Steps 3 and 5). These polynomial multiplications are of $\mo{n\log n}$ time complexity when FFT is used to calculate the convolution of their coefficients. Therefore, the total time complexity is
\begin{equation}
    n\log n + 2\left(\frac{n}{2}\log\frac{n}{2}\right) + 4\left(\frac{n}{4}\log\frac{n}{4}\right) + \cdots = \mo{n\log^2 n}.
\end{equation}

For the space complexity of the inverse nonlinear FFT algorithm, notice that at any stage of execution of the entire algorithm, there can be at most one active recursive call of size $n/2^k$, where $0 \le k \le \ceil{\log_2 n}$. This is because in \cref{alg: phase factor finding}, Step 2 must finish before Step 4 can start, and when Step 4 is executing, the memory used during the execution of the recursive call in Step 2  can be freed. During the execution of a recursive call of size $n/2^k$, all operations, including the polynomial multiplications, require a storage space of size bounded by $C n/2^k$, for some constant $C$ independent of $n$ and $k$. Thus, the maximum space required to run \cref{alg: phase factor finding} is bounded by
\begin{equation}
\label{eq:space-complexity-fast-inverse-NLFT}
C (n + n/2 + n/4 + \dots) \le 2C n,
\end{equation}
which shows that the space complexity of the algorithm is $\mathcal{O}(n)$.

\section{Numerically stable inverse NLFT}
\label{sec:stability}

This section discusses the numerical stability of the two inverse NLFT algorithms. In \cref{ssec:floating-point-background}, we briefly review the notions of forward and backward numerical stability. Next, in \cref{ssec:numerical-instability}, we use examples to demonstrate that the outerness condition of $a^*$ is crucial for numerical stability. We then establish the relationship between NLFT and the Gaussian elimination process for a matrix with displacement structure in \cref{subsec: gauss elimination}, which plays an important role in the subsequent stability proofs. Finally, we present the stability proofs of the layer stripping algorithm and the inverse nonlinear FFT algorithm in \cref{subsec: proof stab layer stripping} and \cref{subsec: proof stab fast alg}, respectively.

\subsection{Background on floating point arithmetic and error analysis}
\label{ssec:floating-point-background}

We briefly review some fundamental concepts from numerical error analysis; for comprehensive discussions, we refer readers to \cite{higham2002accuracy}. To analyze the numerical stability of our algorithm, we adopt the standard model of floating-point arithmetic. Specifically, for any basic arithmetic operation $\circ$, the computed result $\mathrm{fl}(a \circ b)$ satisfies
\[
\mathrm{fl}(a \circ b) = (a \circ b)(1 + \delta), \quad |\delta| \leq \epsilon_{\mathrm{ma}},
\]
where $\epsilon_{\mathrm{ma}}$ denotes the machine precision.

Consider an algorithm whose exact output is $x$, and suppose the target precision is $\epsilon$. We say the algorithm has a bit requirement $r$ if, when using floating point arithmetic with $\epsilon_{\mathrm{ma}} = 2^{-r}$, the computed output $\hat{x}$ satisfies $\frac{\|x - \hat{x}\|}{\norm{x}} \leq \epsilon$. In practice, we prefer algorithms that operate reliably under fixed precision regardless of the problem size, such as $r = 53$ for IEEE double precision. However, in the worst case, numerical error may accumulate, making this goal theoretically unattainable. We say an algorithm is \emph{numerically stable} if the bit requirement is $r = \Or(\operatorname{polylog}(n, 1/\epsilon))$, where $n$ denotes the problem size. In practice, such algorithms tend to perform robustly using standard double precision. On the other hand, an algorithm is numerically unstable if the bit requirement of an algorithm is at least $r = \Omega(\poly(n))$, in which case the error would accumulate and blow up rapidly for moderate values of $n$ under standard double precision arithmetic in practice.

Numerical stability is typically assessed via forward and backward error analysis. The forward error measures the deviation of the computed solution $\hat{x}$ from the exact solution $x$, while the backward error reflects the smallest perturbation to the input that would make $\hat{x}$ an exact solution. For example, when solving a linear system $Ax = b$, assuming $A,b$ are non-zero, we have:
\begin{align}
\text{Forward (relative) error } &:= \frac{\|\hat{x} - x\|}{\|x\|}, \\
\text{Backward (relative) error } &:= \min_{\Delta A} \left\{ \frac{\|\Delta A\|}{\|A\|} : (A + \Delta A)\hat{x} = b \right\},
\end{align}
for some choice of norms (usually the $2$-norm). While forward error bounds are often the ultimate goal, backward error bounds are generally easier to derive. In the case of linear systems, forward and backward errors are linked by the condition number $\kappa(A):= \|A\| \|A^{-1}\|$, with the forward error bounded by the product of the condition number and the backward error \cite[Theorem 7.2]{higham2002accuracy}. We provide the formal version of this sensitive analysis statement in the following lemma.

\begin{lemma}\label{lemma: pert bound}
    Let $A x=b$ and $(A+\Delta A) y=b+\Delta b$, and assume $\kappa(A)\frac{\norm{\Delta A}}{\norm{A}}<\half$ and $b \neq 0$. Then
$$
\frac{\|x-y\|}{\|x\|} \leq 2 \kappa(A)\left(\frac{\norm{\Delta b}}{\|b\|}+\frac{\norm{\Delta A}}{\norm{A}}\right).
$$
\end{lemma}
\begin{proof}
    In \cite[Theorem 7.2]{higham2002accuracy}, we may choose $\epsilon = 1$, $E = \Delta A$, $f = \Delta b$ and obtain
    \begin{equation}
\frac{\|x-y\|}{\|x\|} \leq \frac{1}{1-\kappa(A)\frac{\norm{\Delta A}}{\norm{A}}}\left(\frac{\kappa(A)\|\Delta b\|}{\norm{A}\|x\|}+\left\|A^{-1}\right\|\|\Delta A\|\right)\le 2 \left(\kappa(A)\frac{\|\Delta b\|}{\norm{b}}+\kappa(A)\frac{\norm{\Delta A}}{\norm{A}}\right).
\end{equation}
\end{proof}

\subsection{Illustration of numerical instability without outerness condition}
\label{ssec:numerical-instability}

For a given polynomial $b(z)$, the complementary Laurent polynomial $a(z)$ such that $(a, b) \in \mathcal{S}$ is not unique. The number of possible choices for $a(z)$ is typically exponential in the polynomial degree $n$ (cf.  \cite{Haah2019}), and we provide a precise count in \cref{lem:a-astar-equal-b-bstar}(f). However, as argued in \cref{ssec:complementarity-prelim}, there exists a unique choice of $a(z)$ such that $a^*(z)$ is outer. To illustrate the importance of this property in the layer stripping algorithm and the inverse nonlinear FFT algorithm, we present an example demonstrating potential failures when $a^*(z)$ is non-outer. 

Consider a randomly chosen degree $n-1$ polynomial $b(z)$ with real coefficients such that $\| b(z) \|_{L^\infty(\TT)} \leq \frac{1}{2}$, and let $\ao(z)$ be its complementary Laurent polynomial with $\ao^*(z)$ outer. Define 
\[
\ano(z) = \omega z^{-(n-1)} \ao^*(z) = \omega z^{-(n-1)} \ao(z^{-1}),
\]  
where the sign $\omega\in\{\pm1\}$ is to ensure $\ano(\infty)>0$. Since $\ano(z)\ano^*(z) = \ao(z)\ao^*(z)$ by definition, it follows that $(\ano, b) \in \mathcal{S}$ with $\ano^*(z)$ being non-outer. Performing the layer stripping algorithm on $(\ano, b)$ reveals significant instability.  

In \cref{fig: diff_n}, we increase the polynomial degree $n$, apply the layer stripping algorithm to obtain $\hat{\bga}$, and plot the difference between the NLFT of $\hat{\bga}$ and the ground truth. Even for moderate values of $n$, the difference is $O(1)$. In particular, for $n = 80$, we plot the error in each component $\gamma_k$ in \cref{fig: err_acc}. The results demonstrate that the error accumulates and grows exponentially during the iterations, consistent with the worst-case upper bound in \cite{Haah2019}.

\begin{figure}[!htb]
    \centering
    \subfloat[\label{fig: diff_n}]{
    \begin{minipage}[c]{0.45\textwidth}
        \centering
        \includegraphics[width=0.95\textwidth]{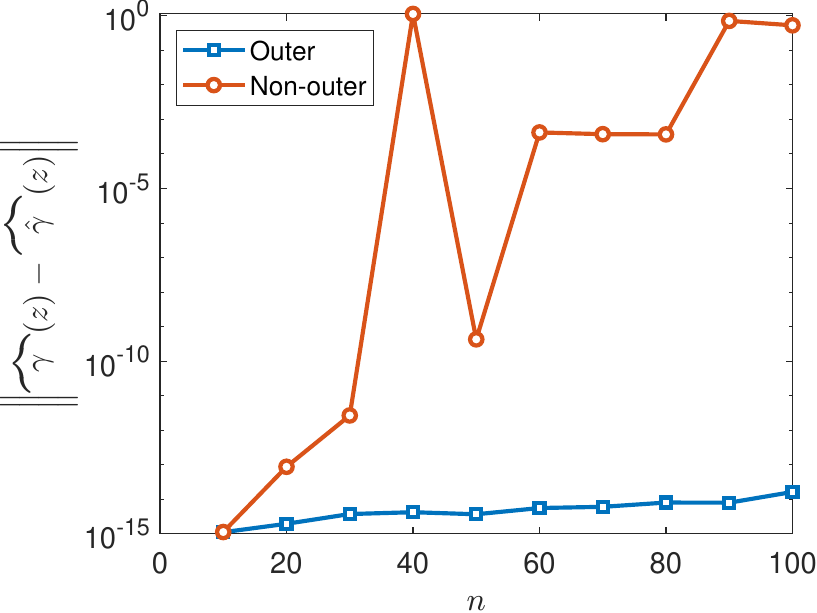}
    \end{minipage}
    }
    \subfloat[\label{fig: err_acc}]{
    \begin{minipage}[c]{0.45\textwidth}
        \centering
        \includegraphics[width=0.95\textwidth]{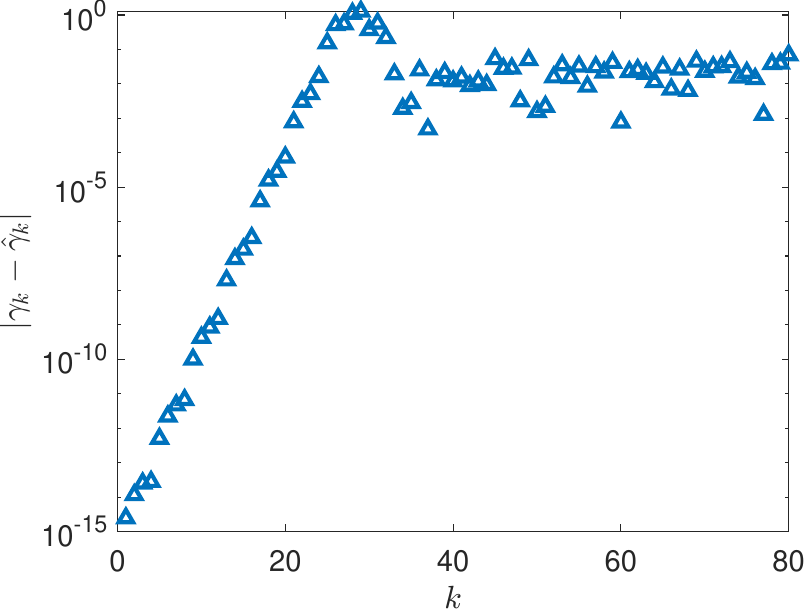}
    \end{minipage}
    }
    \caption{(A) The residual $\|\protect\overbrace{\bga} - \protect\overbrace{\hat{\bga}}\|_{L^{\infty}(\TT)}$ for different degree $n$. (B) For $n=80$, the error for each component $\gamma_k$.}
\end{figure}

\subsection{Inverse NLFT via Gauss elimination}
\label{subsec: gauss elimination}

Fast algorithms for solving linear systems with structured coefficient matrices have been extensively studied. At first glance, these methods may appear unrelated to the layer stripping procedure. However, an algorithm known as Schur's algorithm, which is used for inverting matrices with displacement structure, is closely connected to the layer stripping method discussed here. In the following sections, we explore this surprising connection and demonstrate how it facilitates the stability analysis of layer stripping and the inverse nonlinear FFT algorithm.

We begin by defining a matrix $K := T(\ba_0)T(\ba_0)^* + T(\bb_0)T(\bb_0)^*$, where $T(\ba_0)$ is the lower triangular Toeplitz matrix
\begin{equation}
\label{eq:Ta0-def}
    T(\ba_0) := \begin{pmatrix}
        a_{0,0}&&&&\\
        a_{1,0}&a_{0,0}&&&\\
        a_{2,0}&a_{1,0}&a_{0,0}&&\\
        \vdots&\ddots&\ddots&\ddots&\\
        a_{n-1,0}&\cdots&a_{2,0}&a_{1,0}&a_{0,0}
    \end{pmatrix},
\end{equation}
with first column $\ba_0 = (a_{0,0}, a_{1,0},\ldots,a_{n-1,0})^T$. It follows from \cite[Lemma~2]{KAILATH1979395} that $K$ satisfies the identity
\begin{equation}\label{eq: displacement structure K}
    K - Z_nKZ_n^* = G_0 G_0^* = \ba_0\ba_0^*+ \bb_0\bb_0^*, \quad Z_n:= \begin{pmatrix}
    0&&&\\
    1&0&&\\
    &\ddots &\ddots&\\
    &&1&0
\end{pmatrix}_{n\times n},
\end{equation}
where $Z_n$ is the lower shift matrix. The relationship \cref{eq: displacement structure K} is called the \textit{displacement low-rank structure} of $K$, or the displacement structure for short. The \textit{displacement rank} is defined as the rank of the right-hand side matrix $G_0 G_0^\ast$, which is at most two. For a displacement low-rank matrix, the generalized Schur algorithm \cite{Kailath1995displacement} can perform the $LDL^*$ factorization of it in $\Or(n^2)$ time, where $L$ is a lower-triangular matrix with ones on the diagonal and $D$ is a diagonal matrix (see \cite[Section~4.2.6]{GolubVan2013} for a discussion on the $LDL^\ast$ factorization). Moreover, the $LDL^\ast$ factorization of $K$ is unique as $K$ is non-singular, since $T(\ba_0)$ is also non-singular in \cref{eq:Ta0-def} (recall that $a_{0,0} > 0$ from \cref{lem:layer-stripping-prop}(a)). We will briefly review the algorithm for our specific $K$ and find that it essentially performs layer stripping for the sequences $\ba_0$ and $\bb_0$. 

When we determine $\gamma_0 = \frac{b_{0,0}}{a_{0,0}}$ and multiply a rotation matrix to the right of $(\ba_0, \bb_0)$ as in \cref{eq:G_n recurrence}, we obtain an $n\times 2$ matrix 
\begin{equation}\label{eq: def R}
 \begin{pmatrix}
     a_{0,1} & 0\\ \bu & \bb_1
 \end{pmatrix} := \begin{pmatrix}
a_{0,1} & 0 \\
a_{1,1} & b_{0,1} \\
\vdots & \vdots\\
a_{n-1,1} & b_{n-2,1} 
\end{pmatrix} =  \begin{pmatrix}
    \ba_0 & \bb_0
 \end{pmatrix} \frac{1}{\sqrt{1 + |\gamma_0|^2}} 
    \begin{pmatrix}
        1 & -\gamma_0  \\
         \ol{\gamma_0}  & 1
    \end{pmatrix},
\end{equation}
where the column vector $\bu = [a_{1,1},a_{2,1},\ldots,a_{n-1,1}]^T$. Plugging this into \cref{eq: displacement structure K}, we get
\begin{equation}\label{eq: displacement with R}
    K - Z_nKZ_n^* = \begin{pmatrix}
    \ba_0 & \bb_0
 \end{pmatrix}\begin{pmatrix}
     \ba_0^*\\  \bb_0^*
 \end{pmatrix}= \begin{pmatrix}
     a_{0,1} & 0\\ \bu & \bb_1
 \end{pmatrix}\begin{pmatrix}
     \ol{a_{0,1}} & \bu^*\\ 0 & \bb_1^*
 \end{pmatrix},
\end{equation}
From this, we may recover the first column of $K$ as $\begin{pmatrix}
        \abs{a_{0,1}}^2\\ \ol{a_{0,1}}\bu
\end{pmatrix}$, and therefore the first step of the $LDL^*$ decomposition of $K$ is
\begin{equation}\label{eq: first step LDL}
    K = \begin{pmatrix}
        1 & 0\\
        \bu/a_{0,1} & I
    \end{pmatrix}\begin{pmatrix}
        \abs{a_{0,1}}^2 & 0\\
        0 & K_1
    \end{pmatrix}\begin{pmatrix}
        1 & \bu^*/\ol{a_{0,1}}\\
        0 & I
    \end{pmatrix} = \begin{pmatrix}        \abs{a_{0,1}}^2&a_{0,1}\bu^*\\ \ol{a_{0,1}}\bu&\bu\bu^*+K_1
    \end{pmatrix},
\end{equation}
where $K_1$ is the first Schur complement of $K$. If we substitute this expression into \cref{eq: displacement with R}, we obtain
\[
\begin{pmatrix}        \abs{a_{0,1}}^2&a_{0,1}\bu^*\\\ol{a_{0,1}}\bu&\bu\bu^*+K_1
    \end{pmatrix} - \begin{pmatrix}
        0&\\ \be_0& Z_{n-1}
    \end{pmatrix}\begin{pmatrix}
        \abs{a_{0,1}}^2&a_{0,1}\bu^*\\\ol{a_{0,1}}\bu&\bu\bu^*+K_1
    \end{pmatrix}\begin{pmatrix}
        0& \be_0^*\\& Z_{n-1}^*
    \end{pmatrix} = \begin{pmatrix}
        \abs{a_{0,1}}^2&a_{0,1}\bu^*\\\ol{a_{0,1}}\bu&\bu\bu^*+\bb_1\bb_1^*
    \end{pmatrix},
\]
where $\be_0$ is the unit vector with the first element $1$. By comparing the lower-right block, we deduce that $K_1$ satisfies
\begin{equation}\label{eq: K2 displacement}
    K_1-Z_{n-1} K_1 Z_{n-1}^* = (a_{0,1}\be_0+Z_{n-1}\bu)(a_{0,1}\be_0+Z_{n-1}\bu)^* + \bb_1\bb_1^* = \ba_1\ba_1^*+ \bb_1\bb_1^*,
\end{equation}
where $a_{0,1}\be_0+Z_{n-1}\bu = \ba_1$ is obtained from \cref{eq: def R}. This equation is in the same form as \cref{eq: displacement structure K}; therefore, we can perform this elimination step iteratively as in \cref{eq: first step LDL}. Finally, we can express $K = LDL^*$, where
\begin{equation}
\label{eq:L-D-def}
    L  = \begin{pmatrix}
        1&&&&\\
        \frac{a_{1,1}}{a_{0,1}}&1&&&\\
        \frac{a_{2,1}}{a_{0,1}}&\frac{a_{1,2}}{a_{0,2}}&1&&\\
        \vdots&\vdots&\vdots&\ddots&\\
        \frac{a_{n-1,1}}{a_{0,1}}&\frac{a_{n-2,2}}{a_{0,2}}&\frac{a_{n-3,3}}{a_{0,3}}&\cdots&1
    \end{pmatrix}, \  D  = \begin{pmatrix}
        \abs{a_{0,1}}^2&&&&\\
        &\abs{a_{0,2}}^2&&&\\
        &&\abs{a_{0,3}}^2&&\\
       &&&\ddots&\\
        &&&&\abs{a_{0,n}}^2
    \end{pmatrix}.
\end{equation}
The absolute value notation in $D$ emphasizes its positive definiteness, although by construction we already have $a_{0,k} > 0$ for all $k$.

The significance of relating layer stripping with the displacement structure and the Cholesky factorization lies in establishing the following lemma, which plays a central role in the stability proofs presented in the subsequent subsections.

\begin{lemma}\label{lemma: L gamma = b/a00}
    Let $L$ be as defined in \cref{eq:L-D-def}. Then the NLFT coefficient vector $\bga$ satisfies
    \begin{equation}\label{eq: gamma linear system}
        L\boldsymbol{\gamma} = \frac{1}{a_{00}}\bb_0.
    \end{equation}
\end{lemma}
\begin{proof}
    We prove the lemma by induction. When $n=1$, the lemma holds since $L = 1$ and $\gamma_0=\frac{b_{0,0}}{a_{0,0}}$.

    Assuming the lemma holds for the $n-1$ case, let us consider the lemma for $n$. For convenience, we introduce some notations
    \begin{equation}
        \ba_{+,1}:= \begin{pmatrix}
            a_{1,1}\\ \vdots\\ a_{n-1,1}
        \end{pmatrix},\ 
        \bb_{+,0}:= \begin{pmatrix}
            b_{1,0}\\ \vdots\\ b_{n-1,0}
        \end{pmatrix},\ 
        \bga_{+} := \begin{pmatrix}
            \gamma_{1}\\ \vdots\\ \gamma_{n-1}
        \end{pmatrix},\ \Theta(\gamma_k) := \frac{1}{\sqrt{1 + |\gamma_k|^2}} 
    \begin{pmatrix}
        1 & -\gamma_k \\
        \ol{\gamma_k}  & 1
    \end{pmatrix}.
    \end{equation}
    We denote $L_1\in\CC^{(n-1)\times (n-1)}$ be the matrix containing the last $n-1$ rows and columns of $L$. After calculating $\gamma_0$, the remaining procedure is reduced to the layer stripping process starting from $(\ba_1, \bb_1)$. First, by the induction hypothesis, we have 
    \begin{equation}\label{eq: induction hyp L1}
    L_1\bga_{+} = \frac{1}{a_{0,1}}\bb_1.
    \end{equation}
    Next, from \cref{eq: def R} we know that for $1 \le k \le n-1$,
    \begin{equation}
        (a_{k,1}, b_{k-1,1})\Theta(-\gamma_0) = (a_{k,0}, b_{k,0}).
    \end{equation}
    The second component of this formula gives
    \begin{equation}
        \gamma_0 a_{k,1}+ b_{k-1,1} = \sqrt{1+\abs{\gamma_0}^2} b_{k,0} = \frac{a_{0,1}}{a_{0,0}}b_{k,0},
    \end{equation}
    where the last equality derives from $\quad a_{0,1} = \sqrt{1+\abs{\gamma_0}^2} a_{0,0},$ which is also due to \cref{eq: def R}. We may collect the equations for $1\le k\le n-1$ and write them in vector form to obtain
    \begin{equation*}
        \frac{\gamma_0}{a_{0,1}}\ba_{+,1}+ \frac{1}{a_{0,1}}\bb_{1} =  \frac{1}{a_{0,0}}\bb_{+,0}.
    \end{equation*}
    Together with \cref{eq: induction hyp L1}, we then have
    \begin{equation}
        \frac{\gamma_0}{a_{0,1}}\ba_{+,1}+ L_1\bga_{+} =  \frac{1}{a_{0,0}}\bb_{+,0},
    \end{equation}
    which can be written in the block matrix form
    \begin{equation*}
        \begin{pmatrix}
            1&\\
            \frac{\ba_{+,1}}{a_{0,1}} & L_1
        \end{pmatrix}\begin{pmatrix}
            \gamma_0\\ \bga_{+}
        \end{pmatrix} = \frac{1}{a_{0,0}}\begin{pmatrix}
            {b_{0,0}} \\ {\bb_{+,0}}
        \end{pmatrix}.
    \end{equation*}
This proves $L\bga = \frac{1}{a_{00}}\bb_0$ and completes the induction step.
\end{proof}

\subsection{Proof of stability for the layer stripping algorithm}
\label{subsec: proof stab layer stripping}

\cref{lemma: L gamma = b/a00} suggests an approach to prove the stability of the layer stripping algorithm through backward error estimation. In contrast, prior work \cite{Haah2019} utilized forward error analysis to show that the error in the layer stripping algorithm can be controlled when $\Or(n)$ bits are used. However, this result is unsatisfactory as it does not establish the algorithm's stability under the usual constraint of using only $\Or(\polylog(n))$ bits. In this work, we present a proof of the algorithm's stability under the condition that $a_0^*(z)$ has no zeros in $\cDD$. 

The following \cref{thm: backward stable} compares the output of the layer stripping process with the solution of the linear system in \cref{eq: gamma linear system}. We show that the computed solution $\hbga$ is the exact solution to a slightly perturbed system, which is in the spirit of a backward stability result. Then, in \cref{lemma: norm bounds}, we establish the well-conditionedness of \cref{eq: gamma linear system}. Finally, these results are combined into the overall stability statement in \cref{thm: forward stable layer stripping}.

We use the notation $\hat{\cdot}$ to indicate quantities computed with round-off errors in floating-point arithmetic during the layer stripping procedure. We remark that discussing $\hat{L}$ directly is ambiguous, since $L$ is not explicitly formed during the algorithm. However, we introduce the matrices
\begin{equation}
    U  = \begin{pmatrix}
        a_{0,1}&&&\\
        a_{1,1}&a_{0,2}&&\\
        \vdots&\vdots&\ddots&\\
        a_{n-1,1}&a_{n-2,2}&\cdots&a_{0,n}
    \end{pmatrix}, \  H  = \begin{pmatrix}
        a_{0,1}&&&\\
        &a_{0,2}&&\\
        &&\ddots&\\
        &&&a_{0,n}
    \end{pmatrix}, \text{ and } D = H^2
\end{equation}
which satisfies $UH^{-1} = L$ and $K = LDL^* = UU^*$. We use $\hat{U}$ and $\hat{H}$ to denote the computed counterparts of $U$ and $H$, since their elements are all calculated during layer stripping.

\begin{lemma}\label{thm: backward stable}
    Let $\ema$ be the machine precision, and we assume $n\ema<0.01$. Then there exists a matrix $E$ such that 
    $$(\hat{U}\hat{H}^{-1}+E)\begin{pmatrix}
        \hga_0\\ \vdots \\ \hga_{n-1}
    \end{pmatrix} = \frac{1}{a_{0,0}}\begin{pmatrix}
        b_{0,0}\\ \vdots \\ b_{n-1,0}
    \end{pmatrix},$$
    and satisfies the entry-wise inequality $\abs{E}\le 30n\ema \abs{\hat{U}\hat{H}^{-1}}$. \end{lemma}
\begin{proof}
    We prove the theorem by induction. When $n=1$, the theorem holds since $(1+\Or(\ema))\hga_0 = \frac{b_{0,0}}{a_{0,0}}$ follows from $\gamma_0=\frac{b_{0,0}}{a_{0,0}}$ and floating point arithmetic rules.

    Assuming the theorem already holds for the $n-1$ case, let us consider the theorem for $n$. We denote $M:= \hat{U}\hat{H}^{-1}\in\CC^{n\times n}$, and $M_1\in\CC^{(n-1)\times (n-1)}$ be the matrix containing the last $n-1$ rows and columns of $M$. After calculating $\gamma_0$, the remaining procedure is reduced to the layer stripping process starting from $(\ha_{k,1})_{k=0}^{n-2}$, $(\hb_{k,1})_{k=0}^{n-2}$. By the induction hypothesis, there exists some matrix $E_1\in\CC^{(n-1)\times (n-1)}$ such that
    \begin{equation}\label{eq: induction hyp}
    (M_1+E_1)\begin{pmatrix}
        \hga_1\\ \vdots \\ \hga_{n-1}
    \end{pmatrix} = \frac{1}{\ha_{0,1}}\begin{pmatrix}
        \hb_{0,1}\\ \vdots \\ \hb_{n-2,1}
    \end{pmatrix},
    \end{equation}
    and $\abs{E_1}\le 30(n-1)\ema \abs{M_1}$. From \cref{eq:G_n recurrence} we know that the exact quantities satisfy
    $$(a_{k,1}, b_{k-1,1})\Theta(-\gamma_0) = (a_{k,0}, b_{k,0}),\quad a_{0,1} = \sqrt{1+\abs{\gamma_0}^2} a_{0,0}.$$
    We may use the stability of Givens rotation \cite{higham2002accuracy} to obtain
    \begin{equation}\label{eq: pf1}
        (\ha_{k,1}(1+\eps_{k}^{(1)}), \hb_{k-1,1}(1+\eps_{k}^{(2)}))\Theta(-\hga_0) = (a_{k,0}, b_{k,0}),
    \end{equation}
    \begin{equation}\label{eq: pf2}
        \ha_{0,1} = \sqrt{1+\abs{\hga_0}^2} a_{0,0} (1+\eps^{(0)}),
    \end{equation}
    for some $\abs{\eps_{k}^{(1)}}\le 10\ema$, $\abs{\eps_{k}^{(2)}}\le 10\ema$, and $\abs{\eps^{(0)}}\le 10\ema$. The second component of \cref{eq: pf1} gives
    \begin{equation*}
        \ha_{k,1}(1+\eps_{k}^{(1)})\hga_0+ \hb_{k-1,1}(1+\eps_{k}^{(2)}) = \sqrt{1+\abs{\hga_0}^2} b_{k,0}.
    \end{equation*}
    Dividing it by \cref{eq: pf2}, we get
    \begin{equation}\label{eq: pf4}
        \frac{\ha_{k,1}}{\ha_{0,1}}(1+\eps_{k}^{(3)})\hga_0+ \frac{\hb_{k-1,1}}{\ha_{0,1}}(1+\eps_{k}^{(4)}) =  \frac{b_{k,0}}{a_{0,0}},
    \end{equation}
    where the error terms $\abs{\eps_{k}^{(3)}} = \abs{\eps_{k}^{(1)}+\eps^{(0)}+\eps_{k}^{(1)}\eps^{(0)}}\le 21\ema$, and $\abs{\eps_{k}^{(4)}} = \abs{\eps_{k}^{(2)}+\eps^{(0)}+\eps_{k}^{(2)}\eps^{(0)}}\le 21\ema$.
    For convenience, we denote the column vectors
    \begin{equation}
        \hba_{+,1}:= \begin{pmatrix}
            \ha_{1,1}\\ \vdots\\ \ha_{n-1,1}
        \end{pmatrix},\ 
        \bb_{+,0}:= \begin{pmatrix}
            b_{1,0}\\ \vdots\\ b_{n-1,0}
        \end{pmatrix},\ 
        \hbga_{+} := \begin{pmatrix}
            \hga_{1}\\ \vdots\\ \hga_{n-1}
        \end{pmatrix}.
    \end{equation} 
    We also define two diagonal matrices $D_{\eps,3} = \diag\{\eps_{1}^{(3)},\ldots,\eps_{n-1}^{(3)}\}$ and $D_{\eps,4} = \diag\{\eps_{1}^{(4)},\ldots,\eps_{n-1}^{(4)}\}$.
    Therefore, we may write \cref{eq: pf4} in a vector form
    \begin{equation*}
        (I+D_{\eps,3})\frac{\hga_0\hba_{+,1}}{\ha_{0,1}}+ (I+D_{\eps,4})\frac{\hbb_{1}}{\ha_{0,1}} =  \frac{\bb_{+,0}}{a_{0,0}},
    \end{equation*}
    Solving for the vector $\frac{\hbb_{1}}{\ha_{0,1}}$ and plugging into \cref{eq: induction hyp}, we obtain
    \begin{equation*}
        (M_1+E_1)\hbga_{+} = (I+D_{\eps,4})^{-1}\left(\frac{\bb_{+,0}}{a_{0,0}} - (I+D_{\eps,3})\frac{\hga_0\hba_{+,1}}{\ha_{0,1}}\right),
    \end{equation*}
    which can also be written in the block matrix form
    \begin{equation*}
        \begin{pmatrix}
            1&\\
            (I+D_{\eps,3})\frac{\hba_{+,1}}{\ha_{0,1}} & (I+D_{\eps,4})(M_1+E_1)
        \end{pmatrix}\begin{pmatrix}
            \hga_0\\ \hbga_{+}
        \end{pmatrix} = \frac{1}{a_{0,0}}\begin{pmatrix}
            {b_{0,0}} \\ {\bb_{+,0}}
        \end{pmatrix}.
    \end{equation*}
    Noticing that 
    \[
    M = \begin{pmatrix}
            1&\\
            \frac{\hba_{+,1}}{\ha_{0,1}} & M_1
        \end{pmatrix},
    \]
    we conclude by setting 
    $$E = \begin{pmatrix}
            0&\\
            D_{\eps,3}\frac{\hba_{+,1}}{\ha_{0,1}} & E_1+D_{\eps,4}M_1+D_{\eps,4}E_1
        \end{pmatrix}.$$
    This matrix satisfies $\abs{E}\le 30n\ema\abs{M}$ since we have the componentwise estimate $\abs{D_{\eps,3}\frac{\hba_{+,1}}{\ha_{0,1}}}\le 21\ema \abs{\frac{\hba_{+,1}}{\ha_{0,1}}}$ and 
    $$\abs{E_1+D_{\eps,4}M_1+D_{\eps,4}E_1} \le 30(n-1)\ema \abs{M_1} + 21\ema\abs{M_1} + 630(n-1)\ema^2\abs{M_1} \le 30n\ema\abs{M_1},$$
    where we used $(n-1)\ema < 0.01$ in the last step.
\end{proof}

\begin{lemma}
\label{lemma: norm bounds}
Assume $(a_0,b_0)\in\mathcal{S}_\eta$, and $a_0^*(z)$ is outer. We have the following bounds for the corresponding matrices.
     \begin{equation}
         \eta(2-\eta)\le \lambda_{\min}(K)\le \lambda_{\max}(K)\le 2-\eta
     \end{equation}
     \begin{equation}\label{eq: L norm bound}
         1\le \norm{L}\le \eta^{-\half}
     \end{equation}
     \begin{equation}
         1\le \norm{L^{-1}}\le \eta^{-\half}
     \end{equation}
     \begin{equation}
         \norm{(DL^*)^{-1}} = \norm{(HU^*)^{-1}} \le \frac{1}{\eta(2-\eta)}
     \end{equation}
     In addition, let $\mk \subseteq [0,n-1]$ be an interval, and define a submatrix $A_{\mk}$ of a square matrix $A=(A_{ij})$ as $A_{\mk} := (A_{ij})_{i,j\in\mk}$. Then the last three bounds also hold for any $L_{\mk}$, $(L_{\mk})^{-1}$, and $((DL^*)_{\mk})^{-1}$.
\end{lemma}
\begin{proof}
    For any vector $\bv=(v_0,\ldots, v_{n-1})^T$,  consider the polynomial $v(z) = \sum_{j=0}^{n-1} v_j z^j$. Recalling that $\ba_0$ is the coefficient vector of $a_0^*(z)$, we observe that the elements of the vector $T(\ba_0) \bv$ correspond to the first $n$ coefficients of the polynomial $a_0^*(z)v(z)$. As a consequence of Parseval's identity, we can estimate
    \begin{equation}\label{eq: T(a) norm bound}
        \norm{T(\ba_0)\bv} \le \norm{a_0^*(z)v(z)}_{L^2(\TT)}\le \norm{a_0^*(z)}_{L^{\infty}(\TT)}\norm{v(z)}_{L^2(\TT)} = \norm{a_0^*(z)}_{L^{\infty}(\TT)}\norm{\bv}.
    \end{equation}
    Therefore, we obtain $\norm{T(\ba_0)}\le \norm{a_0^*(z)}_{L^{\infty}(\TT)}\le 1$, and also have $\norm{T(\bb_0)}\le \norm{b_0(z)}_{L^{\infty}(\TT)}\le 1-\eta$ for the same reason. Then we may conclude 
    $$\lambda_{\max}(K) = \norm{K} \le \norm{T(\ba_0)}^2 + \norm{T(\bb_0)}^2 \le 2-2\eta+\eta^2 \le 2-\eta.$$
    For the lower bound, we notice that the power series $a_0^*(z)^{-1} = \sum_{j=0}^{\infty} w_j z^j$ converges for $z\in\TT$ since $a_0^*(z)$ never vanishes in the closed unit disc due to \cref{lem:no-zeros-closed-disk}.  We may check that $T([w_0,\ldots, w_{n-1}]^T)$ is the inverse matrix of $T(\ba_0)$. Using the same argument as above, we can also bound
    $$\norm{T(\ba_0)^{-1}\bv} \le \norm{a_0^*(z)^{-1}v(z)}_{L^2(\TT)}\le \norm{a_0^*(z)^{-1}}_{L^{\infty}(\TT)}\norm{v(z)}_{L^2(\TT)} = \norm{a_0^*(z)^{-1}}_{L^{\infty}(\TT)}\norm{\bv}.$$
    From the relationship $\abs{a_0^*(z)}^2+\abs{b_0(z)}^2 = 1$ on the unit circle $\TT$, we obtain 
    $$\min_{z\in \TT}|a_0^*(z)| \ge \sqrt{1-(1-\eta)^2} = \sqrt{\eta(2-\eta)},$$
    which means
    $$\lambda_{\min}(K)\ge \lambda_{\min}(T(\ba_0)T(\ba_0)^*) = \min_{\bv\ne 0} \frac{\norm{\bv}^2}{\norm{T(\ba_0)^{-1}\bv}^2}\ge \frac{1}{\norm{a_0^*(z)^{-1}}_{L^{\infty}(\TT)}^2} = \min_{z\in \TT}|a_0^*(z)|^2 \ge \eta(2-\eta).$$
    
    Regarding the norm of $L$ and $L^{-1}$, we have the lower bound $\norm{L}\ge 1$ and $\norm{L^{-1}}\ge 1$ since the diagonal elements of them are all $1$. Now recall that $H=\sqrt{D}$ and $U = LH$, which implies $K = UU^*$. Then we may get bounds on the singular values of $U$,
    $$\sqrt{\eta(2-\eta)}\le \sqrt{\lambda_{\min}(K)}= \sigma_{\min}(U)\le \sigma_{\max}(U)= \sqrt{\lambda_{\max}(K)}\le \sqrt{2-\eta}.$$
    Since $H=\diag\{a_{0,1},\ldots, a_{0,n}\}$ is the diagonal part of $U$, we have
    $$\norm{H}=\max_k\{a_{0,k}\}\le \sigma_{\max}(U)\le \sqrt{2-\eta}.$$
    As $U$ is lower diagonal, $H^{-1}$ is also the diagonal part of $U^{-1}$, so similarly
    $$\norm{H^{-1}}=\max_k\{a_{0,k}^{-1}\}\le \sigma_{\max}(U^{-1})\le \frac{1}{\sqrt{\eta(2-\eta)}}.$$
    Therefore, we have the estimates
    $$\norm{L} = \norm{UH^{-1}} \le \norm{U}\norm{H^{-1}} \le \sqrt{2-\eta}\cdot \frac{1}{\sqrt{\eta(2-\eta)}} = \eta^{-1/2},$$
    $$\norm{L^{-1}} = \norm{HU^{-1}} \le \norm{H}\norm{U^{-1}} \le \sqrt{2-\eta}\cdot \frac{1}{\sqrt{\eta(2-\eta)}} = \eta^{-1/2},$$
    $$\norm{(DL^*)^{-1}} = \norm{(HU^*)^{-1}} \le \norm{U^{-1}}\norm{H^{-1}} \le \frac{1}{\sqrt{\eta(2-\eta)}}\cdot \frac{1}{\sqrt{\eta(2-\eta)}} = \frac{1}{\eta(2-\eta)}.$$        
    The bound also holds for $L_{\mk}$ since the norm of a submatrix is bounded by the norms of the original matrix. The bounds for $(L_{\mk})^{-1}$ and $((DL^*)_{\mk})^{-1}$ hold since $(L_{\mk})^{-1}=(L^{-1})_{\mk}$ and $((DL^*)_{\mk})^{-1}=((DL^*)^{-1})_{\mk}$ as triangular matrices.
\end{proof}

\begin{theorem}\label{thm: forward stable layer stripping}
    Assume $(a_0,b_0)\in\mathcal{S}_\eta$, and $a_0^*(z)$ is outer. Suppose further that $n^3\eta^{-5/2}\ema = \Or(1)$. Then, for the layer stripping algorithm, the relative error satisfies 
    \begin{equation}
        \frac{\norm{\hbga-\bga}}{\norm{\bga}} = \Or(n^3\eta^{-5/2}\ema).
    \end{equation}
    In particular, to achieve a target precision $\epsilon$, it suffices to use $\Or(\log(n \eta^{-1} \epsilon^{-1}))$ bits of precision when implementing the layer stripping algorithm in floating-point arithmetic.
\end{theorem}
\begin{proof}
    As $K = L(DL^*) = (UH^{-1})(HU^*)$, we know that $(\hat{U}\hat{H}^{-1})(\hat{H}\hat{U}^*)$ is an approximation of the LU factorization of $K$. From the proof of \cite[Theorem 3]{ni2024fast}, we have an error bound
    \begin{equation}
        \norm{(\hat{U}\hat{H}^{-1})(\hat{H}\hat{U}^*) - K}_F = \norm{\hat{U}\hat{U}^* - K}_F = \Or(n^3) \ema\norm{G_0}_F^2 = \Or(n^3\ema),
    \end{equation}
    as $\norm{G_0}_F^2 = 1$ from \cref{lem:layer-stripping-prop}(b).
 Therefore, by the bounds from \cref{lemma: norm bounds}, we obtain
 \begin{equation*}     \norm{L^{-1}}\norm{(DL^*)^{-1}}\norm{(\hat{U}\hat{H}^{-1})(\hat{H}\hat{U}^*) - K}_F\le \eta^{-1/2} \frac{1}{\eta(2-\eta)} \Or(n^3\ema) = \Or(n^3\eta^{-3/2}\ema).
 \end{equation*}
 To guarantee this quantity is no more than $\half$, we need the assumption that $\Or(n^3\eta^{-3/2}\ema)$ is less than some constant threshold.
 By the sensitivity analysis of LU factorization \cite[Theorem 9.15]{higham2002accuracy} we deduce that
\begin{equation}\label{eq: UH^-1 - L bound}
\begin{aligned}
    \frac{\norm{\hat{U}\hat{H}^{-1}-L}}{\norm{L}}&\le \frac{\norm{L^{-1}}\norm{(DL^*)^{-1}}\norm{(\hat{U}\hat{H}^{-1})(\hat{H}\hat{U}^*) - K}_F}{1-\norm{L^{-1}}\norm{(DL^*)^{-1}}\norm{(\hat{U}\hat{H}^{-1})(\hat{H}\hat{U}^*) - K}_F}\frac{\norm{K}_2}{\norm{K}_F}\\
    &\le 2 \cdot \Or(n^3\eta^{-3/2}\ema)\cdot 1 = \Or(n^3\eta^{-3/2}\ema).
\end{aligned}
\end{equation}

Then we may convert the elementwise bound in \cref{thm: backward stable} into a 2-norm bound
\begin{equation}
    \norm{E}\le\norm{E}_F\le \Or(n\ema)\norm{\hat{U}\hat{H}^{-1}}_F \le \Or(n^2\ema)\norm{\hat{U}\hat{H}^{-1}} \le \Or(n^2\ema)\norm{L},
\end{equation}
where we used $\norm{\hat{U}\hat{H}^{-1}} = \Or(1)\norm{L}$ implied from \cref{eq: UH^-1 - L bound}. Denote $\widetilde{L} := \hat{U}\hat{H}^{-1}+E$ in \cref{thm: backward stable}, then we have 
\begin{equation}
    \widetilde{L}\hbga = \frac{1}{a_{0,0}}\bb_0 = L\bga
\end{equation}
and
\begin{equation}
    \frac{\norm{\widetilde{L} - L}}{\norm{L}}\le \frac{\norm{E} + \norm{\hat{U}\hat{H}^{-1} - L}}{\norm{L}} \le \Or(n^2\ema) + \Or(n^3\eta^{-3/2}\ema) = \Or(n^3\eta^{-3/2}\ema).
\end{equation}
Here we need the assumption $n^3\eta^{-5/2}\ema = \Or(1)$ again to guarantee $\frac{\norm{\widetilde{L} - L}}{\norm{L}}\le \half$, and get the forward error estimate by the sensitivity of the linear system stated in \cref{lemma: pert bound}.
\begin{equation}
    \frac{\norm{\hbga-\bga}}{\norm{\bga}} \le  2\norm{L}\norm{L^{-1}}\left(0+\frac{\norm{\widetilde{L} - L}}{\norm{L}}\right) = \Or(n^3\eta^{-5/2}\ema),
\end{equation}
where we used the estimates for $\norm{L}$ and $\norm{L^{-1}}$ from \cref{lemma: norm bounds} in the last step.
\end{proof}

\subsection{Proof of stability for the inverse nonlinear FFT algorithm}
\label{subsec: proof stab fast alg}
In this subsection, we prove that the inverse nonlinear FFT is also forward stable in the following \cref{thm: stab fast alg}. Similar to \cref{thm: forward stable layer stripping}, we will show that forward stability reduces to proving the backward stability of solving the linear system in \cref{eq: gamma linear system}, together with a bound on the condition number of $L$. Since the proof of backward stability forms the essential part of the argument, we state it separately as \cref{lemma: backward stab fast alg}. The proof of \cref{lemma: backward stab fast alg} invokes several technical lemmas, which are stated and proved at the end of this section.

\begin{theorem}
\label{thm: stab fast alg}
Assume $(a_0,b_0)\in\mathcal{S}_\eta$, and $a_0^*(z)$ is outer. Then the inverse nonlinear FFT algorithm is numerically stable in the sense that, to achieve a target precision $\epsilon$, it suffices to use 
$$\Or(\log^2 n + \log n \log \eta^{-1} + \log\epsilon^{-1})$$
bits of precision when implementing the algorithm with floating-point arithmetic.
\end{theorem}
\begin{proof}
    In \cref{lemma: backward stab fast alg}, we define a function $h(n,\eta) = e^{\mo{\log^2 n + \log n\log\eta^{-1}}}$ and have the perturbation bounds $\norm{E}\le h(n,\eta)\ema\norm{L}$ and $\norm{\bq}\le h(n,\eta)\ema\norm{\bb_0}$ in \cref{eq: pert gamma linear system}. Using the sensitivity of the linear system in \cref{lemma: pert bound}, we obtain
    \begin{equation}\label{eq: gamma pert 2}
    \begin{aligned}
        \frac{\norm{\hbga-\bga}}{\norm{\bga}} &\le 
        2 \norm{L^{-1}}\norm{L} \left(\frac{\norm{\bq}}{\norm{\bb_0}} + \frac{\norm{E}}{\norm{L}}\right) = \Or(\eta^{-1} h(n,\eta)\ema)
    \end{aligned}\end{equation}
    where we used the estimates in \cref{lemma: norm bounds} in the last step. Therefore, in order to make the right hand side of \cref{eq: gamma pert 2} less than $\epsilon$, it suffices to choose $\ema = \epsilon e^{-\mo{\log^2 n + \log n\log\eta^{-1}}}$, which means the bit requirement is $r = \log_2\ema^{-1} = \Or(\log^2 n + \log n \log \eta^{-1} + \log\epsilon^{-1})$. Note that the condition \cref{eq: C0 assumption} is also satisfied under this choice of $\ema$.
\end{proof}

In the following context, we use $C$ with subscripts to represent universal constants independent of $n$, $\eta$, $\epsilon$, and $\ema$. We use $\rev(\bv)$ to represent the reverse order of a vector $\bv$, i.e. $\rev([v_0,v_1,\ldots,v_{n-1}]^T) = [v_{n-1},\ldots,v_1,v_0]^T$. 

\begin{lemma}
\label{lemma: backward stab fast alg}
Under the assumption of \cref{thm: stab fast alg}, there exist some matrix $E$ and vector $\bq$, such that 
    \begin{equation}\label{eq: pert gamma linear system}
        (L+E)\hat{\boldsymbol{\gamma}} = \frac{\bb_0 + \bq}{a_{00}}
    \end{equation}
    with $\norm{E}\le h(n,\eta)\ema\norm{L}$ and $\norm{\bq}\le h(n,\eta)\ema\norm{\bb_0}$ for a function 
    \begin{equation}\label{eq: growth function h(n)}
        h(n,\eta) = e^{\mo{\log^2 n + \log n\log\eta^{-1}}}
    \end{equation} 
    under the additional assumption that $\ema$ satisfies
    \begin{equation}\label{eq: C0 assumption}
        C_0n^3\eta^{-3}h(n,\eta)\ema\le 1,
    \end{equation} where $C_0$ is a universal constant that will be specified in the proof.
\end{lemma}
\begin{proof}
Without loss of generality, we assume $n=2^t$ is a power of 2. Otherwise, we may pad $\ba_0$ and $\bb_0$  with zeros. We will perform mathematical induction on $n$ to prove this theorem, and we abbreviate $h(n,\eta)$ as $h(n)$ since $\eta$ is viewed as a fixed number throughout the proof. The theorem holds for $n=1$ with $h(1)=1$. Next, we will derive a recurrence formula between $h(n)$ and $h(\frac{n}{2})$, and prove it grows no faster than \cref{eq: growth function h(n)}. 

Denote $m = \frac{n}{2}$. Let $\lup$ be the upper left $m\times m$ block of $L$. Recall that
\begin{equation}
    G_m := (\ba_m, \bb_m) = \begin{pmatrix}
{a}_{0,m} & {b}_{0,m} \\
{a}_{1,m} & {b}_{1,m} \\
\vdots & \vdots \\
{a}_{n-1-m,m} & {b}_{n-1-m,m} \\
\end{pmatrix},
\end{equation}
and let $\hbam, \hbbm, \hgm$ be the corresponding quantities calculated in the floating point arithmetic. Define the matrices $\kdown := T(\ba_m)T(\ba_m)^* + T(\bb_m)T(\bb_m)^*$ and $\widetilde{K}_{\mathrm{down}} := T(\hbam)T(\hbam)^* + T(\hbbm)T(\hbbm)^*$. Their $LDL^*$ factorizations are $\kdown = \ldown\ddown\ldown^*$ and $\widetilde{L}_{\mathrm{down}}\widetilde{D}_{\mathrm{down}}\widetilde{L}_{\mathrm{down}}^*$, respectively. Here we use the notation $\tilde{\cdot}$ rather than $\hat{\cdot}$ for some matrices because we never actually calculate these matrix elements during the algorithm, and they are just constructed for our proof.

Denote $\yup = (b_{0,0},\ldots, b_{m-1,0})^T$ and $\ydown = (b_{0,m},\ldots, b_{n-1,0})^T$, $\gup = (\gamma_0,\ldots, \gamma_{m-1})^T$, and $\gdo = (\gamma_m,\ldots, \gamma_{n-1})^T$.

We have the induction hypothesis
\begin{equation}\label{eq: induction hyp 1}
    (\lup+ E_1) \hgup = \frac{\yup + \bq_1}{a_{00}} \text{ with $\norm{E_1}\le h(m)\ema \norm{\lup}$, and $\norm{\bq_1}\le h(m)\ema \norm{\yup}$},
\end{equation}
\begin{equation}\label{eq: induction hyp 2}
    (\widetilde{L}_{\mathrm{down}}+ E_2) \hgdo = \frac{\hbbm + \bq_2}{\hat{a}_{0m}}\text{ with $\norm{E_2}\le h(m)\ema \norm{\widetilde{L}_{\mathrm{down}}}$, and $\norm{\bq_2}\le h(m)\ema \norm{\hbbm}$}.
\end{equation}
From the sensitive analysis of the linear system,
\begin{equation}\label{eq: sensitive up}
    \norm{\hgup-\gup}\le C_2 \kappa(\lup)h(m)\ema \norm{\hgup}\le C_2 \eta^{-1}h(m)\ema \norm{\hgup},
\end{equation}
using the submatrix condition number bound from \cref{lemma: norm bounds}.

The matrix $L$ can be written as a block matrix
\begin{equation}
    L = \begin{pmatrix}
        \lup & \\
        B & \ldown
    \end{pmatrix}.
\end{equation}
Here, the lower right corner of it coincides with $\ldown$, which can be seen from the process of layer stripping. A more precise explanation is as follows. Let $K_j$ be the $j$-th Schur complement of $K$. We can follow the Gaussian elimination process and inductively prove that 
\begin{equation}\label{eq: Kj}
        K_j-Z_{n-j} K_j Z_{n-j}^* = \ba_j\ba_j^*+ \bb_j\bb_j^*,
\end{equation}
similar to \cref{eq: K2 displacement}. We may check that $T(\ba_j)T(\ba_j)^* + T(\bb_j)T(\bb_j)^*$ satisfy \cref{eq: Kj} and the map $X\mapsto X-ZXZ^*$ is invertible, which confirms $K_j = T(\ba_j)T(\ba_j)^* + T(\bb_j)T(\bb_j)^*$. This means $\kdown = K_m$ is the $m$-th Schur complement of $K$. By the definition of Cholesky factorization, we have $\kdown = L_m D_m L_m^*$, where $L_m$ and $D_m$ is the right-bottom $(n-m)\times (n-m)$ block of $L$ and $D$ respectively.

By \cref{lemma: L gamma = b/a00} and the block structure of $L$, we obtain 
$$\frac{1}{a_{0m}}\bb_m = \ldown \gdo = \frac{1}{a_{00}}\ydown - B\gup.$$
Therefore,
\begin{equation}\label{eq: b_m/a_0m error 1}
    \begin{aligned}
        \frac{1}{\hat{a}_{0m}}\hbbm &= \frac{1}{a_{0m}}\bb_m + \left(\frac{1}{\hat{a}_{0m}}\hbbm - \frac{1}{a_{0m}}\bb_m\right) \\
        &= \frac{1}{a_{00}}\ydown - B\gup + \left(\frac{1}{\hat{a}_{0m}}\hbbm - \frac{1}{a_{0m}}\bb_m\right)\\
        &= \frac{1}{a_{00}}\ydown - B\hgup +B(\hgup-\gup)+ \frac{1}{\hat{a}_{0m}}\left(\hbbm - \bb_m\right) + \left(\frac{1}{\hat{a}_{0m}}- \frac{1}{a_{0m}}\right)\bb_m 
    \end{aligned}
\end{equation}
Next, we estimate the three error terms. Note that the proof of some detailed error bounds is postponed to the subsequent lemmas.
\begin{enumerate}[(i)]
    \item First, we estimate $B(\hgup-\gup)$. From \cref{eq: sensitive up}, we obtain
$$\norm{B(\hgup-\gup)}\le \norm{B}\norm{\hgup-\gup}\le \norm{L}\cdot C_2 \eta^{-1}h(m)\ema\norm{\hgup}\le 2C_2 \eta^{-5/2}h(m)\ema\norm{\bb_0},$$
where we used \cref{eq: L norm bound} and \cref{eq: hgup bounded by b0} in the last inequality.

\item Next, we estimate $\left(\frac{1}{\hat{a}_{0m}}- \frac{1}{a_{0m}}\right)\bb_m $, starting with
\begin{equation}\label{eq: b_m estimate}
    \norm{\bb_m} = \norm{a_{0m}\ldown\gdo} \le \abs{a_{0m}}\norm{\ldown}\norm{\gdo} \le 1\cdot\norm{L}\norm{\boldsymbol{\gamma}} \le \norm{L}\norm{L^{-1}}\norm{\frac{\bb_0}{a_{00}}}\le \eta^{-1}\norm{\frac{\bb_0}{a_{00}}}.
\end{equation}
We have 
\begin{equation}\label{eq: est a_0m}
    a_{0m}\ge a_{00}\ge \sqrt{1-(1-\eta)^2}\ge \eta^{\half}.
\end{equation}
From \cref{eq: hbxm - bxm estimate}, we get $\norm{\ba_m-\hbam}\le \half \eta^{\half}$, which is valid if we choose the constant $C_0\ge 2(C_5+C_6)$ in \cref{eq: C0 assumption}. Therefore $\hat{a}_{0m}\ge \half \eta^{\half}$, which lead to the estimation
\begin{equation}
\begin{aligned}
    \norm{\left(\frac{1}{\hat{a}_{0m}}- \frac{1}{a_{0m}}\right)\bb_m} &\le \frac{2}{\eta}\norm{\ba_m-\hbam}\norm{\bb_m}\le 2\eta^{-2}\norm{\ba_m-\hbam}\norm{\frac{\bb_0}{a_{00}}}\\
    &\le 2\eta^{-2}\left(C_5 m^2\sqrt{m}\eta^{-1} + C_6\sqrt{m}\eta^{-2}h(m)\right)\ema\norm{\frac{\bb_0}{a_{00}}},
\end{aligned}
\end{equation}
where we use the estimations from \cref{lemma: hbxm and hbym error estimate}.

\item Finally, according to \cref{eq: hbym - bym estimate}, we obtain
$$\norm{\frac{1}{\hat{a}_{0m}}\left(\hbbm - \bb_m\right)}\le 2\eta^{-1}\norm{\hbbm - \bb_m}\le 2\eta^{-1}\left(C_5 m^2\sqrt{m}\eta^{-1} + C_6\sqrt{m}\eta^{-2}h(m)\right)\norm{\bb_0}\ema.$$
\end{enumerate}

Combining these three error estimations, we can rewrite \cref{eq: b_m/a_0m error 1} as 
\begin{equation}\label{eq: hym/a0m final estimation}
    \frac{1}{\hat{a}_{0m}}\hbbm = \frac{1}{a_{00}}\ydown - B\hgup + \frac{1}{a_{00}}\bs,
\end{equation}
with
\begin{equation}\label{eq: bs bound}
    \norm{\bs}\le 4\eta^{-2}\left(C_5 m^2\sqrt{m}\eta^{-1} + (C_6+C_2)\sqrt{m}\eta^{-2}h(m)\right)\ema\norm{\bb_0},
\end{equation}
where we used the fact that $\abs{a_{00}}\le 1$.

Combining \cref{eq: induction hyp 1}, \cref{eq: induction hyp 2}, and \cref{eq: hym/a0m final estimation}, we get
\begin{equation}\label{eq: pert 2 gamma system}
    \begin{pmatrix}
        \lup+E_1&\\B &\widetilde{L}_{\mathrm{down}}+E_2
    \end{pmatrix}\begin{pmatrix}
        \hgup\\ \hgdo
    \end{pmatrix} = \frac{1}{a_{00}}\begin{pmatrix}
        \yup \\ \ydown 
    \end{pmatrix} + \frac{1}{a_{00}}\begin{pmatrix}
        \bq_1 \\ \bq_2+\bs
    \end{pmatrix},
\end{equation}
which is of the form \cref{eq: pert gamma linear system}.

According to \cref{lemma: L_down error bound}, we have
\begin{equation}
\begin{aligned}
    \norm{\begin{pmatrix}
        \lup+E_1&\\B &\widetilde{L}_{\mathrm{down}}+E_2
    \end{pmatrix} - L} &\le \norm{E_1}+\norm{E_2}+\norm{\ldown-\widetilde{{\ldown}}}\\
    &\le \left(2h(m) + 4C_1m\eta^{-\frac{3}{2}}\left(C_5 m^2\sqrt{m}\eta^{-1} + C_6\sqrt{m}\eta^{-2}h(m)\right)\right)\ema\norm{L}.
\end{aligned}
\end{equation}

Next, we bound the error term $\bq = \left( \begin{smallmatrix}
        \bq_1 \\ \bq_2+\bs
    \end{smallmatrix}\right)$.
Using the closeness result \cref{eq: hbym - bym estimate} and assumption \cref{eq: C0 assumption}, we obtain
\begin{equation}
    \norm{\hbbm-\bb_m}\le \left(C_5 + C_6\right)m^3\eta^{-2}h(m)\ema\norm{\bb_0}\le \norm{\bb_0}
\end{equation}
as long as we choose $C_0\ge C_5+C_6$. Using \cref{eq: b_m estimate} and \cref{eq: est a_0m}, we have
\begin{equation}
    \norm{\bb_m} \le \eta^{-1}\norm{\frac{\bb_0}{a_{00}}} \le \eta^{-3/2}\norm{\bb_0}.
\end{equation}
From the induction hypothesis \cref{eq: induction hyp 2}, we get
\begin{equation}\label{eq: bq2 bound}
    \norm{\bq_2}\le h(m)\ema \norm{\hbbm} \le h(m)\ema \left(\norm{\bb_m} + \norm{\hbbm-\bb_m}\right)\le 2h(m)\ema\eta^{-3/2}\norm{\bb_0}
\end{equation}
Combining \cref{eq: induction hyp 1}, \cref{eq: bs bound} and \cref{eq: bq2 bound}, we obtain
\begin{equation}
    \norm{\begin{pmatrix}
        \bq_1 \\ \bq_2+\bs
    \end{pmatrix}}\le \left(h(m) + 2h(m)\eta^{-3/2} + 4\eta^{-2}\left(C_5 m^2\sqrt{m}\eta^{-1} + (C_6+C_2)\sqrt{m}\eta^{-2}h(m)\right)\right)\ema\norm{\bb_0}.
\end{equation}
Therefore, by comparing \cref{eq: pert 2 gamma system} and \cref{eq: pert gamma linear system}, we only need $h(n)$ satisfy
\begin{equation}
    h(n) = h(2m) = C\cdot\left(m^{3/2}\eta^{-4} h(m) + m^{7/2}\eta^{-3}\right)
\end{equation}
for some constant $C$ to proceed with the induction step. Starting from $h(1) = 1$, we can solve 
\begin{equation}
    h(2^t) = C^t 2^{3t(t-1)/4}\eta^{-4t}\left(1+\eta^{-3}\sum_{k=1}^t C^{-k+1} 2^{-3k(k-1)/4}\eta^{4k}2^{7(k-1)/2}\right).
\end{equation}
Thus, we conclude 
\begin{equation}
    h(n) \le e^{\mo{\log^2 n + \log n \log \eta^{-1}}}.
\end{equation}
\end{proof}

Below, we list the technical lemmas invoked in the proof of \cref{lemma: backward stab fast alg}.
\begin{lemma}\label{lemma: fast NLFT error estimate}
    Given the first $n$ NLFT coefficient $\{\gamma_k\}_{k=0}^{n-1}$, the polynomials $\xi_n(z) = \xi_{0\to n}(z)$ and $\eta_n(z) = \eta_{0\to n}(z)$ defined in \cref{eq: defi of xi and eta} is determined, and we denote $\vxi_n$ and $\veta_n$ as their coefficient vectors. Denote $\cxi_n$ and $\ceta_n$ as the coefficient vectors calculated using \cref{eq: xi and eta update} recursively under floating-point arithmetic. We have the error bounds
    \begin{equation}
        \norm{\ceta_n-\veta_n}\le g(n)\ema,
    \end{equation}
    \begin{equation}
        \norm{\cxi_n-\vxi_n}\le g(n)\ema \sum_{k=0}^{n-1}|\gamma_k|,
    \end{equation}
    \begin{equation}\label{eq: cxi_n bound}
        \norm{\cxi_n}\le \sum_{k=0}^{n-1}|\gamma_k|
    \end{equation}
    for a function $g(n) \le C_3 n^2$.
\end{lemma}
\begin{proof}
Without loss of generality, we assume $n=2^t$ is a power of 2. Otherwise, we may pad $\gamma_k$ with zeros. We will use mathematical induction to prove this lemma. When $n=1$, the lemma holds since the error is 0, and $|\xi_1| = \frac{|\gamma_0|}{\sqrt{1+|\gamma_0|^2}}\le |\gamma_0|$. 

For an induction step, assume the lemma holds for $m = \frac{n}{2} = 2^{t-1}$. Recall that the polynomials $\xi_n$ and $\eta_n$ are calculated using \cref{eq: xi and eta update}. By definition, the coefficient vector of $\eta_m^{\sharp}(z)$ is the conjugate of $\rev(\veta_m)$, therefore the numeric routines of calculating $\veta_n$ and $\vxi_n$ are
\begin{equation}
    \veta_n = \ifft\left(-\fft(\overline{\rev(\vxi_m)})\cdot \fft(\vxi_{m\to n}) + \fft(\veta_{m})\cdot \fft(\veta_{m\to n})\right),
\end{equation}
\begin{equation}
    \vxi_n = \ifft\left(\fft(\overline{\rev(\veta_m)})\cdot \fft(\vxi_{m\to n}) + \fft(\vxi_{m})\cdot \fft(\veta_{m\to n})\right),
\end{equation}
where $\cdot$ means the element-wise product of vectors. As we have $\abs{\xi_n(z)}^2+\abs{\eta_n(z)}^2 = 1$ for any $z\in\TT$, the Parseval's identity gives $\norm{\vxi_n}^2+\norm{\veta_n}^2=1$. For convenience, we assume we are always rounding towards 0 in the floating point arithmetic, and thus $\norm{\cxi_n}\le 1$ and $\norm{\ceta_n}\le 1$ always hold. We first establish \cref{eq: cxi_n bound} by
\begin{equation}
    \norm{\cxi_n}\le \norm{\ceta_m}\norm{\cxi_{m\to n}} + \norm{\cxi_m}\norm{\ceta_{m\to n}}\le \sum_{k=0}^{m-1}|\gamma_k|+\sum_{k=m}^{n-1}|\gamma_k| = \sum_{k=0}^{n-1}|\gamma_k|
\end{equation}
where we used the inequality $\norm{\bu\cdot\bv}\le\norm{\bu}\norm{\bv}$ in the first step, and the induction hypothesis in the second step.

The FFT of a $2^t$-long vector has relative error $\Or(t\ema)$ \cite[Theorem 24.2]{higham2002accuracy}, that is 
\begin{equation}
    \norm{\widehat{\fft(\bv)}-\fft(\bv)}\le C_4t\ema\norm{\fft(\bv)}, \quad \text{for any } \bv\in\CC^{2^t}.
\end{equation}
Adding the FFT error to the error accumulated from previous steps, we get
\begin{equation}\label{eq: fft error estimate 1}
\begin{aligned}
    &\norm{\ceta_n-\veta_n}\\
    \le&\left(\norm{\cxi_m-\vxi_m}\norm{\cxi_{m\to n}} + \norm{\cxi_{m\to n}-\vxi_{m\to n}}\norm{\cxi_m} + \norm{\ceta_m-\veta_m}\norm{\ceta_{m\to n}} + \norm{\ceta_{m\to n}-\veta_{m\to n}}\norm{\ceta_m}\right) \\
    &+ C_4 t\ema \left(\norm{\cxi_m}\norm{\cxi_{m\to n}} + \norm{\cxi_{m\to n}}\norm{\cxi_m} + \norm{\ceta_m}\norm{\ceta_{m\to n}} + \norm{\ceta_{m\to n}}\norm{\ceta_m} + \norm{\ceta_n}\right)\\
    \le{}& 4g(m)\ema+5C_4 t\ema,
\end{aligned}
\end{equation}
and
\begin{equation}
\begin{aligned}
    &\norm{\cxi_n-\vxi_n} \\
    \le&\left(\norm{\ceta_m-\veta_m}\norm{\cxi_{m\to n}} + \norm{\ceta_{m\to n}-\veta_{m\to n}}\norm{\cxi_m} + \norm{\cxi_m-\vxi_m}\norm{\ceta_{m\to n}} + \norm{\cxi_{m\to n}-\vxi_{m\to n}}\norm{\ceta_m}\right) \\
    &+ C_4 t\ema \left(\norm{\ceta_m}\norm{\cxi_{m\to n}} + \norm{\ceta_{m\to n}}\norm{\cxi_m} + \norm{\cxi_m}\norm{\ceta_{m\to n}} + \norm{\cxi_{m\to n}}\norm{\ceta_m} + \norm{\cxi_n}\right)\\
    \le{}& 2\left(g(m)\ema\sum_{k=0}^{m-1}|\gamma_k| + g(m)\ema\sum_{k=m}^{n-1}|\gamma_k|\right)+C_4 t\ema\left(2\sum_{k=0}^{m-1}|\gamma_k|+2\sum_{k=m}^{n-1}|\gamma_k| + \sum_{k=0}^{n-1}|\gamma_k|\right)\\
    ={}& (2g(m)\ema+3C_4 t\ema)\sum_{k=0}^{n-1}|\gamma_k|.
\end{aligned}
\end{equation}
Therefore, we may let $g(2^t) = 4g(2^{t-1})+5 C_4 t$ with $g(2^0) = 0$, which is solved by 
$$g(n) = g(2^t) = 5C_4\sum_{k=1}^{t-1}4^k(t-k) = \frac{5C_4}{9}(4^{t+1}-3t-4) \le \frac{20C_4}{9} n^2.$$

\end{proof}

\begin{lemma}\label{lemma: hbxm and hbym error estimate}
    Under the assumptions of \cref{lemma: backward stab fast alg}, we have the following estimates
    \begin{equation}\label{eq: hbym - bym estimate}
        \norm{\hbbm-\bb_m}\le \left(C_5 m^2\sqrt{m}\eta^{-1} + C_6\sqrt{m}\eta^{-2}h(m)\right)\norm{\bb_0}\ema
    \end{equation}
    \begin{equation}\label{eq: hbxm - bxm estimate}
        \norm{\hbam-\ba_m}\le \left(C_5 m^2\sqrt{m}\eta^{-1} + C_6\sqrt{m}\eta^{-2}h(m)\right)\ema
    \end{equation}
    \begin{equation}\label{eq: hgm - gm estimate}
        \norm{\hgm-G_m}_F\le 2\left(C_5 m^2\sqrt{m}\eta^{-1} + C_6\sqrt{m}\eta^{-2}h(m)\right)\ema
    \end{equation}
    for some constants $C_5$ and $C_6$.
\end{lemma}
\begin{proof}
    We also assume $n=2^t$ is a power of 2 as in \cref{lemma: fast NLFT error estimate}. Let $\vxi_m$ and $\veta_m$ be the accurate coefficient vectors obtained from $\{\gamma_k\}_{k=0}^{m-1}$, $\txi_m$ and $\teta_m$ be the accurate coefficient vectors obtained from $\{\hat{\gamma}_k\}_{k=0}^{m-1}$, and $\hai_m$ and $\heta_m$ be the coefficient vectors calculated using floating point arithmetic from $\{\hat{\gamma}_k\}_{k=0}^{m-1}$. We also define the corresponding polynomials $\tilde{\xi}_m(z)$, $\tilde{\eta}_m(z)$, $\hat{\xi}_m(z)$, and $\hat{\eta}_m(z)$.

    In \cref{lemma: fast NLFT error estimate}, we may treat $n$ and $\gamma$ as dummy variables. Applying the lemma to the sequence $\{\hat{\gamma}_k\}_{k=0}^{m-1}$, then the conclusion becomes 
    \begin{equation}\label{eq: heta_m and teta_m diff bound}
        \norm{\heta_m-\teta_m}\le C_3 m^2\ema,
    \end{equation}
    \begin{equation}\label{eq: hxi_m and txi_m diff bound}
        \norm{\hai_m-\txi_m}\le C_3 m^2\ema \sum_{k=0}^{m-1}|\hat{\gamma}_k|=C_3 m^2\ema\norm{\hgup}_1\le C_3 m^2\ema\sqrt{m}\norm{\hgup},
    \end{equation}
    \begin{equation}
        \norm{\hai_m}\le \sum_{k=0}^{m-1}|\hat{\gamma}_k|=\norm{\hgup}_1\le \sqrt{m}\norm{\hgup}.
    \end{equation}
    For any $z\in\TT$, we have
    \begin{equation}
    \begin{aligned}
        &\norm{\begin{pmatrix}
        \eta_{m}^*(z) & \xi_{m}(z)\\ 
        -\xi_{m}^*(z) & \eta_{m}(z)
    \end{pmatrix} - \begin{pmatrix}
        \tilde{\eta}_{m}^*(z) & \tilde{\xi}_{m}(z)\\ 
        -\tilde{\xi}_{m}^*(z) & \tilde{\eta}_{m}(z)
    \end{pmatrix}} \\
    ={}& \norm{\prod_{k=0}^{m-1}\left[\frac{1}{\sqrt{1+|\gamma_k|^2}}\begin{pmatrix}
        1 & \gamma_k z^k \\ 
        -\overline{\gamma_k}z^{-k} & 1
    \end{pmatrix}\right] - \prod_{k=0}^{m-1}\left[\frac{1}{\sqrt{1+|\gamma_k|^2}}\begin{pmatrix}
        1 & \tilde{\gamma}_k z^k \\ 
        -\overline{\tilde{\gamma}_k}z^{-k} & 1
    \end{pmatrix}\right]}\\
    \le{}& \sum_{k=0}^{m-1}\norm{\left[\frac{1}{\sqrt{1+|\gamma_k|^2}}\begin{pmatrix}
        1 & \gamma_k z^k \\ 
        -\overline{\gamma_k}z^{-k} & 1
    \end{pmatrix}\right] - \left[\frac{1}{\sqrt{1+|\gamma_k|^2}}\begin{pmatrix}
        1 & \tilde{\gamma}_k z^k \\ 
        -\overline{\tilde{\gamma}_k}z^{-k} & 1
    \end{pmatrix}\right]}\\
    \le{}& \sum_{k=0}^{m-1}\abs{\gamma_k-\tilde{\gamma}_k} = \norm{\gup-\hgup}_1\\
    \le{}& \sqrt{m}\norm{\gup-\hgup}_2 \le C_2 \sqrt{m}\eta^{-1}h(m)\ema \norm{\hgup},
    \end{aligned}
    \end{equation}
    where the first inequality holds because the error of the product of unitary matrices is bounded by the sum of errors of each factor, and the last step uses \cref{eq: sensitive up}. Therefore, by Parseval's identity, we get
    \begin{equation}
        \norm{\txi_m-\vxi_m} = \norm{\tilde{\xi}_m(z)-\xi_m(z)}_{L^2(\TT)}\le \max_{z\in\TT}\abs{\tilde{\xi}_m(z)-\xi_m(z)} \le C_2 \sqrt{m}\eta^{-1}h(m)\ema \norm{\hgup}.
    \end{equation}
    The bound 
    \begin{equation}
        \norm{\teta_m-\veta_m} \le C_2 \sqrt{m}\eta^{-1}h(m)\ema \norm{\hgup}
    \end{equation}
    also hold for the same reason. Adding these to \cref{eq: hxi_m and txi_m diff bound} and \cref{eq: heta_m and teta_m diff bound}, we obtain
    \begin{equation}\label{eq: hxi_m and xi_m difference}
        \norm{\hai_m-\vxi_m} \le \left(C_3 m^2\sqrt{m}\ema + C_2 \sqrt{m}\eta^{-1}h(m)\ema \right)\norm{\hgup}
    \end{equation}
    \begin{equation}\label{eq: heta_m and eta_m difference}
        \norm{\heta_m-\veta_m}\le C_3 m^2\ema + C_2 \sqrt{m}\eta^{-1}h(m)\ema \norm{\hgup}
    \end{equation}

    According to \cref{eq: xm and ym update}, the numerical routine of calculating $\ba_m$ and $\bb_m$ is
    \begin{equation}
        \ba_m = \ifft\left(\fft(\overline{\rev(\veta_m)})\cdot \fft(\ba_0) + \fft(\overline{\rev(\vxi_m)})\cdot \fft(\bb_0)\right)
    \end{equation}
    \begin{equation}
        \bb_m = \ifft\left(-\fft(\vxi_{m})\cdot \fft(\ba_0) + \fft(\veta_{m})\cdot \fft(\bb_0)\right)
    \end{equation}
    Similar to \cref{eq: fft error estimate 1}, we have the estimate
    \begin{equation}\label{eq: hbxm intermeidate}
    \begin{aligned}
    &\norm{\hbam-\ba_m}\\
    \le&\left(\norm{\heta_m-\veta_m}\norm{\ba_0} + \norm{\hai_{m}-\vxi_{m}}\norm{\bb_0} \right) \\
    &+ C_4 t\ema \left(\norm{\heta_m}\norm{\ba_0}+\norm{\heta_m}\norm{\ba_0}+\norm{\hai_m}\norm{\bb_0}+\norm{\hai_m}\norm{\bb_0} + \norm{\hbam}\right)\\
    \le{}& \left(C_3 m^2\sqrt{m} + 2C_2 \sqrt{m}\eta^{-1}h(m) \right)\norm{\hgup}\ema+C_3 m^2\ema+5C_4 t\ema,
    \end{aligned}
    \end{equation}
    where we used \cref{eq: hxi_m and xi_m difference} and \cref{eq: heta_m and eta_m difference} in the last step, along with the fact that $\norm{\ba_0}$, $\norm{\bb_0}$, $\norm{\hbam}$, $\norm{\heta_m}$, and $\norm{\hai_m}$ are all bounded by 1. 
    
    Similarly, we may also estimate
    \begin{equation}\label{eq: hym bound}
        \norm{\hbbm}\le \norm{\hai_m}\norm{\ba_0}+\norm{\heta_m}\norm{\bb_0}\le \sqrt{m}\norm{\hgup}+\norm{\bb_0}
    \end{equation}
    and
    \begin{equation}\label{eq: hbym intermeidate}
    \begin{aligned}
    &\norm{\hbbm-\bb_m}\\
    \le&\left(\norm{\hai_m-\vxi_m}\norm{\ba_0} + \norm{\heta_{m}-\veta_{m}}\norm{\bb_0} \right) \\
    &+ C_4 t\ema \left(\norm{\hai_m}\norm{\ba_0}+\norm{\hai_m}\norm{\ba_0}+\norm{\heta_m}\norm{\bb_0}+\norm{\heta_m}\norm{\bb_0} + \norm{\hbbm}\right)\\
    \le{}& \left(C_3 m^2\sqrt{m}\ema + C_2 \sqrt{m}h(m)\ema \right)\norm{\hgup} + \left(C_3 m^2\ema + C_2 \sqrt{m}\eta^{-1}h(m)\ema \norm{\hgup}\right)\norm{\bb_0}\\
    &+3C_4 t\ema\left(\sqrt{m}\norm{\hgup}+\norm{\bb_0}\right)\\
    \le{}& \left(C_3 m^2\sqrt{m} + 2C_2 \sqrt{m}\eta^{-1}h(m)+3C_4 t\sqrt{m} \right)\norm{\hgup}\ema +\left(C_3 m^2 + 3C_4t\right)\norm{\bb_0}\ema,
    \end{aligned}
    \end{equation}
    The last step is to estimate $\norm{\hgup}$. We start by
    $$\norm{\gup} \le \norm{\boldsymbol{\gamma}} = \norm{L^{-1}\frac{\bb_0}{a_{00}}}\le \norm{L^{-1}}\abs{\frac{1}{a_{00}}}\norm{\bb_0}\le \eta^{-1}\norm{\bb_0}.$$
    The closeness result from \cref{eq: sensitive up} implies that $\norm{\hgup}\le 2\norm{\gup}$, where we use the assumption \cref{eq: C0 assumption} as long as we choose the constant $C_0\ge C_2$. Then we obtain
    \begin{equation}\label{eq: hgup bounded by b0}
        \norm{\hgup}\le 2\norm{\gup}\le  2\eta^{-1}\norm{\bb_0}\le 2\eta^{-1}.
    \end{equation}
    The first bound \cref{eq: hbym - bym estimate} follows by substituting $\norm{\hgup}$ with $2\eta^{-1}\norm{\bb_0}$ in \cref{eq: hbym intermeidate}, and the second bound \cref{eq: hbxm - bxm estimate} follows by substituting $\norm{\hgup}$ with $2\eta^{-1}$ in \cref{eq: hbxm intermeidate}.
    Finally, we use $\norm{\hgm-G_m}_F\le \norm{\hbam-\ba_m} + \norm{\hbbm-\bb_m}$ and $\norm{\bb_0}\le 1$ to conclude \cref{eq: hgm - gm estimate}.
\end{proof}

\begin{lemma}\label{lemma: L_down error bound}
    Under the assumptions of \cref{lemma: backward stab fast alg}, we have the following estimates
    \begin{equation}\label{eq: tilde ldown - ldown}
    \norm{\widetilde{L}_{\mathrm{down}}-\ldown} \le 4C_1m\eta^{-\frac{3}{2}}\left(C_5 m^2\sqrt{m}\eta^{-1} + C_6\sqrt{m}\eta^{-2}h(m)\right)\ema\norm{\ldown}.
\end{equation}
\end{lemma}
\begin{proof}
Define an operator $\mathcal{A}$ on the space of matrices $\CC^{m\times m}$ as 
$$\mathcal{A}: X \mapsto X - ZXZ^*,$$
and the operator norm associated with the matrix Frobenius norm as $$\norm{\mathcal{A}}_{\mathrm{op}(F)}:= \max_{\norm{X}_F = 1}\norm{\mathcal{A}(X)}_F.$$
If we flatten the matrix $X$ as a vector of length $m^2$, then this operator has a matrix form $\mathcal{A} = I - Z\otimes Z$. The Frobenius norm of $X$ is the same as the 2 norm of the flattened vector, so the induced operator norm of $\mathcal{A}^{-1}$ is $\norm{\mathcal{A}^{-1}}_{\mathrm{op}(F)} = \norm{(I-Z\otimes Z)^{-1}}_2 = \norm{\sum_{j=0}^{m-1}Z^j\otimes Z^j}_2 \le m$. This also proves the invertibility of $\mathcal{A}$.

By definition, we have $\kdown = \mathcal{A}^{-1}(G_mG_m^*)$ and $\widetilde{K}_{\mathrm{down}} = \mathcal{A}^{-1}(\hgm\hgm^*)$, so the following estimate holds.
$$\begin{aligned}
    \norm{\kdown-\widetilde{K}_{\mathrm{down}}}_F &\le \norm{\mathcal{A}^{-1}}_{\mathrm{op}(F)}\norm{G_mG_m^* - \hgm\hgm^*}_F \\
    &\le \norm{\mathcal{A}^{-1}}_{\mathrm{op}(F)}\left(\norm{G_m} + \norm{\hgm}_F\right)\norm{G_m - \hgm}_F\\
    &\le 2m\norm{G_m - \hgm}_F.
\end{aligned}$$
As $\ldown\cdot(\ddown\ldown^*)$ is the LU factorization of $\kdown$, we have the error estimation \cite[Theorem 9.15]{higham2002accuracy}
\begin{equation}
\begin{aligned}
    \norm{\widetilde{L}_{\mathrm{down}}-\ldown}_F&\le C_1\norm{\ldown}\norm{\ldown^{-1}}\norm{(\ddown\ldown^*)^{-1}}\frac{\norm{\kdown}}{\norm{\kdown}_F}\norm{\widetilde{K}_{\mathrm{down}}-\kdown}_F\\
    &\le C_1\norm{\ldown}\eta^{-\half}\frac{1}{\eta(2-\eta)}\cdot 1\cdot\norm{\widetilde{K}_{\mathrm{down}}-\kdown}_F\\
    &\le 2C_1m\eta^{-\frac{3}{2}}\norm{\ldown}\norm{G_m - \hgm}_F.
\end{aligned}
\end{equation}
Here we used the estimation of $\ldown^{-1}, (\ddown\ldown^*)^{-1}$ from \cref{lemma: norm bounds}, since they are the submatrices of the triangular matrices $L^{-1}$ and $(DL^*)^{-1}$, respectively. Finally, we use \cref{eq: hgm - gm estimate} to conclude this lemma.
\end{proof}

\section{Locally Lipschitz estimates for inverse NLFT}
\label{sec:lipschitz-bounds-nlft}

In this section, we derive some continuity estimates for NLFT and inverse NLFT for the compactly supported case in the $\mathrm{SU}(2)$ setting, which establishes continuity of the maps avoiding the use of infinite dimensional function spaces (cf. \cite[Theorem~2.5, Lemma~3.10]{tsai2005nlft}). In particular, we will show that the NLFT map is Lipschitz continuous, and the inverse NLFT map, while it fails to be uniformly continuous, is locally Lipschitz continuous. The nature of the locally Lipschitz estimate then allows us to identify conditions under which the inverse NLFT map is Lipschitz. In all cases, the domains and codomains of the maps will be endowed with the standard Euclidean topology.

We will suppose $\bga, \bga' \in \ell (0,n-1)$ with $\overbrace{\bga} := (a,b)$ and $\overbrace{\bga'} := (a',b')$, and we will also assume throughout that $n \ge 2$. Let us collect the coefficients of the Laurent polynomials $a, a', b, b'$ into vectors $\ba, \ba', \bb, \bb' \in \CC^n$ (refer to \cref{lem:nlft-ab-degree} why they are $n$ dimensional vectors). Thus, for example, we have $\ba := (a_0,\dots,a_{n-1})$ and $\bb := (b_0,\dots,b_{n-1})$, where $a = \sum_{k=0}^{n-1} a_k z^{-k}$ and $b = \sum_{k=0}^{n-1} b_k z^{k}$. Note that this choice deviates from the definitions of $a_k$ and $b_k$ in \cref{sec:inverse-NLFT}, but this streamlines the presentation in this section. Then the NLFT map $\bga \mapsto (\ba,\bb)$ can be regarded as a map from $\CC^\ast \times \CC^{n-2} \times \CC^\ast$
to $\CC^{n} \times \CC^{n}$. In the next lemma, we establish Lipschitz estimates for the NLFT map, where one merely thinks of the image of the NLFT map as elements of $\CC^{n} \times \CC^{n}$.  

\begin{lemma}
\label{lem:lp-lq-estimates-NLFT}
Let $p, q, r \ge 1$. Then the following estimate holds:
\begin{equation}
    n^{-\frac{1}{p}} \norm{\ba - \ba'}_p + n^{-\frac{1}{q}} \norm{\bb - \bb'}_q \leq 6 n^{\frac{1}{2} - \frac{1}{r}} \norm{\bga - \bga'}_r.
\end{equation}
Thus the NLFT map $\CC^\ast \times \CC^{n-1} \times \CC^\ast \ni \bga \mapsto (\ba, \bb) \in \CC^n \times \CC^n$ is Lipschitz continuous.
\end{lemma}

\begin{proof}
Let $\mathtt{T}(z), \mathtt{T}'(z)$ be the $\mathrm{SU}(2)$ matrices corresponding to $(a(z),b(z))$ and $(a'(z),b'(z))$ respectively, for $z \in \TT$. From \cite[Theorem~2.5]{tsai2005nlft}, we have $\norm{\mathtt{T}(z) - \mathtt{T}'(z)}_2 \leq 3 \norm{\bga - \bga'}_{1}$. We also have $\norm{\bga - \bga'}_1 \leq n^{1 - \frac{1}{r}} \norm{\bga - \bga'}_r$, by the power mean inequality. This establishes the result
\begin{equation}
    \sup_{z \in \TT} \norm{\mathtt{T}(z) - \mathtt{T}'(z)}_{2} \leq 3 n^{1 - \frac{1}{r}} \norm{\bga - \bga'}_r.
\end{equation}
So we only need to prove the direction
\begin{equation}
\label{eq:lp-lq-estimates-NLFT-1}
\sup_{z \in \TT} \norm{\mathtt{T}(z) - \mathtt{T}'(z)}_{2} \geq \frac{1}{2} \left( n^{\frac{1}{2}-\frac{1}{p}} \norm{\ba - \ba'}_p + n^{\frac{1}{2}-\frac{1}{q}} \norm{\bb - \bb'}_q \right).
\end{equation}
For this, we recall that for every $z \in \TT$, we have $\norm{\mathtt{T}(z) - \mathtt{T}'(z)}_{2}^2 \geq \frac{1}{2} \norm{\mathtt{T}(z) - \mathtt{T}'(z)}_{F}^2$. Next, noting that for any $z \in \TT$ and a Laurent polynomial $s$, we have $|\overline{s(z)}| = |s^\ast(z)|$, we may deduce that $\frac{1}{2} \norm{\mathtt{T}(z) - \mathtt{T}'(z)}_{F}^2 = |a(z) - a'(z)|^2 + |b(z) - b'(z)|^2$. Putting everything together, we have 
\begin{equation}
\begin{split}
    & \sup_{z \in \TT} \norm{\mathtt{T}(z) - \mathtt{T}'(z)}_{2}  \geq  \sup_{z \in \TT} \left( |a(z) - a'(z)|^2 +  |b(z) - b'(z)|^2 \right)^{\frac{1}{2}} \\
    & \geq \frac{1}{2}  \sup_{z \in \TT} |a(z) - a'(z)| +  \frac{1}{2} \sup_{z \in \TT} |b(z) - b'(z)| \geq \frac{1}{2} \norm{a - a'}_{L^2(\TT)} + \frac{1}{2} \norm{b - b'}_{L^2(\TT)} \\
    &  = \frac{1}{2} \left( \norm{\ba - \ba'}_2 + \norm{\bb - \bb'}_2 \right) \geq \frac{1}{2} \left( n^{\frac{1}{2}-\frac{1}{p}} \norm{\ba - \ba'}_p + n^{\frac{1}{2}-\frac{1}{q}} \norm{\bb - \bb'}_q \right),
\end{split}
\end{equation}
where the last inequality is again due to the power mean inequality.
\end{proof}

We next simplify the characterization of the image of the NLFT map, as given in \cref{lem:nlft-ab-degree,lem:nlft-bijection}, to the particular setting being discussed here for $\bga \in \ell(0,n-1)$. The image is a subset of $\CC^n \times \CC^n$, and it is given by
\begin{equation}
\label{eq:Sn-def}
    \mathtt{S}_n := \{(\ba,\bb) \in \CC^n \times \CC^n: \RR \ni a_0 > 0, \; \text{xcor}(\ba,\overline{\ba}) + \text{xcor}(\bb, \overline{\bb}) = [\underbrace{0,\dots,0}_{n-1},1,\underbrace{0,\dots,0}_{n-1}]  \; \},
\end{equation}
where the $\text{xcor}(\cdot,\cdot)$ function takes as input two length $n$ vectors and outputs their aperiodic crosscorrelation vector of length $2n-1$ (see \cite[Section~2.6]{proakis2001digital} for definition of crosscorrelation of finite length sequences). The condition $a_0 > 0$ is the consequence of requiring that $a(\infty) > 0$, while the second condition is equivalent to requiring $aa^\ast + bb^\ast = 1$ (see \cref{lem:nlft-bijection}). If we impose the standard Euclidean topology on $\mathtt{S}_n$, then $a_0 > 0$ implies that it is not complete as a metric space. In fact, this is not the only obstruction to completeness --- we also know that $a_{n-1}$, $b_0$ and $b_{n-1}$ are all non-zero by \cref{lem:nlft-ab-degree}, and so each of these is an obstruction to completeness as well. However, $\mathtt{S}_n$ is a bounded subset of $\CC^n \times \CC^n$, and this follows from the second condition in \cref{eq:Sn-def} which yields 
\begin{equation}
\label{eq:ab-norm}
    \norm{\ba}_2^2 + \norm{\bb}_2^2 = 1, \;\; \implies |a_i| \leq 1, \; |b_i| \leq 1, \; \forall i=0,\dots,n-1. 
\end{equation}
Notice that the condition $aa^\ast + bb^\ast=1$ in \cref{lem:nlft-bijection} implies that for all $z \in \TT$, we have $|a^\ast(z)|^2 + |b(z)|^2 = 1$, and so we may conclude that both $a^\ast$ and $b$ are Schur functions \cite{schur1918potenzreihen}, which also leads to the same boundedness conclusion, namely $|a_i| \leq 1$ and $|b_i| \leq 1$, by the property of Schur functions.

The inverse NLFT is a bijectve map from $\mathtt{S}_n \rightarrow \CC^\ast \times \CC^{n-2} \times \CC^\ast$. Unlike the NLFT, we will next give an example showing that the inverse NLFT is not uniformly continuous when one equips $\mathtt{S}_n$ with the Euclidean topology. In fact, such a statement already appears in the non-compactly supported setting in \cite[Page~31]{tsai2005nlft}; so the contribution of this example is to show that the additional assumption of compact support does not change this fact. For example, we consider the following: define the sequence $\{(\ba^{(k)}, \bb^{(k)}) \in \mathtt{S}_n:k = 2,3,\dots\}$ such that 
\begin{equation}
    \ba^{(k)} := \left( \frac{1}{k},0,\dots,0,\sqrt{\frac{1}{2} - \frac{1}{k^2}} \right), \; \bb^{(k)} := \left( \sqrt{\frac{1}{2} - \frac{1}{k^2}},0,\dots,0, -\frac{1}{k} \right).
\end{equation}
Let $\bga^{(k)}$ be the inverse NLFT of $(\ba^{(k)},\bb^{(k)})$, i.e. $\overbrace{\bga^{(k)}} = (\ba^{(k)}, \bb^{(k)})$. In this case the first component of $\bga^{(k)}$ equals $(k^2/2 - 1)^\frac{1}{2}$ by the layer stripping procedure (see \cref{ssec:layer-stripping}). Thus, we may conclude that for all $k \geq 2$, we have
\begin{equation}
    \norm{\bga^{(k+1)} - \bga^{(k)}}_2 \geq \sqrt{(k+1)^2 / 2 - 1} - \sqrt{k^2 / 2 - 1} \geq \frac{1}{\sqrt{2}}.
\end{equation}
On the other hand, we may also calculate for $k \ge 2$,
\begin{equation}
\begin{split}
    & \norm{\ba^{(k+1)} - \ba^{(k)}}_2^2 + \norm{\bb^{(k+1)} - \bb^{(k)}}_2^2 = 2 \left( \frac{1}{k^2 + k} \right)^2 + 2 \left( \sqrt{\frac{1}{2} - \frac{1}{(k+1)^2}} - \sqrt{\frac{1}{2} - \frac{1}{k^2}} \right)^2 \\
    & \leq \frac{2}{k^4} + 2 \left( \sqrt{\frac{1}{2} - \frac{1}{(k+1)^2}} - \sqrt{\frac{1}{2} - \frac{1}{k^2}} \right)^2 \leq \frac{2}{k^4} + \frac{8}{k^6} \leq \frac{10}{k^4},
\end{split}
\end{equation}
which goes to zero as $k \rightarrow 
\infty$. This completes the argument for the nonexistence of uniform continuity of inverse NLFT, even for compactly supported sequences. 

We now proceed to establish a locally Lipschitz estimate for the inverse NLFT map with respect to the Euclidean topology on both the domain and the codomain, thereby avoiding the use of non-Euclidean metrics and quasi-metrics used in the continuity proof in the infinite-dimensional case \cite[Lemma~3.10]{tsai2005nlft}. For example, in that proof, on the codomain of inverse NLFT, Tsai uses the symmetric quasi-metric $d(\bga, \bga'):= \sum_{k=0}^{n-1} \log (1 + |\gamma_k - \gamma'_k|^2)$, for $\bga, \bga' \in \CC^\ast \times \CC^{n-2} \times \CC^\ast$. To establish the locally Lipschitz estimate, we will need two helper lemmas, \cref{lem:b-over-a-bound,lem:bound-T-diff}, proved in \cref{sec:helper-lemmas}. Note that the inverse NLFT map is smooth; hence it must be the case that the map is locally Lipschitz, and so the main contribution of the following lemma is to provide an explicit form of the estimate, which also allows us to derive a Lipschitz continuity condition in \cref{cor:inv-NLFT-lipschitz}. In the proof below, we utilize several well-known results relating different matrix norms, which can be found in any standard numerical linear algebra textbook (cf. \cite{GolubVan2013}). 
\begin{lemma}[Locally Lipschitz estimate]
\label{lem:inv-NLFT-modulus-continuity}
Let $(\ba,\bb), (\ba',\bb') \in \mathtt{S}_n$, and suppose $\norm{\ba - \ba'}_{\infty} \leq \epsilon$, $\norm{\bb - \bb'}_{\infty} \leq \epsilon$, for some $\epsilon > 0$. Let $\bga, \bga'$ be such that $\overbrace{\bga} = (\ba,\bb)$ and $\overbrace{\bga'} = (\ba',\bb')$. Then we have the estimate
\begin{equation}
\label{eq:inv-NLFT-modulus-continuity-1}
\norm{\bga - \bga'}_1 <  \epsilon (3n)^{n} (1 + 1 / \delta)^{2n},
\end{equation}
where $\delta := \min \{a_0, a'_0\}$. Thus, the inverse NLFT is a locally Lipschitz bijective map from $\mathtt{S}_n$ to $\CC^\ast \times \CC^{n-2} \times \CC^\ast$, with respect to the Euclidean topology on both spaces.
\end{lemma}

\begin{proof}
In this proof, we will use the notation of \cref{ssec:layer-stripping}. Let the inverse NLFT of $(\ba,\bb)$ and $(\ba',\bb')$ be $\bga:= (\gamma_0,\dots,\gamma_{n-1})$ and $\bga':= \{\gamma'_0,\dots,\gamma'_{n-1}\}$ respectively. For $0 \leq k \leq n-1$, at the start of the $k^{\text{th}}$ step of the layer stripping algorithm for determining $\gamma_k$, let the $n-k$ dimensional vectors be $\ba_k := (a_{j,k})_{0\le j\le n-k-1}$ and $\bb_k := (b_{j,k})_{0\le j\le n-k-1}$, as defined in \cref{eq:Gk-layer-strip-def}. We initialize $(\ba_0, \bb_0) := (\overline{\ba}, \bb)$, and then iteratively obtain $(\ba_{k+1}, \bb_{k+1})$ from $(\ba_k, \bb_k)$ using \cref{eq:G_n recurrence,eq:G_n shift}. Analogously, for every $0 \leq k \leq n-1$, we define the $n-k$ dimensional vectors $\ba'_k, \bb'_k$ corresponding to the layer stripping process for obtaining $\bga'$, with the initialization $(\ba'_0, \bb'_0) := (\overline{\ba'}, \bb')$. Let us also define the $(n-k) \times 2$ matrices $G_k := (\ba_k, \bb_k)$ and $G'_k := (\ba'_k, \bb'_k)$. For a complex number $\gamma$, we will define the $\mathrm{SU}(2)$ matrix $\Theta(\gamma) :=  \frac{1}{\sqrt{1 + |\gamma|^2}} \left( \begin{smallmatrix}
    1 & -\gamma \\ \overline{\gamma} & 1
\end{smallmatrix} \right) $.

We recall from \cref{lem:layer-stripping-prop}(a) that $a_{0,0} \leq a_{0,1} \leq \dots \leq a_{0,n-1}$, and $a'_{0,0} \leq a'_{0,1} \leq \dots \leq a'_{0,n-1}$. Let $\delta := \min \{a_{0,0}, a'_{0,0}\} > 0$. Note that \cref{lem:layer-stripping-prop}(d) implies $\norm{\bb_k}_{\infty}, \norm{\bb'_k}_{\infty} \leq 1$, for every $k$. At step $k$ of layer stripping, let us assume that we have $\norm{\ba_k - \ba'_k}_{\infty} \le \epsilon_k$ and $\norm{\bb_k - \bb'_k}_{\infty} \le \epsilon_k$, for some $\epsilon_k > 0$ to be specified later. We may then additionally deduce that for every $k$, we have $\norm{G_k - G'_k}_1 \leq (n-k) \epsilon_k$, and $\norm{G'_k}_1 \le \sqrt{n-k} \norm{G'_k}_F = \sqrt{n-k}$, using the fact $\norm{G'_k}_F = 1$ by \cref{lem:layer-stripping-prop}(b). We will use these observations below.

Now at step $k$ of layer stripping, we first have $\gamma_k = \frac{b_{0,k}}{a_{0,k}}$, $\gamma'_k = \frac{b'_{0,k}}{a'_{0,k}}$, and thus by \cref{lem:b-over-a-bound} we get $|\gamma_k - \gamma'_k| \leq \delta^{-1} (1 + 1 / \delta) \epsilon_k$. Next, we estimate
\begin{equation}
\label{eq:inv-NLFT-modulus-continuity-2}
\begin{split}
    & \norm{G_k \Theta(\gamma_k) - G'_k \Theta(\gamma'_k)}_1 \leq \norm{G_k \Theta(\gamma_k) - G'_k \Theta(\gamma_k)}_1 + \norm{G'_k \Theta(\gamma_k) - G'_k \Theta(\gamma'_k)}_1 \\
    & \leq \norm{G_k - G'_k}_1 \norm{\Theta(\gamma_k)}_1 + \norm{G'_k}_1 \norm{\Theta(\gamma_k) - \Theta(\gamma'_k)}_1 < \sqrt{2} \norm{G_k - G'_k}_1  + 3 \norm{G'_k}_1 |\gamma_k - \gamma'_k| \\
    & \leq \sqrt{2} (n-k) \epsilon_k + 3 \delta^{-1} (1 + 1 / \delta) \epsilon_k \sqrt{n-k} < 3n (1 + 1 / \delta  + 1 / \delta^2) \epsilon_k, 
\end{split}
\end{equation}
where in the second line we have used $\norm{\Theta(\gamma_k)}_1 \leq \sqrt{2} \norm{\Theta(\gamma_k)}_2 = \sqrt{2}$ and \cref{lem:bound-T-diff}, and the third line follows from the observations in the last paragraph. Finally, recall that $\norm{G_{k+1} - G'_{k+1}}_1 = \norm{G_k \Theta(\gamma_k) - G'_k \Theta(\gamma'_k)}_1$ by \cref{lem:layer-stripping-prop}(c), and thus \cref{eq:inv-NLFT-modulus-continuity-2} allows us to conclude 
\begin{equation}
\label{eq:inv-NLFT-modulus-continuity-3}
    \max \{\norm{\ba_{k+1} - \ba'_{k+1}}_{\infty},  \norm{\bb_{k+1} - \bb'_{k+1}}_{\infty} \} \leq \norm{G_{k+1} - G'_{k+1}}_1 < 3n (1 + 1 / \delta  + 1 / \delta^2) \epsilon_k.
\end{equation}
We may now set $\epsilon_0 := \epsilon$, and then recursively set $\epsilon_{k+1} := 3n (1 + 1 / \delta  + 1 / \delta^2) \epsilon_k$, for $0 \leq k \leq n-2$. This allows us to estimate
\begin{equation}
\label{eq:inv-NLFT-estimate}
    \norm{\bga - \bga'}_1 = \sum_{k=0}^{n-1} |\gamma_k - \gamma'_k| \le \delta^{-1} (1 + 1 / \delta) \sum_{k=0}^{n-1} \epsilon_k <  \epsilon (3n)^{n} (1 + 1/ \delta)^{2n}.
\end{equation}
This estimate implies that the inverse NLFT map is locally Lipschitz.
\end{proof}

An immediate consequence of \cref{lem:inv-NLFT-modulus-continuity} is the following corollary that gives rise to a Lipschitz estimate on a subset of $\mathtt{S}_n$:
\begin{cor}
\label{cor:inv-NLFT-lipschitz}
Let $\mathtt{S}_{n,\delta} := \{(\ba,\bb) \in \mathtt{S}_n : a_0 > \delta\}$, for some $\delta > 0$. Then the inverse NLFT is a Lipschitz map on $\mathtt{S}_{n,\delta}$, with the Lipschitz constant depending on $n$ and $\delta$.
\end{cor}

A few remarks about the estimate in \cref{lem:inv-NLFT-modulus-continuity} are in order. In comparison with the proof of \cite[Lemma~3.10]{tsai2005nlft} where continuity of inverse NLFT is established without proving an explicit locally Lipschitz estimate, here we have an explicit estimate in \cref{eq:inv-NLFT-modulus-continuity-1}. However, in our case, we need $n$ to be fixed. Clearly, our proof technique cannot lead to a locally Lipschitz estimate for the non-compactly supported setting, as we cannot make the estimate independent of $n$. In fact, it is unclear whether in the non-compactly supported setting of \cite[Lemma~3.10]{tsai2005nlft}, such a locally Lipschitz estimate could hold or not.

One should also compare \cref{eq:inv-NLFT-modulus-continuity-1} with similar estimates that already exist in the literature. Two such recent results can be found in \cite{alexis2024quantum,alexis2024infinite}. The first of these results \cite[Theorem~5]{alexis2024quantum} establishes a Lipschitz estimate on the inverse NLFT when one restricts the domain of inverse NLFT to the set $\mathtt{S}_n^{\beta}:=\{(\ba,\bb) \in \mathtt{S}_n: \overbrace{\bga} = (\ba,\bb), \; \norm{\bga}_1 \le \beta = 0.36\}$. In this case, one obtains the Lipschitz estimate $\norm{\bga - \bga'}_1 \leq 2 \norm{\bb - \bb'}_1$ directly by applying the stated theorem. In contrast, our Lipschitz estimate in \cref{cor:inv-NLFT-lipschitz} does not assume such an explicit bound on the $\ell_1$ norm of $\bga$. Moreover, given $(\ba,\bb) \in \mathtt{S}_n$, the condition we have in \cref{cor:inv-NLFT-lipschitz} to test the membership of $(\ba,\bb)$ in $\mathtt{S}_{n,\delta}$ is straightforward, while checking the membership in the set $\mathtt{S}_n^\beta$ is nontrivial and involves computing the corresponding $\bga$.

Comparison of \cref{cor:inv-NLFT-lipschitz} with the second result \cite[Theorem~12]{alexis2024infinite} is also interesting. Specializing to the compactly supported setting of this paper, that result establishes a Lipschitz-type estimate on the inverse NLFT map by restricting the domain of the inverse NLFT to $\mathbf{B}_\delta$ defined as
\begin{equation}
\label{eq:Beps-def}
    \mathbf{B}_\delta := \left\{(a,b) \in \mathcal{S} : a^\ast \text{ is outer}, \; \int_{\TT} \log |a(z)| > -\delta \right\}, \; \delta > 0,
\end{equation}
with the estimate taking the following form for a constant $C_\delta$ that depends only on $\delta$:
\begin{equation}
    \max_{k \in \ZZ} |\gamma_k - \gamma'_k| \le C_\delta \norm{\frac{b}{a} - \frac{b'}{a'}}_{L^2(\TT)}, \;\; (a,b), (a',b') \in \mathbf{B}_\delta, \; \overbrace{\bga} = (a,b), \; \overbrace{\bga'} = (a',b').
\end{equation}
If we apply the condition in \cref{eq:outer-def}, we observe that the second condition defining $\mathbf{B}_\delta$  in \cref{eq:Beps-def}, assuming that $a^\ast$ is outer, simplifies to 
\begin{equation}
    \int_{\TT} \log |a^\ast(z)| = \int_{\TT} \log |a(z)| > -\delta \iff  a^\ast (0) > e^{-\delta}.
\end{equation}
This is the same condition that we assume in our definition of $\mathtt{S}_{n,\delta}$, except that we do not additionally assume that $a^\ast$ is an outer polynomial. Thus, our Lipschitz estimate in \cref{cor:inv-NLFT-lipschitz} is more general than \cite[Theorem~12]{alexis2024infinite} with respect to this outerness condition on $a^\ast$, in the compactly supported setting. However, in doing so, the price that we pay is that our Lipschitz constant ends up depending also on $n$, the size of the support of $\bga$. However, it is expected that our Lipschitz constant can be improved under additional assumptions, such as an outerness assumption on $a^\ast$.

\bibliographystyle{amsalpha}
\bibliography{ref}

\appendix
\appendixpage
\section{Preimages of the projection maps \texorpdfstring{$\pi_1$ and $\pi_2$}{}}
\label{app:complementary-poly}

This appendix is a self-contained description of the structure of the set of complementary polynomials, arising in the context of NLFT, as discussed in \cref{ssec:complementarity-prelim}. Specifically, given $a \in \mathcal{A}$ (resp. $b \in \mathcal{B}$), we seek to understand some of the structure present in the set $\{b \in \mathcal{B} : (a,b) \in \mathcal{S} \}$ (resp. $\{a \in \mathcal{A}: (a,b) \in \mathcal{S}\}$). For this, first in \cref{ssec:counting-func} below, we prove \cref{lem:a-astar-equal-b-bstar} that characterizes the properties of the zeros of the polynomial solutions to the equation $aa^\ast = bb^\ast$, for some fixed polynomial $b$. This is then used to state \cref{lem:preimage-a,lem:preimage-b} in \cref{ssec:preimage-structure}.

\subsection{Polynomial solutions to \texorpdfstring{$a a^\ast = \text{constant}$}{}}
\label{ssec:counting-func}

For the statement of \cref{lem:a-astar-equal-b-bstar}, we will need to define a counting function. Let $R = \{\alpha_j \in \CC^\ast: j=0,\dots,k-1\}$ be a multiset. We define an equivalence relation on $R$: $\alpha_j \sim \alpha_k$ if and only if $\alpha_j = 1/\overline{\alpha_k}$ or $\alpha_j = \alpha_k$, and denote the quotient set of equivalence classes as $R/\sim$ (note that $R/\sim$ is a set and not a multiset, though each equivalence class can be a multiset). Suppose that each equivalence class has an even size. If $y \in R/\sim$, we define the function $\sharp: R/\sim \;\; \rightarrow \NN$
\begin{equation}
\label{eq:sharp-def}
\sharp(y) := 
\begin{cases}
    1  &\text{if} \; \TT \cap y \neq \emptyset  \\
    1 + \frac{|y|}{2} & \text{otherwise},
\end{cases}
\end{equation}
and then define the counting function
\begin{equation}
\label{eq:counting-func-def}
\mathcal{N}(R) := 
\prod_{y \in R/\sim} \sharp(y).
\end{equation}

\begin{lemma}
\label{lem:a-astar-equal-b-bstar}
Let $a, b$ be polynomials satisfying the conditions (i) $a(0) \neq 0$, $b(0) \neq 0$, and (ii) $a(z) a^\ast(z) = b(z) b^\ast(z)$ for  infinitely many $z \in \CC^\ast$. Then $a$ and $b$ have the same degree $m$. Let $C_a$ and $C_b$ be the coefficients of the highest degree terms of $a$ and $b$, respectively. Suppose the zeros of $a$ and $b$ are given by the multisets $\mathcal{R}_a:= \{\alpha_1,\dots,\alpha_m\}$ and $\mathcal{R}_b := \{\beta_1,\dots,\beta_m\}$ respectively, and let $\mathcal{R}_{ab}$ be the multiset of common zeros of $a$ and $b$. Define the multisets $\overline{\mathcal{R}}_{a}:= \mathcal{R}_a \setminus \mathcal{R}_{ab}$ and $\overline{\mathcal{R}}_{b}:= \mathcal{R}_b \setminus \mathcal{R}_{ab}$. Then 
\begin{enumerate}[(a)]
    \item The multisets $\{\alpha_1,\dots,\alpha_m, 1/\overline{\alpha_1},\dots, 1/\overline{\alpha_m}\}$ and $\{\beta_1,\dots,\beta_m, 1/\overline{\beta_1},\dots, 1/\overline{\beta_m}\}$ are equal.
    \item $\mathcal{R}_a$, $\mathcal{R}_b$ and $\mathcal{R}_{ab}$ have the same multiplicities for each zero in $\TT$, while $\overline{\mathcal{R}}_{a}$ and $\overline{\mathcal{R}}_{b}$ do not contain any element in $\TT$.
    \item If $y \in \overline{\mathcal{R}}_{a}$ (resp. $\overline{\mathcal{R}}_{b}$) with multiplicity $k$, then it implies $y \not \in \overline{\mathcal{R}}_{b}$ (resp. $\overline{\mathcal{R}}_{a}$), $1/\overline{y} \not \in \overline{\mathcal{R}}_{a}$ (resp. $\overline{\mathcal{R}}_{b}$), and $1/\overline{y} \in \overline{\mathcal{R}}_{b}$ (resp. $\overline{\mathcal{R}}_{a}$) with multiplicity $k$.
    \item $|\overline{\mathcal{R}}_{a}| = |\overline{\mathcal{R}}_{b}|$ and $\prod_{y \in \overline{\mathcal{R}}_{a}} y = \prod_{y \in \overline{\mathcal{R}}_{b}} 1/\overline{y}$.
    \item $C_b = \lambda C_a \left| \prod_{y \in \overline{\mathcal{R}}_a} \right|$ for some $\lambda \in \TT$.
    \item Suppose we impose the additional requirement that $\mathbb{R} \ni b(0) > 0$ (resp. $< 0$). Then the number of polynomials $b$ satisfying (i), (ii) for a fixed polynomial $a$, is given by $\mathcal{N}(R)$, where $R:=\{\alpha_1,\dots,\alpha_m, 1/\overline{\alpha_1},\dots, 1/\overline{\alpha_m}\}$ and $\mathcal{N}(R)$ is defined in \cref{eq:counting-func-def}.
\end{enumerate}
\end{lemma}

\begin{proof}
Suppose $a(z) := C_a \prod_{j=1}^{m}(z - \alpha_j)$ and  $b(z) := C_b \prod_{j=1}^{n}(z - \beta_j)$, for some $C_a, C_b \in \CC^\ast$. Condition (i) ensures that none of the $\alpha_j$ or $\beta_j$ is zero. This implies  that we have  
\begin{equation}
\begin{split}
    a(z)a^\ast(z) &= |C_a|^2 z^{-m} \prod_{j=1}^m (z - \alpha_j)(1 - \overline{\alpha_j}z),\; b(z)b^\ast(z) = |C_b|^2 z^{-n} \prod_{j=1}^n (z - \beta_j)(1 - \overline{\beta_j}z), \\
\end{split}   
\end{equation}
defined for all $z \in \CC^\ast$. By condition (ii), the two polynomials $|C_a|^2 z^n \prod_{j=1}^m (z - \alpha_j)(1 - \overline{\alpha_j}z)$ and $|C_b|^2 z^{m} \prod_{j=1}^n (z - \beta_j)(1 - \overline{\beta_j}z)$ agree on an infinite set, and hence they must be the same. In particular, these two polynomials must have the same set of zeros and identical multiplicities for each zero. This implies $m=n$, since none of the $\alpha_j$ or $\beta_j$ are zero, and establishes equality of the multisets $\{\alpha_1,\dots,\alpha_m, 1/\overline{\alpha_1},\dots, 1/\overline{\alpha_m}\}$ and $\{\beta_1,\dots,\beta_m, 1/\overline{\beta_1},\dots, 1/\overline{\beta_m}\}$, thus proving the equality of degrees of $a$ and $b$, and part (a). 

Parts (b) and (c) immediately follow from part (a), while part (d) follows from part (c).

For part (e), we note that $|C_a|^2 \prod_{j=1}^m \overline{\alpha_j} = |C_b|^2 \prod_{j=1}^m \overline{\beta_j}$ by equality of the polynomials above. After taking complex conjugate on both sides, this implies $|C_a|^2 \left( \prod_{y \in \mathcal{R}_{ab}} y \right) \left( \prod_{y \in \overline{\mathcal{R}}_{a}} y \right) = |C_b|^2 \left( \prod_{y \in \mathcal{R}_{ab}} y \right) \left( \prod_{y \in \overline{\mathcal{R}}_{b}} y \right)$, and now the conclusion follows as a consequence of part (d).

For part (f), we first note that the number of ways to choose the multiset $\mathcal{R}_b$ such that part (a) holds is given exactly by $\mathcal{N}(R)$. Now, once $\mathcal{R}_b$ is chosen, the condition $b(0) > 0$ (resp. $< 0$) implies that there is a unique choice of $\lambda$ in part (e) satisfying all the assumptions, finishing the proof.
\end{proof}

\subsection{Structure of the fibers}
\label{ssec:preimage-structure}

We can use \cref{lem:a-astar-equal-b-bstar} to characterize the fibers of the projection maps $\pi_1$ and $\pi_2$, which we defined in \cref{ssec:complementarity-prelim}. In order to state these results, it will be useful to keep in mind the following property that is a consequence of \cref{lem:a-astar-equal-b-bstar}(a): if $a \in \mathcal{A}$ (resp. $b \in \mathcal{B}$), then the zeros of $1 - aa^\ast$ (resp. $1 - bb^\ast$) can be paired up so that every pair of zeros $\{\alpha,\beta\}$ satisfies $\alpha = 1/ \overline{\beta}$. In particular, if $y \in \TT$ is a zero, then it has even multiplicity.

\begin{lemma}
\label{lem:preimage-b}
Let $b \in \mathcal{B}$ and let $R$ be the multiset of zeros of $1 - b b^\ast$ in $\CC^\ast$. Then
\begin{enumerate}[(a)]
    \item The preimage $\pi_2^{-1}(b) \subseteq \mathcal{S}$ is finite, and $|\pi_2^{-1}(b)| = \mathcal{N}(R)$, with $\mathcal{N}(R)$ defined in \cref{eq:counting-func-def}. 
    \item Let $(a,b), (a',b) \in \pi_2^{-1}(b)$. Then $a \neq a'$ if and only if the multisets of zeros of $a^\ast$ and $(a')^\ast$ are not equal. Moreover, $a$ and $a'$ share the same zeros on $\TT$ counted with multiplicity.
    \item If $a$ is a polynomial such that $a(0) \neq 0$, let $R_a$ denote the multiset of its zeros. Then if $(a^\ast,b) \in \pi_2^{-1}(b)$ it implies $R = \cup_{\alpha \in R_a} \{\alpha, 1/\overline{\alpha}\}$ as multisets. Conversely, if $R = \cup_{\alpha \in R_a} \{\alpha, 1/\overline{\alpha}\}$ as multisets, then there exists an unique $\lambda \in \CC^\ast$ such that $(\lambda a^\ast,b) \in \pi_2^{-1}(b)$.
\end{enumerate}
\end{lemma}

\begin{lemma}
\label{lem:preimage-a}
Let $a \in \mathcal{A}$ and let $R$ be the multiset of zeros of $1 - a a^\ast$ in $\CC^\ast$. Then
\begin{enumerate}[(a)]
    \item The preimage $\pi_1^{-1}(a) \subseteq \mathcal{S}$ is a disjoint union of $\mathcal{N}(R)$ components, with $\mathcal{N}(R)$ defined in \cref{eq:counting-func-def}, where each component is of the form $\{(a, \lambda z^k b): k \in \ZZ, \lambda \in \TT\}$, for some fixed $(a,b)$ in that component. Thus, each component is in one-to-one correspondence with $\ZZ \times \TT$.
    \item Let $(a,b), (a,b') \in \pi_1^{-1}(a)$. Then $b$ and $b'$ are not in the same component as defined in part (a) if and only if the multisets of zeros in $\CC^\ast$ of $b$ and $b'$ are not equal. Moreover, $b$ and $b'$ share the same zeros on $\TT$ counted with multiplicity.
    \item If $b$ is a Laurent polynomial, let $R_b$ denote the multiset of its zeros in $\CC^\ast$. Then if $(a,b) \in \pi_1^{-1}(a)$ it implies $R = \cup_{\alpha \in R_b} \{\alpha, 1/\overline{\alpha}\}$ as multisets. Conversely, if $R = \cup_{\alpha \in R_b} \{\alpha, 1/\overline{\alpha}\}$ as multisets, then there exists $\lambda \in \CC^\ast$ such that $(a,\lambda b) \in \pi_2^{-1}(a)$.
\end{enumerate}
\end{lemma}
The proofs of both \cref{lem:preimage-a,lem:preimage-b} follow from \cref{lem:a-astar-equal-b-bstar}, and the properties imposed on $a$ and $b$ by \cref{lem:nlft-bijection}. The only differences between the two lemmas appear because we must have $a^\ast(0) > 0$ and $a^\ast$ must be a polynomial, but this need not be true for $b$. The detailed proofs are skipped.

A consequence of \cref{lem:preimage-b} is that given $b \in \mathcal{B}$, there is an unique $a \in \mathcal{A}$ such that $(a,b) \in \mathcal{S}$ and $a^\ast$ is an outer polynomial (recall that an outer polynomial has no zeros in $\DD$), as already established in \cite{tsai2005nlft,alexis2024infinite}. In fact, the roles of $a$ and $b$ can be interchanged to obtain a similar statement as a corollary of \cref{lem:preimage-a}: if $a \in \mathcal{A}$, then the set of all outer $b \in \mathcal{B}$ such that $(a,b) \in \mathcal{S}$, is of the form $\{\lambda b': \lambda \in \TT\}$, for any fixed choice of an outer $b'$ satisfying $(a,b') \in \mathcal{S}$ (such a $b'$ always exists). This happens because if the lowest degree of a Laurent polynomial $b$ is non-zero, then either $b(0)=0$ or $b \not \in H^\infty(\DD)$, where $H^\infty(\DD)$ is set of all bounded holomorphic functions on $\DD$, and both these conditions contradict the requirements for a function to be an outer function (see \cref{sec:outer-functions}).

\section{Outer functions}
\label{sec:outer-functions}

In \cref{ssec:complementarity-prelim} we introduced the concept of outer polynomials, which we defined to be polynomials with no zeros in the unit disk $\DD$. Here, we give a definition of outer functions that is also applicable more generally to holomorphic functions on $\DD$ that are not polynomials. There are several equivalent characterizations, and we will state one of the simplest ones based on \cite[Chapter~17]{rudin1987real}. Let us denote the set of bounded holomorphic functions on the unit disk by $H^\infty(\DD)$. If $f \in H^\infty(\DD)$ and $f(0) \neq 0$, then we say that $f$ is an \textit{outer function} if and only if 
\begin{equation}
\label{eq:outer-def}
    \log |f(0)| = \frac{1}{2\pi} \int_{0}^{2 \pi} \log |\hat{f}(e^{i\theta})| \; d\theta,
\end{equation}
where $\hat{f} (e^{i \theta})$ denotes the radial limit $\lim_{r \rightarrow 1^{-}} f(r e^{i \theta})$. Note that by \cite[Theorem~17.17]{rudin1987real}, $\hat{f}$ exists almost everywhere on $\TT$ and the integral on the right-hand side of \cref{eq:outer-def} is finite, and thus the condition makes sense under our assumptions. We may check that if a polynomial $f$ has no zeros in $\DD$ then it satisfies \cref{eq:outer-def} with $\hat{f} = f$ on $\TT$, while if $f$ has zeros in $\DD$ (still with $f(0) \ne 0$), then one must have $\log |f(0)| < \frac{1}{2\pi} \int_{0}^{2 \pi} \log |f(e^{i\theta})| \; d\theta$, both being consequences of Jensen's formula \cite[Theorem~15.18]{rudin1987real}. For polynomials, these facts can also be derived avoiding Jensen's formula entirely, because of the much simpler known result in complex integration theory \cite[Chapter~3, Exercise~11]{stein2010complex}:
\begin{equation}
    \frac{1}{2 \pi} \int_{0}^{2 \pi} \log |1 - \omega e^{i \theta}| \; d \theta = 
    \begin{cases}
        0 & \text{ if } |\omega| \le 1, \\
        \log |\omega| & \text{ if } |\omega| > 1.
    \end{cases}
\end{equation}

\section{Helper lemmas used in \texorpdfstring{\cref{sec:lipschitz-bounds-nlft}}{}}
\label{sec:helper-lemmas}

\begin{lemma}
\label{lem:b-over-a-bound}
Let $a, b, a', b' \in \CC$, with $|a - a'| \le \epsilon$ and $|b - b'| \le \epsilon$. Suppose there exist $\delta, \delta'$ such that $0 < \delta \le \min \{|a|, |a'|\}$ and $\max \{|b|, |b'| \} \le \delta'$. Then we have 
\begin{equation}
 \left| \frac{b}{a} - \frac{b'}{a'} \right| \leq \epsilon \left( 1 + \max \{|b/a|, |b'/a'|\} \right) / \min \{|a|, |a'|\} \le \epsilon \delta^{-1} ( 1 + \delta' / \delta).
\end{equation}
\end{lemma}

\begin{proof}
We have 
\begin{equation}
\begin{split}
    & |b/a - b'/a'| \leq |(b- b')/ a| + |b'| \; |1/a - 1/a'| \le \epsilon \left( 1 + |b'|/|a'| \right) / |a| \\
    & \le \epsilon \left( 1 + \max \{|b/a|, |b'/a'|\} \right) / \min \{|a|, |a'|\} \le \epsilon \delta^{-1} ( 1 + \delta' / \delta).
\end{split}
\end{equation}
\end{proof}

\begin{lemma}
\label{lem:bound-T-diff}
Suppose $\gamma_k \in \CC$ and we define $\Theta(\gamma_k) :=  \frac{1}{\sqrt{1 + |\gamma_k|^2}} \left( \begin{smallmatrix}
    1 & -\gamma_k \\ \overline{\gamma_k} & 1
\end{smallmatrix} \right)$ for $k=1,2$. Then the map $\CC \ni \gamma \rightarrow \Theta(\gamma) \in \mathrm{SU}(2)$ is Lipschitz continuous satisfying $\norm{\Theta(\gamma_1) - \Theta(\gamma_2)}_{2} \leq \sqrt{10} \; |\gamma_1 - \gamma_2|$. Additionally, we also have $\norm{\Theta(\gamma_1) - \Theta(\gamma_2)}_{1} < 3 \; |\gamma_1 - \gamma_2|$.
\end{lemma}

\begin{proof}
Without loss of generality, we may assume that $0 \le |\gamma_1| \leq |\gamma_2|$. Note that by the mean value theorem, for some $\lambda \in [|\gamma_1|, |\gamma_2|]$, we have
\begin{equation}
    \frac{1}{\sqrt{1 + |\gamma_1|^2}} - \frac{1}{\sqrt{1 + |\gamma_2|^2}} = (|\gamma_2| - |\gamma_1|) \lambda (1 + \lambda^2)^{-\frac{3}{2}} < |\gamma_1 - \gamma_2|. 
\end{equation}
This, in turn, allows us to estimate
\begin{equation}
\begin{split}
    & \left| \frac{\gamma_1}{\sqrt{1 + |\gamma_1|^2}} - \frac{\gamma_2}{\sqrt{1 + |\gamma_2|^2}} \right| = \left| \frac{\gamma_1 - \gamma_2}{\sqrt{1 + |\gamma_2|^2}} + \frac{\gamma_1}{\sqrt{1 + |\gamma_1|^2}} - \frac{\gamma_1}{\sqrt{1 + |\gamma_2|^2}} \right| \\
    & \leq \frac{|\gamma_1 - \gamma_2|}{\sqrt{1 + |\gamma_2|^2}} + |\gamma_1| \; |\gamma_1 - \gamma_2| \; \frac{\lambda}{(1 + \lambda^2)^{\frac{3}{2}}} \leq \frac{|\gamma_1 - \gamma_2|}{(1 + \lambda^2)^{\frac{1}{2}}} + |\gamma_1 - \gamma_2| \; \frac{\lambda^2}{ (1 + \lambda^2)^{\frac{3}{2}}} \\
    & \le 2 \; |\gamma_1 - \gamma_2| \; (1 + |\lambda|^2)^{-\frac{1}{2}} < 2 \, |\gamma_1 - \gamma_2|.
\end{split}
\end{equation}
This gives the bound $\norm{\Theta(\gamma_1) - \Theta(\gamma_2)}_{1} < 3 |\gamma_1 - \gamma_2|$. Additionally, we have $\norm{\Theta(\gamma_1) - \Theta(\gamma_2)}_{2}^2 \leq \norm{\Theta(\gamma_1) - \Theta(\gamma_2)}_{F}^2 < 10 \; |\gamma_1 - \gamma_2|^2$.
\end{proof}

\end{document}